\documentclass[10pt,a4paper]{article}

\usepackage{a4wide}

\usepackage[dvipsnames]{xcolor}

\usepackage{paralist}

\title{A decision procedure for\\ intuitionistic modal logic IS4 (and IK4)}
\author{Marianna Girlando, Roman Kuznets, Sonia Marin, 
  Lutz Stra\ss burger}

\usepackage{amsmath} 
\usepackage{amssymb} 
\usepackage{amsthm}
\usepackage{xspace}
\usepackage{pdfpages}
\usepackage{cmll} 

\usepackage{latexsym} 
\usepackage{colonequals} 
\usepackage{wasysym}
\usepackage{standalone}
\usepackage{virginialake}
\vlnosmallleftlabels

\usepackage{graphicx} 
\usepackage{tikz}
\usetikzlibrary{decorations.pathreplacing,patterns}
\usepackage{tikz-qtree}
\usepackage[shortlabels]{enumitem}
\usepackage{multirow}

\usepackage{thmtools} 
\usepackage{thm-restate}

\usepackage[hidelinks]{hyperref} 

\newdimen\mydisplayskip
\newenvironment{smallequation*}
{\par\nobreak\vskip\mydisplayskip\noindent\bgroup\small\csname equation*\endcsname}{\csname endequation*\endcsname\egroup}

\theoremstyle{plain}
\newtheorem{theorem}{Theorem}[section]
\newtheorem{proposition}[theorem]{Proposition}
\newtheorem{lemma}[theorem]{Lemma}
\newtheorem{corollary}[theorem]{Corollary}
\theoremstyle{definition}
\newtheorem{observation}[theorem]{Observation}
\newtheorem{notation}[theorem]{Notation}
\newtheorem{definition}[theorem]{Definition}
\newtheorem{construction}[theorem]{Construction}
\newtheorem{algorithm}[theorem]{Algorithm}
\newtheorem{remark}[theorem]{Remark}
\newtheorem{example}[theorem]{Example}

\DeclareFontFamily{U} {MnSymbolC}{}
\DeclareFontShape{U}{MnSymbolC}{m}{n}{
	<-6>  MnSymbolC5
	<6-7>  MnSymbolC6
	<7-8>  MnSymbolC7
	<8-9>  MnSymbolC8
	<9-10> MnSymbolC9
	<10-12> MnSymbolC10
	<12->   MnSymbolC12}{}
\DeclareFontShape{U}{MnSymbolC}{b}{n}{
	<-6>  MnSymbolC-Bold5
	<6-7>  MnSymbolC-Bold6
	<7-8>  MnSymbolC-Bold7
	<8-9>  MnSymbolC-Bold8
	<9-10> MnSymbolC-Bold9
	<10-12> MnSymbolC-Bold10
	<12->   MnSymbolC-Bold12}{}
\DeclareSymbolFont{MnSyC}         {U}  {MnSymbolC}{m}{n}
\DeclareMathSymbol{\diamondplus}{\mathbin}{MnSyC}{124}
\DeclareMathSymbol{\boxtimes}{\mathbin}{MnSyC}{117}

\vlnostructuressyntax
\newcommand{\vlderivationauxnc}[1]{#1{\box\derboxone}\vlderivationterm}
\newcommand{\vlderivationnc}{\vlderivationinit\vlderivationauxnc}
\makeatletter
\newbox\@conclbox
\newdimen\@conclheight
\newcommand*{\A}{\mathcal{A}}
\newcommand{\quand}{\quad\mbox{and}\quad}
\newcommand{\qquand}{\qquad\mbox{and}\qquad}

\newcommand{\proviso}[1]{\mbox{\scriptsize #1}}
\newcommand{\defn}[1]{{\textit{\textbf{#1}}\/}}

\newcommand{\lefmark}{\bullet}
\newcommand{\rigmark}{\circ}
\def\lef#1{#1^\bullet}
\def\rig#1{#1^\circ}

\def\wbox{\BOX}
\def\wdia{\DIA}

\newcommand{\w}{\rn{weak}}

\newcommand{\weakr}{\rn{weak}}
\newcommand{\necr}{\rn{nec}}
\newcommand{\mpr}{\rn{mp}}

\newcommand{\monl}{\lef{\rn{mon}}}
\newcommand{\idr}{\rn{id}}

\newcommand*{\ax}[1]{\mathsf{#1}}
\newcommand*{\kax}[1][]{\ax{k_{#1}}}

\newcommand{\vax}{\ax{4}}

\newcommand{\tax}{\ax{t}}
\newcommand*{\IK}{\mathsf{IK}}
\newcommand*{\K}{\mathsf{K}}
\newcommand*{\IKfour}{\mathsf{IK4}}

\newcommand*{\ISfour}{\mathsf{IS4}}

\newcommand*{\Sfour}{\mathsf{S4}}

\newcommand*{\NOT}{\neg}
\newcommand*{\AND}{\mathbin{\wedge}}

\newcommand*{\cand}{\mathbin{\wedge}}
\newcommand*{\TOP}{\mathord{\top}}
\newcommand*{\OR}{\mathbin{\vee}}

\newcommand*{\BOT}{\mathord{\bot}}
\newcommand*{\IMP}{\mathbin{\supset}}

\newcommand*{\BOX}{\mathord{\Box}}

\newcommand*{\DIA}{\mathord{\Diamond}}

\newcommand{\lseq}[3]{#1 , #2 \SEQ #3}
\newcommand{\B}{\mathcal{R}}
\newcommand{\Left}{\Gamma} 
\newcommand{\Right}{\Delta} 
\newcommand*{\fm}[1]{{\color{notgreen}#1}}
\newcommand*{\lb}[1]{{\color{blue}#1}}
\newcommand*{\lbG}[2]{{\color{blue}#1^{\sq{#2}}}}

\newcommand*{\lbc}[1]{{\color{blue}#1}}

\newcommand{\fmb}[1]{\fm{#1^{{\color{black}\bullet}}}}
\newcommand{\fmw}[1]{\fm{#1^{{\color{black}\circ}}}}
\newcommand*{\rel}{R}
\newcommand{\rrel}{R}
\newcommand{\lrel}[1]{R^{\leftrightarrow}_{\sq{#1}}}

\newcommand{\grel}[1]{R_{\sq #1}}
\newcommand*{\labels}[2]{\lb{#1}\mathord{:}\fm{#2}}
\newcommand*{\labelsb}[2]{\lb{#1}\mathord{:}\fmb{#2}}
\newcommand*{\labelsw}[2]{\lb{#1}\mathord{:}\fmw{#2}}
\newcommand*{\accs}[2]{\lb{#1}R\lb{#2}}

\newcommand{\sle}[1]{\mathbin{\le_{\sq{#1}}}}
\newcommand{\sleq}[1]{\mathbin{=_{\sq{#1}}}}
\newcommand{\nsle}[1]{\mathbin{\not\le_{\sq{#1}}}}
\newcommand{\srel}[1]{\mathbin{\rel_{\sq{#1}}}}
\newcommand{\srelx}[1]{\mathbin{\rel_{#1}}}
\newcommand{\svis}[1]{\mathbin{\Rsh_{\!\sq{#1}}}}
\newcommand{\svisx}[1]{\mathbin{\Rsh_{\!{#1}}}}
\newcommand*{\accsq}[2]{\lb{#1}\srel{G}\lb{#2}}
\newcommand{\accsqq}[3][G]{\lb{#2}\srel{#1}\lb{#3}}
\newcommand{\accsqx}[3][G]{\lb{#2}\srelx{#1}\lb{#3}}

\newcommand*{\accsqh}[2]{\lb{#1}\mathord{R^{\leftrightarrow}_{\sq{G}}}\lb{#2}}
\newcommand*{\accsqc}[2]{\lbc{#1}\mathord{\srel{G}}\lbc{#2}}

\newcommand*{\accsqp}[3]{\lb{#1}\srel{#3}\lb{#2}}

\newcommand*{\futsq}[2]{\lb{#1}\sle{G}{\lb{#2}}}
\newcommand*{\futsqq}[3][G]{\lb{#2}\sle{#1}\lb{#3}}
\newcommand*{\vissq}[2]{\lb{#1}\svis{G}{\lb{#2}}}

\newcommand*{\vissqx}[3][G]{\lb{#2}\svisx{#1}{\lb{#3}}}
\newcommand*{\futsqc}[2]{\lbc{#1}\mathord{\sle{G}}{\lbc{#2}}}

\newcommand*{\futsqp}[3]{\lb{#1}\mathord{\sle{#3}}{\lb{#2}}}
\newcommand*{\nfutsq}[2]{\lb{#1}\nsle{G}{\lb{#2}}}

\newcommand*{\futs}[2]{\lb{#1}\mathord\le{\lb{#2}}}

\newcommand{\labelsof}[1]{\ell(\sq{#1})}
\newcommand{\toplabelsof}[1]{\ell^\top(\sq{#1})}
\newcommand{\labelsofx}[1]{\ell({#1})}

\newcommand{\emb}{{\color{black}\mathsf{e}}}

\newcommand{\embh}{{\color{black}\hat{\mathsf{e}}}}
\newcommand{\embof}[1]{{\color{black}\emb(}\lb{#1}\textcolor{black}{)}}
\newcommand{\embofh}[1]{{\color{black}\hat\emb(}\lb{#1}\textcolor{black}{)}}

\newcommand{\embofp}[1]{{\color{black}\emb'(}\lb{#1}\textcolor{black}{)}}

\newcommand{\embofhn}[2]{{\color{black}\embh_{#1}}(\lb{#2})}

\newcommand{\SEQ}{\Longrightarrow}

\newcommand*{\rn}[1]  {\ensuremath{\mathsf{#1}}}
\newcommand*{\lab}{\mathsf{lab}}
\newcommand*{\rr}{\mathsf{r}}
\newcommand*{\lr}{\mathsf{p}}

\newcommand*{\labrn}[2][]  {\rn{#2}_{#1}}

\newcommand{\M}{\mathfrak{M}}

\newcommand{\F}{\mathfrak{F}}
\newcommand{\force}[3]{#1,\lb{#2}\Vdash\fm{#3}}
\newcommand{\nforce}[3]{#1,\lb{#2}\not\Vdash\fm{#3}}

\newcommand{\der}{\delta}
\newcommand{\hder}{\hat\delta}
\newcommand{\deri}{\delta}
\newcommand{\derib}{\vartheta}
\newcommand{\deric}{\xi}

\newcommand{\hderi}{\hat\delta}
\newcommand{\hderib}{\hat\vartheta}

\newcommand{\height}[1]{|#1|}
\newcommand{\mset}[1]{\left\{\mkern-5mu\left|#1\right|\mkern-5mu\right\}}

\newcommand{\Rref}{\rel\rn{rf}}

\newcommand{\Rtr}{\rel\rn{tr}}

\newcommand{\Lref}{\mathord\le\rn{rf}}

\newcommand{\Ltr}{\mathord\le\rn{tr}}

\def\gG{\mathcal{G}}

\newcommand{\fone}{\rn{fc}}

\newcommand{\ftwo}{\rn{bc}}

\newcommand{\fhat}{{\rn{fbc}}}

\def\tuple#1{\langle#1\rangle}

\newcommand{\treeof}[1]{\mathsf{T}(\lb {#1})}

\newcommand{\lbeq}{\mathord\sim}

\newcommand{\seqsetredsymb}{\mathbin{\lower.25ex\rlap{$\rightsquigarrow$}\raise.25ex\hbox{$\rightsquigarrow$}}}

\newcommand{\hv}{\hat v}
\newcommand{\hx}{\hat x}
\newcommand{\hy}{\hat y}
\newcommand{\hz}{\hat z}

\newcommand{\hG}{{\hat G}}

\newcommand{\hL}{{\hat L}}

\newcommand{\hw}{{\hat w}}

\newcommand{\ssatof}[1]{\lceil{\sq{#1}}\rceil}
\newcommand{\satof}[1]{\left\lceil{#1}\right\rceil}
\newcommand{\liftssatof}[1]{{\sq{#1}}^{\Uparrow}}
\newcommand{\liftsatof}[1]{{#1}^{\Uparrow}}

\newcommand{\Nat}{\mathbb{N}}

\newcommand{\set}[1]{\{#1\}}

\newcommand{\sq}[1]{\mathcal{#1}}

\newcommand{\sqset}[1]{\rlap{$\mathcal{#1}$}\mkern0.75mu\mathcal{#1}}

\makeatletter
\newcommand*{\overRightarrow}{\mathpalette{\overarrow@\Rightarrowfill@}}
\newcommand*{\overLeftarrow}{\mathpalette{\overarrow@\Leftarrowfill@}}
\makeatother

\newcommand{\lliftf}[3]{\sq{#1}\mathord\uparrow^{\labels{#2}{#3}}}

\newcommand{\layerof}[1]{L_{\lb{#1}}}
\newcommand{\layerxof}[2]{L_{\lb{#1}}^{\sq{#2}}}
\newcommand{\layerx}[2]{L_{#1}^{\sq{#2}}}
\newcommand{\layerxxof}[2]{L_{\lb{#1}}^{#2}}
\newcommand{\clusterof}[1]{\lbc{C_{\lb{#1}}}}

\newcommand{\futuresof}[1]{\lb{#1\mathord{\uparrow}}} 

\newcommand{\sizeof}[1]{\left|#1\right|}

\newcommand{\modelof}[1]{\M_{\sq{#1}}}
\newcommand{\mmodelof}[1]{\M_{#1}}

\newcommand{\blackf}[3]{\labels{\!#1}{#2}^{\bullet}\in\sq{#3}}
\newcommand{\whitef}[3]{\labels{\!#1}{#2}^{\circ}\in\sq{#3}}

\newcommand{\blackff}[2]{\labels{\!#1}{#2}^{\bullet}}
\newcommand{\whiteff}[2]{\labels{\!#1}{#2}^{\circ}}
\newcommand{\simul}{\rn{S}}

\newcommand*{\labISf}{\lab\ISfour}
\newcommand*{\labIKf}{\lab\IKfour}
\newcommand{\labISfs}{\labISf^{\mathsf{s}}}
\newcommand{\labIKfs}{\labIKf^{\mathsf{s}}}
\newcommand{\labISloop}{\labISf^{\circlearrowleft}}
\newcommand{\labISloops}{\labISf^{\circlearrowleft}_{\mathsf{s}}}
\newcommand{\labISloopp}{\labISf^{\circlearrowleft}_{\mathsf{e}}}
\newcommand{\labISnoloop}{\labISf^{\rightrightarrows}}

\newcommand{\sysS}{\mathsf{S}}
\newcommand{\id}{\labrn{id}}

\newcommand{\tr}[1]{\mathcal{#1}}

\newcommand{\shrinkr}{\mathsf{shrink}}
\newcommand{\satr}{\mathsf{s}\text{-}\mathsf{sat}}
\newcommand{\ssatr}{\mathsf{s}\text{-}\mathsf{sat}^\ast}
\newcommand{\bdiam}{\lef\DIA_\mathsf{m}}
\def\rigsym{\circ}
\newcommand{\lliftfw}[3]{\sq#1\mathord\uparrow^{\labelsw{#2}{#3}}}
\newcommand{\liftrim}{\rig\IMP_\mathsf{m}}
\newcommand{\liftrbox}{\rig\BOX_\mathsf{m}}
\newcommand{\liftr}{\rig{\mathsf{r}}_\mathsf{m}}

\newcommand{\eqlab}[2]{\lb{#1} \lbeq \lb{#2}}

\newcommand{\LEAF}{\mbox{\scriptsize\textleaf}}
\newcommand{\TRI}{\triangledown}

\newcommand{\labelsubst}[2]{ [\lb{#1} / \lb{#2}]}
\newcommand{\loopr}{\mathsf{loop}}
\newcommand{\looprule}{\mathsf{loop}}
\newcommand{\dialoopr}{\mathsf{loop}_{\DIA}}
\newcommand{\boxloopr}{\mathsf{loop}_{\BOX}}

\newcommand{\Xloopr}{\mathsf{loop}_{\mathsf X}}
\newcommand{\Rloopr}{\mathsf{loop}_{\mathsf R}}
\newcommand{\Sloopr}{\mathsf{loop}_{\mathsf S}}
\newcommand{\imploopr}{\mathsf{loop}_{\IMP}}
\newcommand{\leafloopr}{\mathsf{loop}_{\LEAF}}

\newcommand{\triangleloopr}{\mathsf{loop}_{\TRI}}

\newcommand{\lford}{\prec}
\newcommand{\lord}[2]{\lb{#1}\prec\lb{#2}}

\newcommand{\cF}{\mathcal{F}}

\newcommand{\unhappyof}[1]{U(\sq{#1})}
\newcommand{\liftof}[1]{\sq{#1}\mathord{\uparrow}\mkern-3.5mu\mathord{\uparrow}}

\newcommand{\interval}[2]{\left[\lb{#1}..\lb{#2}\right)}

\newcommand{\trans}[1]{#1^\circlearrowleft}

\newcommand{\clsubsymb}{\Subset}

\newcommand{\clsubG}[4]{\lbc{#1^{\sq{#2}}}\clsubsymb\lbc{#3^{\sq{#4}}}}

\newcommand{\lseqto}{\mathbin{\stackrel{\Uparrow}{\longrightarrow}}}
\newcommand{\sseqto}{\mathbin{\stackrel{\satof{\cdot}}{\longrightarrow}}}

\newcommand{\evils}{\sigma}

\newcommand{\ldepthof}[1]{\mathsf{d}(\lb{#1})}
\newcommand{\ldepthi}[2][i]{\mathsf{d}_{#1}(#2)}
\newcommand{\lreldepth}[2]{\partial(\lb{#1},\lb{#2})}
\newcommand{\lreldepthij}[3]{\partial_{#1}^{#2}(#3)}

\newcommand{\lreldepthi}[2][i]{\hat\partial_{#1}(#2)}

\newcommand{\subf}[1]{\mathsf{s}(\fm{#1})}

\newcommand{\depthf}[1]{\mathsf{D}(\fm{#1})}
\newcommand{\heightf}[1]{\mathsf{H}(\fm{#1})}
\newcommand{\maxlayerf}[1]{\mathsf{L}(\fm{#1})}

\newcommand{\rank}[1]{\mathrm{rk}(#1)}

\newcommand{\sr}{\mathsf{s}}

\newcommand{\below}{\mathbin{\ll}}
\newcommand{\sabove}{\mathbin{\gg}}
\newcommand{\psbelow}{\below}

\newcommand{\sdominv}{\leftharpoonup}

\newcommand{\topeq}{\coh}

\newcommand{\kurucz}{Kurucz's formula\xspace}

\newcommand{\satofgen}{\lceil\cdot\rceil}
\newcommand{\liftssatofgen}{{\Uparrow}}

\newcommand{\numberof}[1]{\##1}

\definecolor{notgreen}{rgb}{.1,.6,.1}
\definecolor{lutzgreen}{rgb}{.2,.5,.3}

\begin{document}
	
\maketitle

\tableofcontents
\section{Introduction}
\label{sec:intro}
%
In this paper, we show that the two intuitionistic modal logics $\ISfour$~and~$\IKfour$ are decidable.
We provide a constructive decision procedure, that, given a formula, produces either  a proof showing the formula to be valid or a finite countermodel falsifying the formula. This means, we also show the \emph{finite model property} for both logics.

Our procedure is based on a \emph{fully labelled proof system}  presented in~\cite{mar:mor:str:2021}.
This system inherits the advantages of labelled systems for intuitionistic propositional logic and for classical modal logics, in particular, all inference rules are invertible 
and there is a direct correspondence between sequents and models.
This choice was crucial to the emergence of the solution, which also required an intricate organisation of  proof search and loop-checking to reach a full decision procedure. 

A preliminary version of this work has been presented at the LICS 2023 conference~\cite{girlando2023intuitionistic}. However, the proof search algorithm described in that work is incorrect. More precisely, it does not terminate because of an insufficient loop condition. The mistake has been found by Agi Kurucz\footnote{Personal communication.}, and we discuss her counterexample later in Section~\ref{sec:agi-ex}. 

The main purpose of this work is to provide a correction to~\cite{girlando2023intuitionistic}, which is the reason that this introduction is quite short. For a more detailed discussion on related work, we refer the reader to~\cite{girlando2023intuitionistic}.

\bigskip

\textbf{Organisation.}
In Section~\ref{sec:prelims} we first introduce the logic $\ISfour$ and explain why proving decidability turned out to be so difficult. Then we recall the fully labelled proof system for $\ISfour$ that will be the basis for our proof search algorithm. Its key property is that all rules are invertible. Additionally, we will show how these sequents relate to the birelational models of $\ISfour$, and we introduce some structural properties of our sequents that we use in the course of this paper. 

In Section~\ref{sec:macro-rules} we give all the details of our proof search algorithm. There are three main ingredients: First, we introduce macro-rules that are derivable in our fully labelled proof system and that keep the structure of the sequents in a certain nice way. Second, the proof search alternates between two phases that we call \emph{saturation} and \emph{lifting}. And third, we introduce \emph{unsound} $\loopr$-rules, that allow the continuation of the proof search after a loop has been detected. We exemplify the proof search algorithm with Kurucz's formula in Section~\ref{sec:agi-ex}. 

In Section~\ref{sec:termination}, we prove that our algorithm always terminates.
In Section~\ref{sec:countermodel}, we demonstrate  how to retrieve a countermodel from a failed proof search, and by this showing completeness of our algorithm.
In Section~\ref{sec:soundness} we show soundness by eliminating all $\loopr$-rule instances from the derivation constructed by our algorithm.

Finally, in Section~\ref{sec:IK4}, we explain the necessary modification needed to adapt our algorithm to the case of the modal logic $\IKfour$.

\bigskip

\textbf{Acknowledgements.} We are deeply grateful to Agi Kurucz for finding the counterexample and communicating it to us. We would like to thank Ian Shillito for pointing out several errors in the first version of this draft. We also thank the reviewers of the conference paper for their helpful comments and suggestions. 
Finally, we wish to thank Anupam Das, David Fern\'andez Duque, Marianela Morales, Nicola Olivetti, Elaine Pimentel, Revantha Ramanayake, Alex Simpson, and Yde Venema  for insightful discussions. \\
Roman Kuznets was supported by ERDF-Project Knowledge in the Age of Distrust, project  No.~CZ.02.01.01/00/23$\_$025/0008711. 

\bigskip 

\textbf{Usage of AI.} No AI tools were used in the preparation of the paper, neither in assistance with the proofs nor in editing the text or the figures. 

\section{Preliminaries}
\label{sec:prelims}
	\subsection{What is $\ISfour$?}
	\label{sec:IS4}
	
	In this paper, formulas are denoted 
	by capital letters~$\fm A$, $\fm B$, $\fm C$,~\ldots\ and are constructed from a countable set~$\mathcal{A}$ of \emph{atomic propositions} (denoted by lowercase~$\fm a$, $\fm b$, $\fm c$,~\ldots) as
	\[
	\scalebox{1}{$
		\fm A \coloncolonequals
		\fm \BOT \mid \fm a \mid \fm{(A \AND A)} \mid \fm{(A \OR A)} \mid  \fm{(A \IMP A)} \mid \fm{\BOX A} \mid \fm{\DIA A}$}
	\]
	
	Intuitionistic modal logic~$\K$, or $\IK$~for short, is obtained from
	intuitionistic propositional logic by extending the syntax with 
	two modalities~$\wbox$~and~$\wdia$, standing most generally for \emph{necessity} and \emph{possibility} respectively, and by adding \emph{$\kax{}$-axioms}
	 \begin{equation}
		 	\label{eq:kax}
		 	\begin{array}{rc}
			 		\kax[1]\colon&\fm{\BOX(A\IMP B)\IMP(\BOX A\IMP\BOX B)}\\
			 		\kax[2]\colon&\fm{\BOX(A\IMP B)\IMP(\DIA A\IMP\DIA B)}\\
			 		\kax[3]\colon&\fm{\DIA(A\OR B)\IMP(\DIA A\OR\DIA B)}\\
			 		\kax[4]\colon&\fm{(\DIA A\IMP \BOX B)\IMP\BOX(A\IMP B)}\\
			 		\kax[5]\colon&\fm{\DIA\BOT\IMP\BOT}\\
			 	\end{array}
		 \end{equation}

	A formula is a theorem of~$\IK$ if{f} it is derivable from this set via the rules of \emph{necessitation} and \emph{modus ponens}:
	\[
	\vlinf{\necr}{}{\fm{\wbox A}}{\fm A}
	\qquand
	\vliiinf{\mpr}{}{\fm B}{\fm A}{}{\fm{A\IMP B}}
	\quad
	\]
	Intuitionistic modal logic~$\Sfour$, or $\ISfour$~for short, is obtained from~$\IK$ by
	adding  axioms
	\begin{equation}
		\label{eq:vax}
		\begin{array}{r@{\;}l}
			\vax\colon&
			\fm{(\wdia\wdia A\IMP\wdia A)\cand(\wbox A\IMP\wbox\wbox A)}
			\\ 
			\tax\colon&
			\fm{(A\IMP\wdia A)\cand(\wbox A\IMP A)}
		\end{array}
	\end{equation}
	
	Note that in the classical case  axioms~$\kax[2]$--$\kax[5]$
	in~\eqref{eq:kax} would follow from~$\kax[1]$, but due to the lack of
	De~Morgan duality, this is not the case in intuitionistic
	logic. Similarly, in~\eqref{eq:vax} both conjuncts are
	needed because they do not follow from each other as in the classical
	case.
	
	Let us now recall
	the \emph{birelational models}~\cite{plotkin:stirling:86,ewald:86} for
	intuitionistic modal logics, which combine the Kripke
	semantics for intuitionistic propositional logic and
	classical modal logics.
	
	\begin{definition}
		A \defn{birelational frame}~$\F$ is a triple~$\langle W, \rel, \le \rangle$
		of a nonempty set~$W$ of \defn{worlds}  equipped with an \defn{accessibility relation}~$\rel$ and a preorder~$\le$ (i.e.,~a reflexive and transitive relation) satisfying:
		\begin{enumerate}
			\item[($\fone$)] \emph{Forward confluence}: For all~$\lb x, \lb y, \lb z \in W$, if $\futs xz$ and $\accs xy$, there exists~$\lb u \in W$ such that $\accs zu$ and $\futs yu$ (see figure below right);
			\item[($\ftwo$)] \emph{Backward confluence}: For all~$\lb x, \lb y, \lb z \in W$, if $\accs xy$ and $\futs yz$, there exists~$\lb u \in W$ such that $\futs xu$ and $\accs uz$ (see figure below left).
		\end{enumerate}	
		
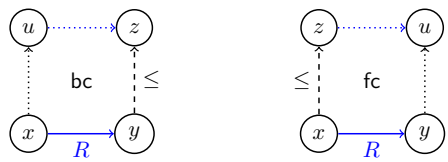
\begin{center}
				
	\begin{tikzpicture}[thick, every node/.style={scale=1.2}]

		\tikzstyle{node}=[circle,draw]
		\tikzstyle{s-node}=[circle,draw,inner sep=2pt]
		\tikzstyle{nonode}=[inner sep=0pt]
		\tikzstyle{access}=[->,blue]
		\tikzstyle{future}=[->,dashed]
		\tikzstyle{deleted}=[->,gray]
		\tikzstyle{loop}=[->,red]
		
		\node[] (3) at (5,3) [] {$\ftwo$};
		
		\node[label=below :{}] (3) at (4,2) [node, 
		] {$x$};
		\node[label=below :{}] (4) at (6,2) [node,
		] {$y$};
		\node[label=above :{}] (10) at (4,4) [node, 
		] {$u$};
		\node[label=above :{}] (11) at (6,4) [node, 
		] {$z$};
		
		\draw[access] (3) edge[below] node {$R$} (4);
		\draw[access,dotted] (10) -- (11);

		\draw[future,dotted] (3) -- (10);
		\draw[future] (4) edge[right] node {$\le$}  (11);
	\end{tikzpicture}

\hspace*{2cm}

\begin{tikzpicture}[thick, every node/.style={scale=1.2}]

	\tikzstyle{node}=[circle,draw]
	\tikzstyle{s-node}=[circle,draw,inner sep=2pt]
	\tikzstyle{nonode}=[inner sep=0pt]
	\tikzstyle{access}=[->,blue]
	\tikzstyle{future}=[->,dashed]
	\tikzstyle{deleted}=[->,gray]
	\tikzstyle{loop}=[->,red]
	
	\node[] (3) at (5,3) [] {$\fone$};
	
	\node[label=below :{}] (3) at (4,2) [node, 
	] {$x$};
	\node[label=below :{}] (4) at (6,2) [node, 
	] {$y$};
	\node[label=above :{}] (10) at (4,4) [node,
	] {$z$};
	\node[label=above :{}] (11) at (6,4) [node, 
	] {$u$};
	
	\draw[access] (3) edge[below] node {$R$} (4);
	\draw[access,dotted] (10) -- (11);
	
	\draw[future] (3) edge[left] node {$\le$} (10);
	\draw[future,dotted] (4) -- (11);
\end{tikzpicture}
		\end{center}
	\end{definition}
	
	\begin{definition}
		\label{model}
		A \defn{birelational model}~$\M$ is a quadruple~$\langle W, \rel,\le,V \rangle$ with $\langle W, \rel, \le \rangle$ a birelational frame and $V\colon W \to 2^\mathcal{A}$ a \defn{valuation function}, that is, a function mapping each world~$\lb w$ to the subset of propositional atoms that are true at~$\lb w$, additionally subject to the \defn{monotonicity condition}:
		if $\futs w{w'}$, then $V(\lb w)\subseteq V(\lb{w'})$.
		
		We write $\force\M wa$ if{f} $\fm a \in V(\lb w)$ and recursively extend the relation~$\Vdash$ to all formulas following the rules for both intuitionistic and modal Kripke models: 
		\[
		\begin{array}{@{\;\;}l@{\;}c@{\;\;}l}
			\nforce\M w\BOT; & & \\
			\force\M w{A \AND B} & \mbox{if{f}} & \force\M wA \mbox{ and } \force\M wB;\\
			
			\force\M w{A \OR B} & \mbox{if{f}} & \force\M wA \mbox{ or } \force\M wB;\\
			
			\force\M w{A \IMP B} & \mbox{if{f}} & \mbox{for all } \lb{w'} \mbox{ with } \futs w{w'},\mbox{if $\force\M{w'}A$, then $\force\M{w'}B$};\\
			
			\force\M w{\BOX A} & \mbox{if{f}} & \mbox{for all } \lb{w'} \mbox{ and } \lb u \mbox{ with } \futs w{w'} \mbox{and } \accs{w'}u, \mbox{ we have } \force\M uA; \hfill \\ 
			
			\force\M w{\DIA A} & \mbox{if{f}} & \mbox{there exists } \lb u \mbox{ such that } \accs wu \mbox{ and } \force\M uA.
		\end{array}
		\]
	\end{definition}
	
	From the monotonicity of valuation function~$V$, we get the monotonicity property for relation~$\Vdash$:
	
	\begin{proposition}[Monotonicity] 
		For any formula~$\fm A$ and for any~$\lb w, \lb{w'} \in W$, if $\futs w{w'}$ and $\force\M wA$, then $\force\M{w'}A$.
	\end{proposition}
	
	\begin{definition}[Validity]
		A formula~$\fm A$ is \defn{valid in a model}~$\M = \langle W, \rel, \le, V \rangle$ if{f}  $\force\M wA$ for all~$\lb w \in W$.
		A formula~$\fm A$ is \defn{valid in a frame}~$\F = \langle W, \rel, \le \rangle$ if{f} it is valid in~$\langle W, R, \le, V \rangle$ for all valuations~$V$.
	\end{definition}
	
	The correspondence between syntax and semantics for~$\ISfour$ can be stated as follows:
	
	\begin{theorem}[Completeness~\cite{fischer-servi:84,plotkin:stirling:86}]\label{thm:plotkin}
		A formula~$\fm A$ is a theorem of\/~$\ISfour$ if and only if $\fm A$~is valid in every birelational frame~$\langle W, \rel, \le \rangle$ where $\rel$~is reflexive and transitive.
	\end{theorem}
	
	Unless stated otherwise, in the remainder of the article,  relation $R$ is assumed reflexive and transitive  in all considered  birelational frames and models.
	
	\subsection{Why proving decidability of $\ISfour$ is difficult?}
	\label{sec:difficulty}
	
	In this section we highlight the main difficulties that we encountered in tackling the decidability problem for  $\ISfour$ and give a hint at the key ingredients of our method in the process.
	
	One way to prove decidability for a logic is to perform proof search in a sound and complete deductive system with the intention of either finding a proof or constructing a countermodel from a failed proof search. %
	
	For~$\ISfour$ several such deductive systems exist, the first being Simpson's labelled systems~\cite{simpson:phd}. 
	Moreover, there are
	various kinds of nested sequents systems: single-conclusion~\cite{str:fossacs13}, multiple conclusion~\cite{kuz:str:2017maehara}, and also focused variants~\cite{cha:mar:str:fscd16}.
	A natural question to ask is why none of these systems has been used to prove decidability of~$\ISfour$.  
	
	\textbf{The need for more labelling.}
	The aforementioned systems rely on what we could call a mixed approach: they internalise the modal accessibility relation $R$ within the sequent syntax, using either  labels and relational atoms or nesting, but they rely on a traditional structural approach for the intuitionistic aspect of the logic, e.g., single-conclusion sequents (at least in certain rules). 
	One might think that combining the traditional loop-check for the intuitionistic part with the label-based loop-check for the modal part would be a way to a decision procedure.
	But the situation is more complicated because
	\begin{enumerate}
		\item 
		the classical $\Sfour$-loop test, which is looking for a
		repetition along the $\rel$-relation, cannot be applied to the right-hand-side of a sequent, as the conclusion formula can sometimes  be replaced by a new one;
		
		\item the structural approach to the intuitionistic system also means the rules are not all invertible and the procedure requires backtracking, so the
		modality loop-check also needs to be combined with the necessary backtracking.
	\end{enumerate}
	
	Both of these problems can be overcome by using a fully labelled proof system that incorporates both relations~$\rel$~and~$\le$~\cite{mar:mor:str:2021, maffezioli:naibo:negri:13}.
	This has the same advantages as moving from a structural to a labelled approach for intuitionistic propositional logic, mentioned in the introduction.
	Not only does this system  re-establish the close relationship between a sequent and a model, as is known from classical modal logic, it also enables us to make all rules in the system invertible.
	Moreover, explicit relational atoms in the sequent syntax  make it easy to implement the loop-checks and to represent the back edges explicitly when constructing a countermodel.\looseness=-1

	\textbf{The backtracking/termination trade-off.}
	Naive proof search is not terminating, with two possible
	sources of non-termination: the first inherited from the classical
	modal logic~$\Sfour$, and the second from intuitionistic
	propositional logic. 
	In both those cases independently, the problem would be solved by a
	simple loop-check. However, it is not straightforward to combine the two, as the following example shows:
	
	\begin{example}
		Consider the following formula, which is not provable in~$\ISfour$:
		\begin{equation}
			\label{eq:back-example}
			\fm{\wbox\Bigl((\wbox
				a\IMP\BOT)\cand\bigl((a\IMP\BOT)\IMP\BOT\bigr)\Bigr)\IMP
				\BOT}
			\text{.}
		\end{equation}
		Let $\fm A=\fm{\wbox a\IMP\BOT}$ and $\fm B=\fm{(a\IMP\BOT)\IMP\BOT}$.  In
		order to construct a countermodel~$\M$, we need a world~$\lb {w_1}$ that forces
		$\fm{\wbox(A\cand B)}$ and, therefore, also  $\fm A$~and~$\fm B$.
		Consequently, every world~$\lb{w'}$ such that $\futs {w_1}{w}$ and $\accs w{w'}$ for some~$\lb w$ should force these formulas. 
		Then all these worlds must
		force neither $\fm{\wbox a}$ nor $\fm{a\IMP\BOT}$. The latter means that for
		each such world~$\lb{w'}$, there must be a world~$\lb v$ with $\futs{w'}v$ and
		$\force\M va$. Of course, in turn, $\lb v$~must not force~$\fm{\wbox a}$, so there must be  worlds~$\lb u$ and $\lb{u'}$ with $\futs v{u'}$, $\accs{u'}u$, and
		$\nforce\M ua$. But this~$\lb u$ must not force~$\fm{a\IMP\BOT}$ so there must be a world~$\lb{v_1}$ with $\futs u{v_1}$ and $\force\M{v_1}a$,  and so on.  Thus, a naive implementation of a
		countermodel construction via proof search will keep adding worlds
		\emph{ad~infinitum} because neither of the two loop-checks will detect the repetition (see Fig.~\ref{fig:inf-model}, Left).
		\begin{figure}
			\begin{center}
					\begin{tikzpicture}[thick, every node/.style={scale=1.2}]
		
		\tikzstyle{node}=[circle,draw]
		\tikzstyle{s-node}=[circle,draw,inner sep=2pt]
		\tikzstyle{nonode}=[inner sep=0pt]
		\tikzstyle{access}=[->,blue]
		\tikzstyle{future}=[dashed,->]
		\tikzstyle{suricata}=[dashed,red]
		\tikzstyle{deleted}=[->,gray]
		\tikzstyle{loop}=[->,red]		
		
		\node[label={[align=center] right :{$\Vdash \BOX(A \AND B), \Vdash A, \Vdash B$\\$\not\Vdash \BOX a, \not\Vdash a \IMP\BOT$}}] (1) at (2,0) [node] {$w_1$}; 
		
		\node[] (2) at (2,2) [node] {$w_2$};
		\node[label= right :{$\not\Vdash a$}] (3) at (4,2) [node, fill=green!30] {$w_3$};
		%
		
		\node[label= right :{$\Vdash a$}] (4) at (4,4) [node, fill=purple!30] {$w_4$};
		\node[] (5) at (2,4) [node] {$w_5$};
		
		\node[label= right:{}] (6) at (4,6) [node] {$w_6$};
		\node[label= right :{$\not\Vdash a$}] (7) at (6,6) [node,fill=green!30] {$w_7$};
		\node[] (8) at (2,6) [node] {$w_8$};
		
		\node[label= right :{$\Vdash a$}] (9) at (6,8) [node,fill=purple!30] {$w_9$};
		\node[label= right :{}] (10) at (4,8) [s-node] {$w_{10}$};
		\node[] (11) at (2,8) [s-node] {$w_{11}$};

		\draw[access] (2) -- (3);
		
		\draw[access] (5) -- (4);
		
		\draw[access] (8) -- (6);
		\draw[access] (6) -- (7);
		
		\draw[access] (11) -- (10);
		\draw[access] (10) -- (9);
		
		\draw[future] (1) -- (2);
		\draw[future] (2) -- (5);
		\draw[future] (5) -- (8);
		\draw[future] (8) -- (11);
		\draw[future] (11) -- (2,9);
		\draw[future] (3) -- (4);
		\draw[future] (4) -- (6);
		\draw[future] (6) -- (10);
		\draw[future] (10) -- (4,9);
		\draw[future] (7) -- (9);
		\draw[future] (9) -- (6,9);

	\end{tikzpicture}
\hspace*{2cm}
\begin{tikzpicture}[thick, every node/.style={scale=1.2}]

\tikzstyle{node}=[circle,draw]
\tikzstyle{s-node}=[circle,draw,inner sep=2pt]
\tikzstyle{nonode}=[inner sep=0pt]
\tikzstyle{access}=[->,blue]
\tikzstyle{future}=[dashed,->]
\tikzstyle{suricata}=[dashed,red]
\tikzstyle{deleted}=[->,gray]
\tikzstyle{loop}=[->,red]		

\node[] (1) at (2,0) [node] {$w_1$}; 

\node[] (2) at (2,2) [node] {$w_2$};
\node[label= right :{$\not\Vdash a$}] (3) at (4,2) [node, fill=green!30] {$w_3$};
%

\node[label= right :{$\Vdash a$}] (4) at (4,4) [node, fill=purple!30] {$w_4$};
\node[] (5) at (2,4) [node] {$w_5$};

\node[label= right:{}] (6) at (4,6) [node] {$w_6$};
\node[] (8) at (2,6) [node] {$w_8$};

\node[label= right :{}] (10) at (4,8) [s-node] {$w_{10}$};
\node[] (11) at (2,8) [s-node] {$w_{11}$};

\draw[access] (2) -- (3);

\draw[access] (5) -- (4);

\draw[access] (8) -- (6);
\draw[access] (6) edge[bend right] (3);

\draw[access] (11) -- (10);
\draw[access] (10) edge[bend left] (4);

\draw[future] (1) -- (2);
\draw[future] (2) -- (5);
\draw[future] (5) -- (8);
\draw[future] (8) -- (11);
%
\draw[future] (3) -- (4);
\draw[future] (4) -- (6);
\draw[future] (6) -- (10);
%
%

\end{tikzpicture}
	
			\end{center}
			\caption{Left: Illustration of the potential non-termination issue. \quad Right: Illustration of a break of condition $\ftwo$ when identifying nodes unrestrictedly.}
			\label{fig:inf-model}
		\end{figure}
		In other words, 
		\begin{enumerate}\setcounter{enumi}{2}
			\item to account for the interaction of modalities and intuitionistic implications, the two loop-checks must get along well with each other. 
		\end{enumerate}
		
	\end{example}

\textbf{The interleaving of proof search and loops.}
For solving this third problem, we have to implement a more sophisticated loop-check involving both relations.
This directly leads us to the fourth problem.

\begin{example}
	Assume, for the sake of example, that we designed a suitable loop-check such that we could stop proof search at the stage of the structure presented on Fig.~\ref{fig:inf-model} (Left) and that we identified $\lb{w_7}$ to $\lb{w_3}$ and $\lb{w_{9}}$ to $\lb{w_4}$. 
	This would create ``backlinks'' between $\lb{w_{10}}$ and $\lb{w_4}$ as well as between~$\lb{w_6}$~and~$\lb{w_3}$.
	However, this would lead to a violation of~$\ftwo$  as now $\lb{w_{10}} R \lb{w_4} \mathord{\le} \lb{w_6}$ but there is no $\lb{w'}$ such that $\lb{w_{10}} \mathord{\le} \lb{w'} R \lb{w_6}$. (see Fig.~\ref{fig:inf-model}, Right).
	This means that %
	\begin{enumerate}\setcounter{enumi}{3}
		\item the standard method of constructing a (finite) countermodel from a failed proof search via identifying labels/worlds that create a loop fails in the setting of birelational models. We break the $\fone$/$\ftwo$ properties, which in turn would force us to add new worlds, which would mean we have to continue proof search.
	\end{enumerate}
\end{example}

We solve this problem by identifying (substituting) labels not only after the completion of but also during   proof search.  
This preserves unprovability, but could, \emph{a priori}, be unsound. 
This means that, when terminating a branch on a non-axiom\-at\-ic sequent, it is still possible to extract a countermodel from it.
However, when reaching only axiomatic leaves, it remains to be shown that a sound proof can be obtained from the proof attempt (potentially containing identification of labels).
So, instead of doing naive proof search and then constructing a countermodel from a failed proof search by ``folding'' the failed sequent, we perform the  folding already during the proof search, which is now a countermodel search, and then construct a proper proof from a failed countermodel search by ``unfolding'' the search tree.
The loop-check ensuring termination has to be subtly calibrated for this final step of  unfolding the proof attempt into a real proof.

	\subsection{The proof system}
	\label{sec:proof:system}

	We now present the fully labelled sequent calculus $\labISfs$ for $\ISfour$ that forms the basis of our decision algorithm.
	It is an equivalent formulation (together with some notational variations) of the labelled sequent calculus $\labISf$ for $\ISfour$ presented in~\cite{mar:mor:str:2021}.
	
	To define a labelled proof system, we first enrich the language of $\ISfour$ by a countable set of \defn{labels}, denoted by $\lb x, \lb y, \lb z,$ \textit{etc.} 
	Then, we define \defn{relational atoms} as expression $\accs xy$ or $\futs xy$ where $\lb x$~and~$\lb y$~are labels, and \defn{labelled formulas} as pairs~$\labels xA$ of a label~$\lb x$ and a formula~$\fm A$.  
	The labelled calculus introduced by Simpson in~\cite{simpson:phd} employs only one kind of  relational atoms, $\accs xy$. We here follow the fully labelled approach from \cite{mar:mor:str:2021}, which instead employs relational atoms in correspondence to both relations of birelational frames. 
	
	In the literature, a \defn{labelled sequent} is usually defined as a triple~$\lseq\B\Left\Right$
	where $\B$~is a set of relational atoms and $\Left$~and~$\Right$~are multisets (or sets) of labelled formulas respectively, all written as comma-separated lists.  
	To simplify subsequent definitions and proofs, we employ a different notation for our sequents.  First, we assign a polarity $\lefmark$ (\defn{input}) or $\rigmark$~(\defn{output}) to each labelled formula, which will be written as $\labelsb xA$ or $\labelsw xA$, respectively.
        In this case, we will also say that the labelled formula $\labelsb xA$ or $\labelsw xA$ \defn{occurs} at label $\lb x$.
	Then, we define a \defn{polarised labelled sequent}~$\sq G$ as a \emph{set} of relational atoms and  labelled formulas with polarities; we write $\sq G, \sq G'$ for $\sq G \cup \sq G '$, and we write $\sq G, F$ for $\sq G \cup \set{F}$, where $F$ is a relational atom or a labelled formula with polarity. In the following, we simply write \defn{sequent} to mean \emph{polarised labelled sequent}.
	
	\begin{example} 
		\label{ex:pol_sq}
		$ \sq G  = \futs x y, \accs y z , \accs y u, \labelsb{x}{a\AND b}, \labelsw{y}{c}, \labelsw{z}{\BOX a}$ is such a sequent.
	\end{example}
	
	Intuitively, the polarities play the role of the sequent arrow $\SEQ$, allowing us to identify the left- and right-hand side of a sequent. The input formulas, or $^\lefmark$-formulas, are those that would occur in the antecedent of a labelled sequent, and the output formulas, or  $^\rigmark$-formulas, are those occurring in the succedent.\footnote{We do not assign polarities to relational atoms because, as in the case of standard labelled sequents, it is sufficient to restrict their use to  the antecedent.}
	Thus, every sequent can be written as a labelled sequent in the usual notation,  and vice versa. 
	\begin{example}
		The sequent $\sq G$ from Example~\ref{ex:pol_sq} can be written as labelled sequent in the usual notation (as in~\cite{mar:mor:str:2021}) as $\futs x y, \accs y z , \accs y u, \labels{x}{a\AND b} \SEQ  \labels{y}{c}, \labels{z}{\BOX a}$. 
	\end{example}

	\begin{figure*}[!t]
		\def\myskip{3ex}
		\scalebox{1}{\hbox{
				\begin{minipage}{.97\textwidth}
					\begin{center}
						\begin{tabular}{c c}
							$\vlinf{\labrn{id}}{}{\sq G,  \labelsb xa , \labelsw xa }{}$
							&
							$\vlinf{\lef\BOT}{}{\sq G, \labelsb x\BOT }{}$\\[\myskip]
							$ 	\vlinf{\lef\AND}{}{\sq G ,\labelsb x{A \AND B}}{\sq G,\labelsb x{A \AND B}, \labelsb xA, \labelsb xB }$
							&
							$\vliinf{\rig\AND}{}{\sq G, \labelsw x{A \AND B}}{\sq G,\labelsw x{A \AND B}, \labelsw xA}{\sq G, \labelsw x{A \AND B}, \labelsw xB} $\\[\myskip]
							$\vliinf{\lef\OR}{}{\sq G, \labelsb x{A \OR B}}{\sq G, \labelsb x{A \OR B}, \labelsb xA }{\sq G, \labelsb x{A \OR B}, \labelsb xB }$
							&
							$	\vlinf{\rig\OR}{}{\sq G, \labelsw x{A \OR B}}{\sq G,  \labelsw x{A \OR B},  \labelsw xA, \labelsw xB}$ \\[\myskip]
							$ 	\vliinf{\lef\IMP}{}{\sq G,   \labelsb x{A \IMP B}}{\sq G,  \labelsb x{A \IMP B},  \labelsw xA}{\sq G,  \labelsb x{A \IMP B}, \labelsb xB}$
							&
							$\vlinf{\rig\IMP}{\proviso{$\lb z$ fresh}}{\sq G, \labelsw x{A \IMP B}}{\sq G, \futs xz, \labelsb zA , \labelsw x{A \IMP B}, \labelsw zB} $
							\\[\myskip]
							$ 	\vlinf{\lef\BOX}{}{\sq G, \accs xy, \labelsb x{\BOX A} }{\sq G, \accs xy,  \labelsb x{\BOX A}, \labelsb yA }$ 
							&
							$\vlinf{\rig\BOX}{\proviso{$\lb u,\lb z$ fresh}}{\sq G, \labelsw x{\BOX A}}{\sq G, \futs xu, \accs uz, \labelsw x{\BOX A},  \labelsw zA} $
							\\[\myskip]
							$\vlinf{\lef\DIA}{\proviso{$\lb y$  fresh}}{\sq G, \labelsb x{\DIA A} }{\sq G, \accs xy, \labelsb x{\DIA A}, \labelsb yA }$
							&
							$		\vlinf{\rig\DIA}{}{\sq G ,\accs xy,  \labelsw x{\DIA A}}{\sq G, \accs xy,  \labelsw x{\DIA A}, \labelsw yA}$ \\ 
							\multicolumn{2}{l}{\hbox to .97\linewidth{\dotfill}}\\ 
							\multicolumn{2}{c}{
								$		
								\vlinf{\monl}{}{\sq G, \futs xy, \labelsb xA}{\sq G, \futs xy, \labelsb xA, \labelsb yA}
								$
								\qquad 
								$ 	\vlinf{\Lref}{}{\sq G}{\sq G, \futs xx}$
								\qquad 
								$\vlinf{\Ltr}{}{\sq G, \futs xy, \futs yz, }{\sq G, \futs xy, \futs yz, \futs xz}
								$
							}
							\\[\myskip]
							$  \vlinf{\fone}{\proviso{$\lb u$ fresh}}{\sq G, \accs xy, \futs xz}{\sq G, \accs xy, \futs xz, \futs yu, \accs zu} $
							&
							$ \vlinf{\ftwo}{\proviso{$\lb u$ fresh}
							}{\sq G, \accs xy, \futs yz }{\sq G, \accs xy, \futs yz, \futs xu, \accs uz}$
							\\[\myskip]
							 $ 	\vlinf{\Rref}{}{\sq G}{\sq G ,\accs xx}$
							&
							$	 \vlinf{\Rtr}{}{\sq G, \accs xy, \accs yz}{\sq G, \accs xy, \accs yz, \accs {x}z} $\\[0.5cm]
						\end{tabular}
					\end{center}
				\end{minipage}
			}		
		}
		\caption{System~$\labISfs$ 
		}
		\label{fig:labIKp}
	\end{figure*}
	
	The rules of $\labISfs$ are displayed in Figure~\ref{fig:labIKp}.\footnote{The superscript $\mathsf{s}$ stands for \emph{sets}, as our calculus works on sets, unlike the original $\labISf$~\cite{mar:mor:str:2021}, which works on multisets.}  The rules in the upper part of the Figure are \defn{logical rules}, and those in the lower part of the figure are \defn{structural rules}. 
	Observe that, since sequents are defined as sets, contraction is embedded into the system. We have chosen a cumulative version of the rules, with the principal formula repeated in the premise(s), as this will become useful in the definition of the proof-search algorithm. Finally, thanks to the presence of $\monl$, an explicit structural rule for monotonicity, rules $\id$, $\lef\IMP$, and $\lef\BOX$ 
	are the same labelled rules we would use in a calculus for classical modal logic (refer, e.g., to~\cite{negri:2005}). 
	
	\begin{definition}
		A \defn{derivation tree} (or \defn{derivation} for short) over a set $\sysS$ of inference rules is a tree whose nodes are labelled with sequents and whose inner nodes are additionally labelled by inference rules, such that  
                whenever a node is labelled with $\sq G$ and its children with $\sq G_1, \dots, \sq G_k$, respectively, where $k\geq 0$, then that node is also labelled by an inference rule $\rr$, such that $\vliiinf{\rr}{}{\sq G}{\sq G_1}{\ldots}{\sq G_k}$ is a correct instance of the  inference rule $\rr\in\sysS$. 
		A \defn{proof} is a derivation where each leaf is the conclusion of $\id$ or $\lef\BOT$ (and labelled by $\id$ or $\lef\BOT$, respectively).
		The \defn{height} of a derivation $\deri$ is the length of one of its longest branches, denoted by $\sizeof{\deri}$.
        \end{definition}

	An example of a proof in $\labISfs$ can be found in Figure~\ref{fig:ex-deriv}.
	Let us now introduce some standard terminology. A rule $\rr$ is \defn{admissible} if{f}, whenever there are proofs of its premise(s), 
	there is a proof of its conclusion. 
	In case the height of the derivation is preserved, i.e., if whenever there are proofs of the premise(s) of $\rr$ whose height is bounded by $n$, there is a proof  of the conclusion of $\rr$ whose height is also bounded by $n$, then we say  that $\rr$ is \defn{height-preserving admissible} (short: \defn{hp-admissible}). 
	Rule $\rr$  is \defn{derivable} whenever there is a derivation of its conclusion whose leaves 
	consist of its premise(s). 
	Rule $\rr$ is \defn{invertible} if{f}, whenever there is a proof of its conclusion, then there are proofs of all its premises. 
	
	\begin{lemma}
          The rule $\vlinf{\w}{}{\sq G, \sq H}{\sq G}$  is hp-admissible in $\labISfs$. 
	\end{lemma}

	\begin{proof}
		By routine induction on the height of the derivation of the premise of $\w$.	\end{proof}
	


	\newcommand{\idm}{\mathsf{id}_{\ast}}
	\newcommand{\lefim}{\lef{\IMP}_{\ast}}
	\newcommand{\lefbm}{\lef{\BOX}_{\ast}}
	
	Admissibility of cut in $\labISfs$ follows (indirectly) from soundness and completeness of $\labISfs$, which we establish by showing how to translate every proof in $\labISfs$ into a proof in $\labISf$ and back. System $\labISf$ is the fully labelled sequent calculus for  $\ISfour$ introduced in \cite{mar:mor:str:2021}, which employs
	labelled sequents $\lseq\B\Left\Right$ as discussed above in their multiset formulation. The rules of $\labISf$ are the same as the rules of $\labISfs$, except that they are not cumulative, rule $\monl$ is not primitive (it is admissible, see below), and rules $\id$, $\lef\IMP$, and $\lef\BOX$ are formulated as follows (employing our notation, and using an index $\ast$ to distinguish them from our rules in Figure~\ref{fig:labIKp}):
	\vspace{-0.5cm}
	\[
	\vlinf{\idm}{}{\sq G,  \futs x y, \labelsb xa , \labelsw ya }{}
	\]
	\vspace{-0.2cm}
	\[
	\vliinf{\lefim}{}{\sq G, \futs x y,  \labelsb x{A \IMP B}}{\sq G, \futs x y, \labelsb x{A \IMP B},  \labelsw yA}{\sq G,  \futs x y, \labelsb yB}
	\qquad\quad 
	\vlinf{\lefbm}{}{\sq G, \futs xy ,\accs y z, \labelsb x{\BOX A} }{\sq G, \futs xy ,\accs y z,  \labelsb x{\BOX A}, \labelsb zA }
	\]

\begin{figure}[!t]
	\scalebox{.9}{
		$
		\vlderivationnc{
			\vlin{\rig\IMP  }{}{\labelsw{1}{\BOX (\DIA A \AND \DIA b) \IMP \BOT}}{
				\vlin{\Rref+\lef\BOX+\lef\AND}{}{\futs{1}{2},\labelsw{1}{\BOX (\DIA A \AND \DIA b) \IMP \BOT},\labelsb{2}{\Box (\DIA A \AND \DIA b) },\labelsw{2}{\BOT}}{
					\vlin{\lef\DIA}{}{\futs{1}{2},\accs{2}{2},\labelsw{1}{\BOX (\DIA A \AND \DIA b) \IMP \BOT},\labelsb{2}{\Box (\DIA A \AND \DIA b) },\labelsb{2}{\DIA A \AND \DIA b }, \labelsb{2}{\Diamond A}, \labelsb{2}{\DIA b},\labelsw{2}{\BOT}}{
						\vliin{\Lref+\lef\IMP}{}{\futs{1}{2},\accs{2}{2}, \accs{2}{3},\Gamma,\labelsb{3}{A}}{
							\vlin{\rig\IMP}{}{\futs{1}{2},\accs{2}{2}, \accs{2}{3},\futs{3}{3},\Gamma,\labelsb{3}{A},\labelsw{3}{c \IMP \DIA b}}{
								\vlin{\lef\BOX+\lef\AND}{}{\futs{1}{2},\accs{2}{2}, \accs{2}{3},\futs{3}{3},\futs{3}{5},\Gamma,\labelsb{3}{A},\labelsw{3}{c \IMP \DIA b},\labelsb{5}{c},\labelsw{5}{\DIA b}}{
									\vlin{\lef\DIA}{}{\futs{1}{2},\accs{2}{2}, \accs{2}{3},\futs{3}{3},\futs{3}{5},\Gamma,\labelsb{3}{\DIA A \AND \DIA b}, \labelsb{3}{\DIA A}, \labelsb{3}{\DIA b},\labelsb{3}{A}, \labelsw{3}{c \IMP \DIA b},\labelsb{5}{c},\labelsw{5}{\DIA b}}{
										\vlin{\fone+\lef{\rn{mon}}}{}{\futs{1}{2},\accs{2}{2}, \accs{2}{3},\futs{3}{3},\futs{3}{5},\accs{3}{4},\Gamma,\Delta, \labelsb{4}{b},\labelsw{5}{\DIA b}}{
											\vlin{\rig\DIA}{}{\futs{1}{2},\accs{2}{2}, \accs{2}{3},\futs{3}{3},\futs{3}{5},\accs{3}{4}, \futs{4}{6}, \accs{5}{6},\Gamma,\Delta,\labelsb{4}{b}, \labelsb{6}{b},\labelsw{5}{\DIA b}}{
												\vlin{\idr}{}{\futs{1}{2},\accs{2}{2}, \accs{2}{3},\futs{3}{3},\futs{3}{5},\accs{3}{4}, \futs{4}{6}, \accs{5}{6},\Gamma,\Delta, \labelsb{4}{b}, \labelsb{6}{b},\labelsw{5}{\DIA b}, \labelsw{6}{b}}{
													\vlhy{}
												}
											}
										}
									}
								}
							}
						}{
							\vlin{\hspace*{-1.5cm}\lef\BOT}{}{\sq G}{\vlhy{}}
						}
					}
				}
			}
		}
		$
	}
	
	\caption{Proof in $\labISfs$ of  $\fm{\BOX (\DIA A \AND \DIA b) \IMP \BOT}$, where we set
	$ \fm A = \fm{(c \IMP \DIA b ) \IMP \BOT} $, and \\
	$\Gamma = \{\labelsw{1}{\BOX (\DIA A \AND \DIA b) \IMP \BOT}, \labelsb{2}{\Box (\DIA A \AND \DIA b) },\labelsb{2}{\DIA A \AND \DIA b }, \labelsb{2}{\DIA A}, \labelsb{2}{\DIA b}, \labelsw{2}{\BOT}\}$, \\
	$\Delta = \{\labelsb{3}{\DIA A \AND \DIA b }, \labelsb{3}{\DIA A}, \labelsb{3}{\DIA b},\labelsb{3}{A},\labelsw{3}{c \IMP \DIA b},\labelsb{5}{c}\}$, 
	and\\
	$ \sq G = \futs{1}{2},\accs{2}{2}, \accs{2}{3}, \futs{3}{3},\Gamma,\labelsb{3}{A},\labelsb{3}{\BOT}$. 
	}
	\label{fig:ex-deriv}
\end{figure}
        
	The following results, proved in~\cite{mar:mor:str:2021}, establish soundness and completeness of $\labISf$ with respect to the Hilbert-style axiomatisation for the logic.

	\begin{theorem}[\cite{mar:mor:str:2021}]
		\label{thm:labIK}
		A formula~$\fm A$~is a theorem of\/~$\ISfour$ if{f} for every $\lb x$, the sequent  $\labelsw xA$ has a proof in\/~$\labISf$.
	\end{theorem}

\begin{remark}
  To be precise, the system for $\ISfour$ in~\cite{mar:mor:str:2021} employs the following four rules
  \begin{equation}
    \label{eq:g-rules}
    \begin{array}{c}
\vlinf{\sf g_{2001}}{\proviso{$\lb z'$ fresh}}{\sq G,  \accs x y, \accs y z}{\sq G,  \accs x y, \accs y z, \futs z {z'}, \accs x {z'}}
\qquad
\vlinf{\sf g_{0120}}{\proviso{$\lb x'$ fresh}}{\sq G,  \accs x y, \accs y z}{\sq G,  \accs x y, \accs y z, \futs x {x'}, \accs {x'} z}
\\\\[-2ex]
\vlinf{\sf g_{0001}}{\proviso{$\lb x'$ fresh}}{\sq G}{\sq G,   \futs x {x'}, \accs x {x'}}
\qquad
\vlinf{\sf g_{0100}}{\proviso{$\lb x'$ fresh}}{\sq G}{\sq G,   \futs x {x'}, \accs {x'}x} 
    \end{array}
  \end{equation}
(corresponding to the frame conditions for the axioms $\sf g_{2001}\colon\fm{\DIA\DIA A\IMP\DIA A}$ and $\sf g_{0120}\colon\fm{\BOX A\IMP\BOX\BOX A}$ and $\sf g_{0001}\colon\fm{A\IMP\DIA A}$ and $\sf g_{0100}\colon\fm{\BOX A\IMP A}$), instead of $\Rtr$ and $\Rref$, corresponding to the transitivity and reflexivity frame conditions. However, it is easy to see that all proofs in~\cite{mar:mor:str:2021}, in particular, the cut elimination theorem, do also work with the rules $\Rtr$ and $\Rref$ in place of the rules in~\eqref{eq:g-rules} above.
\end{remark}

We are now ready to prove the two directions of our translation result. 

\begin{lemma}
	Let $\sq G$ be a sequent (i.e., a \emph{set} of relational atoms and labelled formulas with polarities). Then, $\sq G$ has a proof in $\labISf$ if{f} $\sq G$ has a proof in $\labISfs$. 
\end{lemma}

\begin{proof}
	Both directions are proved by induction on the height of the proof of $\sq G$ in  $\labISf$ (resp. in $\labISfs$), by step-wise translating every rule employed in the proof into (combinations of) rules in in  $\labISfs$ (resp. in $\labISf$). 
	
	First suppose we have a proof of  $\sq G$ in $\labISfs$. To transform this proof into a proof of  $\sq G$ in $\labISf$, we only need to simulate rules  $\labrn{id}$, $\lef{\IMP}$, and $\lef \BOX$  of $\labISfs$ using rules of $\labISf$, as  rule $\monl$ is admissible in $\labISf$. In particular, any sequent which is a conclusion of $\labrn{id}$ is also provable in $\labISf$, using $\Lref$ and $\idm$. Similarly, rule  $\lef \BOX$ can be simulated using rules $\Lref$, $\lefbm$, and the admissible rule $\w$, as follows: 
	\[
		\vlinf{\lef \BOX}{}{\sq G, \accs x y, \labelsb x{\BOX A} }{\sq G, \accs x y,  \labelsb x{\BOX A}, \labelsb yA } 
		\quad 
		\leadsto 
		\quad 
		\vlderivation{
			\vlin{\Lref}{}{ \sq G, \accs x y, \labelsb x{\BOX A} }{
				\vlin{\lefbm}{}{  \sq G, \futs x x, \accs x y, \labelsb x{\BOX A}  }{
					\vlin{\w}{}{  \sq G, \futs x x, \accs x y, \labelsb x{\BOX A}, \labelsb{y}{A} }{
						\vlhy{  \sq G,\accs x y, \labelsb x{\BOX A}, \labelsb{y}{A} }
					}
				}
			}
		}
	\]
	We do not show the case of rule $\lef{\IMP}$, which is similar.
        
	For the other direction, suppose we have a proof of $\sq G$ in $\labISf$. To step-wise translate this proof into  a proof of $\sq G$ in $\labISfs$, we need to show how to simulate the rules $\idm$, $\lefim$ and $\lefbm$  in $\labISf$. This can be easily done using $\monl$, which is a primitive rule of $\labISfs$, and $\w$, which is admissible in $\labISfs$. In particular, the conclusion of  rule $\idm$ is derivable in $\labISfs$, by applying $\monl$ and $\labrn{id}$. To simulate rule  $\lefbm$ we proceed as follows (again, the case of $\lefim$ is similar):
	\[
	\vlinf{\lefbm}{}{\sq G, \futs xy ,\accs y z, \labelsb x{\BOX A} }{\sq G, \futs xy ,\accs y z,  \labelsb x{\BOX A}, \labelsb zA }
	\quad 
	\leadsto 
	\quad 
	\vlderivation{
		\vlin{\monl}{}{ \sq G, \futs xy ,\accs y z, \labelsb x{\BOX A} }{
			\vlin{\lef \BOX}{}{ \sq G, \futs xy ,\accs y z, \labelsb x{\BOX A}, \labelsb y{\BOX A} }{
				\vlin{\w}{}{ \sq G, \futs xy ,\accs y z, \labelsb x{\BOX A}, \labelsb y{\BOX A}, \labelsb{z}{A} }{
					\vlhy{  \sq G, \futs xy ,\accs y z, \labelsb x{\BOX A},  \labelsb{z}{A} }
				}
			}
		}
	}\qedhere
	\]
\end{proof}

The following result is then an easy corollary of the equivalence of the proof systems $\labISf$ and $\labISfs$. 
	
	\begin{theorem}
		\label{thm:labIKs}
		A formula~$\fm A$~is a theorem of\/~$\ISfour$ if{f} for every $\lb x$, the sequent  $\labelsw xA$ has a proof in\/~$\labISfs$.
	\end{theorem}

In what follows, we shall give an alternative proof of completeness of $\labISfs$. The completeness proof is semantic, and is a corollary of our decision procedure, which produces a countermodel for a sequent in case it is not provable in $\labISfs$.

	\subsection{Models from sequents}
	\label{sec:models}

A fully labelled  sequent contains sufficient information to extract a birelational model. This will be useful when proving completeness, as we will be able to immediately construct a (counter)model from a leaf of a failed proof-search tree.

\begin{notation}\label{not:fully-labelled}
	Let $\sq G$~be a sequent.  We write $\labelsof G$ for the set of labels occurring in $\sq G$. On this set we define three binary relations $\srel{G}$, and $\sle G$, and $\svis G$ as follows: 	$\accsq xy$ iff $\accs xy\in\sq G$, and $\futsq xy$ iff $\futs xy\in\sq G$, and $\vissq xy$ iff there is an $\lb {x'}\in\labelsof G$ such that $\futsq x{x'}$ and $\accsq {x'}y$.
\end{notation}

\begin{definition}[Model of a sequent]
	\label{dfn:modelofsequent}
	Let $\sq{G}$~be a sequent. We define the \defn{model~$\modelof{G}$ of}~$\sq{G}$ to be the quadruple
	$\modelof{G}=\tuple{\labelsof G,\srel G,\sle G,V}$ where $V\colon \labelsof G\to 2^\A$ is such that for all atoms $\fm a \in \A$ we have 
	$\fm a \in V(\lb w)$ if{f} $\blackf waG$.
\end{definition}

This model will be a proper birelational model, provided that the sequent satisfies a number of properties, as we will show in Theorem~\ref{thm:completeness} below.
Intuitively, we want all the rules of $\labISfs$ to have been exhaustively but non-redundantly applied to the sequent. We shall express these requirements with the notion of  a \emph{happy sequent} below.

\begin{definition}[Happy labelled formula]\label{def:happyformula}
	Let $\sq G$ be a sequent. We say that a labelled formula~ $\blackf xAG$ (resp.~$\whitef xAG$) is \defn{happy} 
	if{f} the following conditions hold: 
	\begin{itemize}[noitemsep]
		\item $\blackf xaG$ is always happy; 
		\item $\whitef xa G$ is happy if{f}\/ $\blackff xa \notin \sq G$;
		\item $\blackf x\BOT G $ is never happy;
		\item $\whitef x\BOT G$ is always happy;
		\item $\blackf x{A \AND B}G$ is happy if{f}\/ $\blackf xAG$ and\/ $\blackf xB G$; 
		\item
		$\whitef x{A \AND B}G$  is happy if{f}\/ $\whiteff xA \in \sq G$ or\/ $\whiteff xB \in \sq G$; 
		\item
		$\blackf x{A \OR B}G$  is happy  if{f}\/ $\blackff xA \in \sq G$ or\/ $\blackff xB \in \sq G$; 
		\item
		$\whitef x{A \OR B}G$ is happy  if{f}\/ $\whiteff xA \in \sq G$ and\/ $\whiteff xB \in \sq G$; 
		\item
		$\blackf x{A \IMP B}G$ is happy if{f}\/ 
		$\whitef xAG$ or\/ $\blackff xB \in \sq G$;
		\item
		$\whitef x{A \IMP B}G$ is happy if{f}\/ 
		$\blackff yA \in \sq G$ and\/ $\whiteff yB \in \sq G$ for some~$ \lb y $ with\/ $\futsq xy $;
		\item
		$\blackf x{\BOX A}G$ is happy if{f}\/ 
		  $\blackff zA \in \sq G$  
                  for  all~$ \lb z$ with $ \accsq xz $;
		\item
		$\whitef x{\BOX A}G$ is happy if{f}\/ 
		$\whitef zAG$ for some $ \lb y$, $\lb z $ s.t.\ $ \futsq xy$ and $ \accsq yz $, 
		i.e., if{f} $\whitef zAG$ for some  $\lb z $ s.t.\  $ \vissq xz $;
		\item
		$\blackf x{\DIA A}G$ is happy if{f}\/ 
		$\blackf yAG$ for some~$ \lb y $ with\/ $ \accsq xy $; 
		\item 
		$\whitef x{\DIA A}G$ is happy if{f}\/ 
		  $\whiteff yA \in \sq G$ 
                  for all~$\lb y$ s.t.\ $ \accsq xy $.
	\end{itemize}
	Otherwise, the labelled formula is \defn{unhappy}.	
\end{definition}

\begin{definition}[Happy label]\label{def:happy-label}
  A label~$ \lb x $ occurring in a sequent~$\sq{G}$ is \defn{happy} if{f} all labelled formulas that occur at~$ \lb x$ in~$\sq G$ are happy.
\end{definition}

\begin{definition}[Structurally happy sequent] 
	\label{def:struct-sound-seq}
	A sequent~$\sq{G}$ is \defn{structurally happy} if{f} we have:
	\begin{enumerate}[leftmargin=3.5em, noitemsep]
		\item[$( \monl )$] if $\futsq xy$  and $\blackf xCG$, then $\blackf yCG$;
		\item[$ (\fone )$] if $\accsq xy$ and  $\futsq xz$, then there is~$\lb u$ such that $\futsq yu $ and $ \accsq zu$;
		\item[$(\ftwo)$] if $ \accsq xy $ and  $ \futsq yz $, then there is~$\lb u$ such that $ \futsq xu $ and $ \accsq uz$;		
		\item[$ (\Ltr )$] if $ \futsq xy $ and $\futsq yz $, then $\futsq xz $;	
		\item[$ (\Lref )$] $\futsq xx$ for all $ \lb x $ occurring in $\sq{G}$; 
		\item[$ (\Rtr )$] if $ \accsq xy $ and $\accsq yz $, then $\accsq xz $;	
		\item[$ (\Rref )$] $\accsq xx$ for all $ \lb x $ occurring in $\sq{G}$.
	\end{enumerate}
\end{definition}

\begin{definition}[Happy sequent]\label{def:happysequent}
	A sequent~$\sq{G}$ is \defn{happy} if{f} it is structurally happy and all labels in the sequent are happy.
\end{definition}

We can now show that the model of a happy sequent is a birelational model. Moreover, all input formulas are satisfied in the model, and all output formulas are falsified in the model. 

\begin{figure}[!t]
	\begin{minipage}{0.6\textwidth}
		$$
		\begin{array}{c c l}
			\gG^{\sq {R}} & = & \futs 1 1, \futs 22, \futs 33, \futs 44, \futs 55, \futs 66,\\
			& & \accs 11, \accs 22, \accs 33, \accs 44, \accs 55, \accs 66,\\
			& & \futs{1}{2}, \futs{2}{5}, \futs{3}{4}, \futs 1 5,\accs{2}{3}, \accs{5}{4}, \accs{4}{6}, \accs{5}{6}\\[1ex]
			\lef\gG & = & \labelsb{1}{a \OR b}, \labelsb{2}{a \OR b}, \labelsb{5}{a \OR b},\\  
			& & \labelsb{1}{a}, \labelsb{2}{a}, \labelsb{5}{a},  \labelsb{4}{\DIA d},  \labelsb{6}{d} \\[1ex]
			\sq G^\rigmark & = & \labelsw{1}{\BOX (c \OR (\DIA d \IMP e))}, \labelsw{3}{ c \OR (\DIA d \IMP e) }, \\
			& &\labelsw{3}{c}, \labelsw{3}{\DIA d \IMP e},  \labelsw{4}{e}\\[1ex]
                        \gG&=&\gG^{\sq R}\cup\lef\gG\cup\rig\gG
		\end{array}
		$$
	\end{minipage}
	\begin{minipage}{0.3\textwidth}
		\begin{tikzpicture}[thick, every node/.style={scale=1.2}]

\tikzstyle{node}=[circle,draw]
\tikzstyle{s-node}=[circle,draw,inner sep=2pt]
\tikzstyle{nonode}=[inner sep=0pt]
\tikzstyle{access}=[->,blue]
\tikzstyle{future}=[->,dashed]
\tikzstyle{deleted}=[->,gray]
\tikzstyle{loop}=[->,red]

\node[label={[align=center]right :{
	$\lef a$
	}}] 
	(1) at (2,0) [node, 
	] {$1$};

\node[label={[align=center]above right :{
	$\lef a$
	}}] 
	(2) at (2,2) [node,
	] {$2$};
\node[label={[align=left]right :{
	}}] (3) at (4,2) [node, 
	] {$3$};

\node[label={[align=center]above right:{
			$\lef a$
	}}] (100) at (2,4) [node, 
	] {$5$};

\node
	  (10) at (4,4) [node, 
	  ] {$4$};
\node[label=right :{$\lef d$}] (11) at (6,4) [node,
] {$6$};

\draw[access] (2) -- (3);

\draw[access] (100) -- (10);
\draw[access] (10) -- (11);

\draw[future] (1) -- (2);
\draw[future] (3) -- (10);
\draw[future] (2) -- (100);
%
\end{tikzpicture}
	\end{minipage}
	\caption{Model $\modelof{G}$ of a happy sequent $\sq G$. The nodes represent the worlds of the model, the dashed arrows represent the $\le$-relation, and the solid arrows represent the $\rel$-relation. We do not represent the reflexive  $\le$-edges and $\rel$-edges of $\modelof{G}$, as well as edges $\futs 1 5$ and $\accs 5 6$, which follow by transitivity.  
	}
	\label{fig:happy}
\end{figure}

\begin{theorem}
	\label{thm:completeness}
	For a happy sequent\/~$\sq{G}$, its model~$\modelof{G}=\tuple{\labelsof G,\srel G,\sle G,V}$ is a birelational model satisfying the following:  
	$i)$ if\/ $\blackf xAG$, then $\force{\modelof{G}}xA$; and $ii)$ if\/ $\whitef xAG$, then $\nforce{\modelof{G}}xA$.
\end{theorem}

\begin{proof}
  This is a straightforward induction on the structure of $\fm{A}$ (see, e.g., \cite{GKMMS:IK}).
\end{proof}
\begin{example}
	\label{ex:happy}
	Take the model $\modelof{G}$ of the happy sequent~$\sq{G}$ represented in Figure~\ref{fig:happy}. 
	It holds that, e.g., $\force{\modelof{G}}{\lb 4}{\DIA d}$  and  $\nforce{\modelof{G}}{1}{ \BOX (c \OR (\DIA d \IMP e))}$.
\end{example}

In what follows, we will often make use of the notion of axiomatic sequent.

\begin{definition}[Axiomatic sequent]
	A sequent~$\sq{G}$ is \defn{axiomatic} if{f} there is a label~$\lb x$ such that either $\blackf xaG$ and $\whitef xaG$ for some~$\fm a$, or $\blackf x\BOT G$.  Otherwise, $\sq{G}$~is called \defn{non-axiomatic}.
\end{definition}

\begin{remark}
	An axiomatic sequent~$\sq G$ is never happy, because either $\whitef xaG$ or $\blackf x\BOT G$ is unhappy.
\end{remark}

	\subsection{Layers, clusters}
	\label{sec:layers:clusters}

	Our algorithm relies on the fact that  the sequents created by the algorithm are all \emph{layered} (provided the endsequent contains only one label), where each layer can be thought of as a classical $\Sfour$-model.

	\begin{definition}[Layer]\label{def:layer}
		For a sequent~$\sq G$, we define the relation~$\lrel G$ to be the reflexive and transitive closure of $\grel G \cup \grel{G}^{-1}$. Since this is an equivalence relation, we can define
		a \defn{layer}~$L$ in~$\sq G$ to be an equivalence class of~$\lrel G$.
		A layer $L$ is \defn{happy} if{f} all labels $\lb x\in L$ are happy.
	\end{definition}

	\begin{definition}[Layered sequent]\label{def:lay-sq}
		We say that a sequent~$\sq{G}$ is \defn{layered} if{f} for any 
		labels~$\lb x$, $\lb{x'}$, $\lb y$,~and~$\lb{y'}$ occurring in~$\sq G$: 
		\begin{enumerate}
			\item if   $\accsqh xy$ for $\lb x \neq \lb y$, then $\nfutsq xy$ and $\nfutsq yx$; and
			\item 
			\label{def:lay-sq-it2}
			if  $\accsqh xy$ and $\accsqh{x'}{y'}$ and $\futsq{x}{x'}$ for $\lb x \neq \lb {x'}$,  then  $\nfutsq{y'}{y}$.
		\end{enumerate}
		For layers~$L_1$~and~$L_2$, we define $L_1 \le L_2$ whenever there are labels $\lb x \in L_1$ and $\lb y \in L_2$ such that $\futsq xy$. We write $L_1<L_2$ if{f} $L_1 \le L_2$ and $L_1 \neq L_2$.
	\end{definition}

        \begin{proposition}
          Let $L_1$ and $L_2$ be layers in a layered sequent $\sq G$, and let $\lb x\in\labelsof G$. If $\lb x\in L_1$ and $\lb x\in L_2$, then $L_1=L_2$.
        \end{proposition}

        \begin{proof}
          Take $\lb y\in L_1$. This means that $\accsqh xy$. Therefore also $\lb y\in L_2$. And vice-versa.
        \end{proof}

        \begin{definition}[Layer of a label]\label{def:layer-of-label}
          Let $\sq G$ be a layered sequent and let $\lb x\in\labelsof G$. The \defn{layer of~$\lb x$}, denoted by $\layerof x$, is the unique layer with $\lb x\in\layerof x$. 
        \end{definition}

\begin{proposition}
	\label{prop:hat_satur}
	Structurally happy sequents\/~$\sq G$  are also structurally happy w.r.t.~$\lrel G$:
	\begin{enumerate}[leftmargin=3.5em]
		\item[$(\fhat)$] if $\accsqh xy $ and  $ \futsq xz $, then there is~$\lb u$ such that $ \futsq yu $ and $\accsqh zu$.
	\end{enumerate}
\end{proposition}

\begin{proof}
	There must be a sequence $\lb{x_0}, \dots, \lb{x_n}$ of labels with $n\ge0$ such that $\lb{x_0}=\lb x$, $\lb{x_n}=\lb y$, and  for each $0 \le i\le n-1$, either $\accsq{x_i}{x_{i+1}}$ or $\accsq{x_{i+1}}{x_{i}}$. 
	We use induction on $n$. 
	If $n=0$, i.e., $\lb x = \lb y$, then $\lb u \colonequals \lb z$ suffices. Otherwise, by IH, there is~$\lb v$ such that $\futsq{x_{n-1}}v$ and $\accsqh zv$. If $\accsq{x_{n}}{x_{n-1}}$, there is $\lb u$ by $(\ftwo)$-structural happiness  such that  $\futsq{x_n}u$ and $\accsq uv$. If  $\accsq{x_{n-1}}{x_n}$, there is $\lb u$ by $(\fone)$-structural happiness such that $\futsq{x_n}u$ and $\accsq vu$. Either way, $\futsq yu$ and $\accsqh zu$.
\end{proof}    

\begin{proposition}
	\label{prop:order}
	For a structurally happy and layered sequent $\sq G$, the relation\/~$\le$ is a partial order on the layers of $\sq G$.
\end{proposition}

\begin{proof}
	For reflexivity, consider any layer~$L$. For any label $\lb x \in L$, by ($ \Lref $)-structural happiness, $\futsq xx$. Hence, $L \le L$ by Definition~\ref{def:lay-sq}.
	
	For transitivity, let $L_1 \le L_2$  and $L_2 \le L_3$ for layers~$L_1$, $L_2$,~and~$L_3$.  By Definition~\ref{def:lay-sq}, $\futsq x{y'}$ and $\futsq yz$ for some labels~$\lb x \in L_1$, $\lb {y'}, \lb y \in L_2$, and $\lb z \in L_3$. Since $\accsqh {y'}y$, by ($\fhat$)-structural happiness (Proposition~\ref{prop:hat_satur}), there is a label $\lb {z'}$ such that $\futsq {y'}{z'}$ and $\accsqh {z'}z$. The latter means that $\lb {z'} \in L_3$.  By ($ \Ltr $)-structural happiness, $\futsq x{z'}$. Hence, $L_1 \le L_3$ by Definition~\ref{def:lay-sq}.
	
	For antisymmetry, let $L_1\le L_2$ and $L_1\neq L_2$.
	By Definition~\ref{def:lay-sq}, there are labels $\lb{x} \in L_1$ and $\lb{x'} \in L_2$ such that  $\futsq{x}{x'}$. For arbitrary labels $\lb{y} \in L_1$ and $\lb{y'} \in L_2$, we have $\accsqh xy$ and $\accsqh{x'}{y'}$. Since $L_1 \neq L_2$, we have $\lb x \neq \lb{x'}$. Hence, $\nfutsq{y'}y$ by~\ref{def:lay-sq-it2}) of Definition~\ref{def:lay-sq}. Since $\nfutsq{y'}y$ for any $\lb{y'} \in L_2$ and $\lb{y} \in L_1$, it is not the case that $L_2 \leq L_1$. 
\end{proof}

        \begin{remark}
        \label{rem:antisymm}
        The second condition in Def.~\ref{def:lay-sq} is only needed to establish antisymmetry in Proposition~\ref{prop:order}. If $\sle G$ on labels is already antisymmetric then the second condition follows from the first. 
        \end{remark}

In a sequent $\sq G$  constructed by $\labISfs$,  each layer also has the structure of (the reflexive and transitive closure of) a tree  with respect to~$\grel G$. 
But, as stated before, in order to search for a proof and a countermodel at the same time, we need to weaken this tree structure on the layers.%

\begin{definition}[Having no past/no future]
	A label~$\lb z$ in a sequent~$\sq G$ \defn{does not have a past} if{f} $\futsq xz$~implies $\lb x =  \lb z$ for all~$\lb x\in\labelsof G$, and $\lb z$ \defn{does not have a future} if{f} $\futsq zx$~implies $\lb x =  \lb z$ for all~$\lb x\in\labelsof G$.
	%
\end{definition}

\begin{proposition}
\label{prop:no_future_layer}
Let $L$ be a layer in a layered and structurally happy sequent $\sq G$. Then the following three statements are pairwise equivalent
\begin{enumerate}[(i)]
\item $L$ is maximal w.r.t.~the partial order $\le$ on the layers of $\sq G$;
\item no label from $L$ has a future;
\item at least one label $\lb z \in L$ has no future.
\end{enumerate}
Such maximal layers will also be called \defn{topmost}.
\end{proposition}
\begin{proof}
The implication from (i) to (ii) directly follows from the definition of layeredness and  $\le$ on layers. 
The implication from (ii) to (iii) is trivial. 
For the implication from (iii) to (i),
suppose towards a contradiction that $\lb z \in L$ has no future but $L < L'$ for some  layer $L'$. This means that there are labels $\lb y \in L$ and $\lb{y'} \in L'$ such that $\futsq y{y'}$. By $(\fhat)$-structural happiness (Prop.~\ref{prop:hat_satur}), since $\accsqh zy$, there is a label $\lb {z'}$ such that $\futsq z{z'}$ and $\accsqh {z'}{y'}$, i.e., $\lb{z'} \in L'$. Since $L \ne L'$, it follows that $\lb z \ne \lb{z'}$, hence, $\lb z$~has a future, in contradiction to the assumption.
\end{proof}

\begin{definition}[Cluster]
  \label{def:cluster}
  If~$\sq G$ is structurally happy, then $\grel G \cap \grel{G}^{-1}$~is an equivalence relation, and we can define a \defn{cluster}~$\lbc C$ in~$\sq G$ to be an equivalence class $\lbc C=\{\lb{x_1},...,\lb{x_n}\}$ of~$\grel G \cap \grel{G}^{-1}$. For a label $\lb x\in\labelsof G$, we write $\clusterof x$ for the unique cluster containing $\lb x$, and call it the \defn{cluster of~$\lb x$}. A cluster~$\lbc C=\set{\lb{x_1}}$ containing only one label is called \defn{singleton}. On clusters, we define two binary relations:
  \begin{itemize}
  \item $\futsqc{C_1}{C_2}$ if{f} for all $\lb y\in \lbc{C_2}$ there is $\lb x\in \lbc{C_1}$ with $\futsq xy$.
  \item $\accsqc{C_1}{C_2}$ if{f} there are $\lb x\in\lbc C_1$ and $\lb y\in\lbc C_2$ with $\accsq xy$. 
  \end{itemize}
  We sometimes abuse the notation and replace one of the clusters by a label $\lb x$ even when $\{\lb x\}$ is not a cluster. Nevertheless, the definitions are then applied verbatim to $\{\lb x\}$.
\end{definition}

\begin{proposition}[restate = OrderClusters, name = ] 
	\label{prop:order_clusters}
	For a structurally happy and layered sequent\/~$\sq G$, the relation $\grel G$ and the relation $\sle G$ are partial orders on the clusters of $\sq G$. 
\end{proposition}

\begin{proof}
	For reflexivity of~$\grel G$, consider any cluster~$\lbc C$. For any label $\lb x \in \lbc C$, by ($ \Rref $)-structural happiness, $\accsq xx$. Hence, $\accsq CC$ by Definition~\ref{def:cluster}. 
	For transitivity of~$\grel G$, let $\accsq{C_1}{C_2}$  and $\accsq{C_2}{C_3}$ for clusters~$\lbc C_1$, $\lbc C_2$,~and~$\lbc C_3$.  By Definition~\ref{def:cluster}, $\accsq xy$ and $\accsq uz$ for some labels~$\lb x \in \lbc{C_1}$, $\lb y, \lb u \in \lbc{C_2}$, and $\lb z \in \lbc{C_3}$. Since $\accsq yu$, by $(\Rtr)$-structural happiness $\accsq xz$. Thus, $\accsq{C_1}{C_3}$ by Definition~\ref{def:cluster}. 
	For antisymmetry of~$\grel G$, let $\accsq{C_1}{C_2}$ and $\accsq{C_2}{C_1}$.  
	By Definition~\ref{def:cluster}, there are labels $\lb{x}, \lb{y} \in \lbc{C_1}$ and $\lb{x'},\lb{y'} \in \lbc{C_2}$ such that  $\accsq{x}{x'}$ and $\accsq{y'}{y}$. For arbitrary labels $\lb{u} \in \lbc{C_1}$ and $\lb{v} \in \lbc{C_2}$, we have $\accsq ux$, and $\accsq{x'}v$, hence, by $(\Rtr)$-structural happiness, $\accsq uv$. Similarly, $\accsq vu$ because $\accsq v{y'}$ and $\accsq yu$. Since both $\accsq uv$ and $\accsq vu$  for all $\lb{u} \in \lbc{C_1}$ and $\lb{v} \in \lbc{C_2}$, all labels in these two clusters form one equivalence class w.r.t.~$\grel G \cap \grel{G}^{-1}$, i.e.,~$\lbc{C_1} = \lbc{C_2}$. 
	
	For reflexivity of~$\sle G$, consider any cluster~$\lbc C$. For every label $\lb y \in \lbc C$, by ($ \Lref $)-structural happiness, $\futsq yy$. Hence, $\futsq CC$ by Definition~\ref{def:cluster}. 
	For transitivity of~$\sle G$, let $\futsq{C_1}{C_2}$  and $\futsq{C_2}{C_3}$ for clusters~$\lbc C_1$, $\lbc C_2$,~and~$\lbc C_3$.  By Definition~\ref{def:cluster}, for every  $\lb z \in \lbc{C_3}$, there is $\lb y \in \lbc{C_2}$ such that $\futsq yz$. In its turn, for this~$\lb y$, there is $\lb x \in \lbc{C_1}$ such that $\futsq xy$. By $(\Ltr)$-structural happiness $\futsq xz$ for this~$\lb x$. Thus, $\futsq{C_1}{C_3}$ by Definition~\ref{def:cluster}. Hence, $\sle G$ is a preorder.
	
	Assume now additionally that $\sq G$ is layered.
	For antisymmetry of~$\sle G$, let $\futsq{C_1}{C_2}$ and $\lbc{C_1}\neq\lbc{C_2}$, i.e., $\lbc{C_1} \cap \lbc{C_2} = \varnothing$. To show that $\nfutsq{C_2}{C_1}$, we consider any label $\lb x \in \lbc{C_1}$ and show that $\nfutsq yx$ for any~$\lb y \in \lbc{C_2}$.
	By Definition~\ref{def:cluster}, for each $\lb y \in \lbc{C_2}$, there is some  label $\lb{x'} \in \lbc{C_1}$ such that  $\futsq{x'}y$. We have $\accsqh{x'}x$ because they belong to the same cluster, $\accsq yy$, and hence $\accsqh yy$, by $(\Rref)$-structural happiness, and $\lb{x'} \ne \lb y$ because they are from disjoint clusters. Hence,  $\nfutsq yx$ by Definition~\ref{def:lay-sq}. 
	Since $\lb y$ was chosen arbitrarily, it follows that $\nfutsq{C_2}{C_1}$ and $\sle G$~is a partial order.	
\end{proof}

\section{Proof Search with Macro-Rules}
\label{sec:macro-rules}

To define our algorithm, it is convenient to ``group together'' the rules of $\labISfs$ into a few \emph{macro-rules}. 
This allows us to include part of the proof search strategy within the macro-rules, or as side conditions to the macro-rules.

Furthermore, these macro-rules are such that every sequent $\sq G$ visited by our algorithm 
is structurally happy and satisfies two additional properties: the order relation $\le$ between layers of~$\sq G$  forms a tree and the order relation  $\grel G$ between clusters in the same layer also forms a tree. 
	\begin{definition}[Nice sequent]\label{def:nice}
		A sequent $\sq G$ is \defn{nice} if{f} it is (i) layered (Def.~\ref{def:lay-sq}), (ii) structurally happy (Def.~\ref{def:struct-sound-seq}), and is such that (iii) the set of its layers form a tree with respect to the relation~$\le$, and (iv) in every layer $L$, all the clusters in $L$ form a tree with respect to the relation~$\grel G$. 
	\end{definition}

\subsection{Semi-Saturation}	
\label{subs:semi-sat-dia}

First observe that all logical rules of $\labISfs$, except for $\rig\IMP$ and $\rig\BOX$, only work locally on a layer, whereas $\rig\IMP$ and $\rig\BOX$ create a new layer. Our objective is to first `exhaust' all the rules applicable to a layer before moving to the next one. However, because of the rule  $\lef\DIA$ (in combination with the structural rules), a layer might grow indefinitely. We thus need to handle with care applications of rule $\lef\DIA$, which introduces fresh labels in a layer. But we can eagerly apply rules $\lef\AND,\rig\AND,\lef\OR,\rig\OR,\lef\IMP,\lef\BOX$, and $\rig\DIA$ to a layer. 
This process is called \emph{semi-saturation} and will be encompassed by our first macro-rule.

\begin{definition}[Semi-saturation]\label{def:sat-sqset}
	Given a sequent $\sq G$, a \defn{semi-saturation tree} for $\sq G$ is a derivation tree $\tr T$ having sequents $\sq G_1, \dots, \sq G_n$ as leaves and 
 	constructed by applying the rules in $\set{\lef\AND,\rig\AND,\lef\OR,\rig\OR,\lef\IMP,\lef\BOX,\rig\DIA}$. 
	With this, we can define the following \emph{macro rule} 
	\begin{equation}
		\label{eq:satr}
		\vliiinf{\satr}{}{\sq G}{\sq G_1}{\ldots}{\sq G_n}
	\end{equation}
	called \defn{semi-saturation}. A sequent $\sq G$ is called \defn{semi-saturated} if{f} all labelled formulas occurring in $\sq G$ are happy, except those of shape $\fmb\BOT$, $\fmw a$, $\fmw{A \IMP B}$, $\fmw{\BOX A}$ and~$\fmb{\DIA A}$. We write
 	\begin{equation}
		\label{eq:ssatr}
		\vliiinf{\ssatr}{}{\sq G}{\sq G_1}{\ldots}{\sq G_n}
	\end{equation}
       for an instance of $\satr$ if all $\sq G_1,\ldots,\sq G_n$ are semi-saturated, and in this case we call the set $\set{\sq G_1,\ldots,\sq G_n}$ a \defn{semi-saturation} of $\sq G$.
\end{definition}

We immediately obtain the following:
	
\begin{lemma}\label{lem:satr}
  The $\satr$-rule is derivable in $\labISfs$.
\end{lemma}

\begin{proof}
  By definition, every instance of the $\satr$-rule is a derivation tree in  $\sysS=\set{\lef\AND,\rig\AND,\lef\OR,\rig\OR,\lef\IMP,\lef\BOX,\rig\DIA}$. 
\end{proof}

\begin{lemma}\label{lem:satr-nice}
  If an instance of the $\satr$-rule only modifies labels that do not have a future and if its conclusion is nice, then so is each of its premises.
\end{lemma}
\begin{proof}
  Immediate, as we do not generate any new labels in the premise of the rule. 
\end{proof}

\subsection{Black Diamonds}

	The next macro-rule corresponds to an application of $\lef\DIA$ followed by structural rules. 
	\begin{equation}
          \label{eq:bdiam}
	\vlinf{\bdiam}{\proviso{$\lb y$  fresh}}{\sq G, \labelsb x{\DIA A} }{\sq G, \labelsb x{\DIA A} \cup \{ \accs xy, \accs y y,\futs y y,  \labelsb yA\} \cup \{ \accs v y \mid \accs v x \in \sq G\} }
        \end{equation}
        In the premise of the $\bdiam$-rule above, the label $\lb x$ is called the \defn{parent} of the new label~$\lb y$. 
        
	\begin{lemma}\label{lem:bdia}
	The $\bdiam$-rule is derivable in $\labISfs$.
        \end{lemma}
       
        \begin{proof}
          The rule $\bdiam$ is an application of $\lef\DIA$ followed bottom-up by an instance of $\Lref$ and $\Rref$, and several instances of $\Rtr$.
      \end{proof}

  \begin{lemma}\label{lem:bdiam-nice}
    If the conclusion of a $\bdiam$-rule instance is nice and $\lb x$ has no future, then its premise is also nice.
  \end{lemma}

  \begin{proof}
  Property (i) is preserved because the only new $\le$-relational atom  added is the reflexive $\futs y {y'}$.
    The only new label $\lb y$ and all new relational atoms $\accs u {u'}$ are added to the layer $\layerof x$, hence, the structure of layers in the sequent is not changed  and Property~(iii) holds trivially. Property (ii), i.e, that the premise is structurally happy,  also follows immediately, as we close $\grel G$ under transitivity and reflexivity, and add the necessary $\sle G$-reflexive loop for $\lb y$.  $(\monl)$ is preserved trivially, while $( \fone )$, and $(\ftwo)$  are preserved because $\lb x$~has no future, meaning that  $\layerof x$ is a topmost layer and no $\lb v$ with newly added $\accs v y$ has a future either (Prop.~\ref{prop:no_future_layer}). Finally, (iv) holds because adding $\lb y$ does not break the tree-structure on the clusters of the layer~$\layerof x$. 
  \end{proof}

\subsection{Layer Saturation}

\begin{definition}[Almost happy]\label{def:almost-happy}
  A label~$ \lb x $ in a sequent $\sq G$ is \defn{almost happy} if{f} all formulas occurring at~$\lb x$ are happy except, possibly, those of the shapes $\fmb\BOT$, $\fmw a$, $\fmw{A \IMP B}$, and~$\fmw{\BOX A}$.
  A layer~$L$ (resp.~sequent~$\sq G$) is \defn{almost happy} if{f} all labels in~$L$ (resp.~$\sq G$) are almost happy. 
\end{definition}

One step in our proof search algorithm will be to make sequents almost happy. As this does not involve the creation of new layers, we call this procedure \emph{layer saturation}. The basic idea is to use semi-saturation and then use the $\bdiam$-rule to make the $\lef\DIA$-formulas happy. As this might create new unhappy formulas, we have to semi-saturate again. And so on. As this might not terminate, we have to implement a loop check. This is very similar to what is usually done for classical $\Sfour$.

\begin{definition}[Equivalent labels]\label{def:eq-labels}
	Let $\sq G$ be a sequent and let $\lb x,\lb y\in\labelsof G$. We say that $\lb x$~and~$\lb y$~are \defn{equivalent}, denoted as $\lb x \lbeq \lb y$, if{f} for all formulas~$\fm A$, we have $\blackf xA G$ if{f}\/ $\blackf yA G$, and also\ $\whitef xA G$ if{f}\/ $\whitef yA G$. 
\end{definition}


\begin{notation}[Substituting and clustering]\label{not:subst}
  Let $\sq G$ be a sequent and let $\accsq xy$. We write $\sq G \labelsubst{y}{x}$ for the result of replacing $\lb y$ with $\lb x$ everywhere in $\sq G$. 
  And we write $\interval xy$ for the set $\set{\lb v\mid \accsq xv \mbox{ and } \accsq vy \mbox{ and } \lb v\neq\lb y}$ of labels  between $\lb x$ and $\lb y$, without $\lb y$. 
  Finally, for a set $V\subseteq\labelsof G$ of labels, we write $\trans V$ for the following set of relational atoms
  $$\trans V\;=\;\set{\accs uv\mid \lb u,\lb v\in V}\cup\set{\accs uz\mid\lb u\in V \mbox{~and~} \lb z\in\labelsof G\setminus V\mbox{~and~}\exists\lb w\in V.\;\accsq wz}\quad.$$
  In this paper we make heavy use of the set $\trans{\interval xy}$. When we add this to a sequent $\sq G$ with $\accsq xy$, then all labels in $\interval xy$ will form a cluster.
  The sequent $\sq G \labelsubst{y}{x}\cup\trans{\interval xy}$ is then obtained from $\sq G$ by adding the edges $\trans{\interval xy}$ and then replacing $\lb y$ with $\lb x$.
\end{notation}

\begin{figure}
		\begin{center}
			\begin{tikzpicture}[thick, every node/.style={scale=1.2},font = {\large}]
		
		\tikzstyle{node}=[circle,fill=black,inner sep=1.5pt]
		\tikzstyle{nonode}=[inner sep=0pt]
		\tikzstyle{access}=[blue]
		\tikzstyle{loop}=[blue]
		\tikzstyle{future}=[dashed,->]
		\tikzstyle{prev}=[densely dotted]
		
		\draw[loop] (4,5) circle [radius=1cm];
		
		\node[] (z0') at (1,5) [] {};
		\node[label=below:{$z_1$}] (z1') at (2,5) [node] {};
		\node[label=above:{$x$\,\,\,}] (x') at (3,5) [node] {};
		\node[label=above:{$z_2$}] (z2') at (4.4,5.9) [node] {};
		\node[label=right:{$z_3$}] (z3') at (5.4,6.5) [node] {};
		\node[label=below:{$z_4$}] (z4') at (4.45,4.1) [node] {};
		\node[label=right:{$z_5$}] (z5') at (5.4,4.7) [node] {};

		\node[] (z0) at (1,2) [] {};
		\node[label=below:{$z_1$}] (z1) at (2,2) [node] {};
		\node[label=below:{$x$}] (x) at (3,2) [node] {};
		\node[label=below:{$z_2$}] (z2) at (4,2) [node] {};
		\node[label=right:{$z_3$}] (z3) at (5,2.7) [node] {};
		\node[label=below:{$z_4$}] (z4) at (6,2) [node] {};
		\node[label=right:{$z_5$}] (z5) at (7,2.7) [node] {};
		\node[label=below:{$y$}] (y) at (7.5,1.3) [node] {};

		\draw[-,thin] (1,3.3) -- (7.8,3.3);
		\node[] (lb) at (8.4,3.3) {\small $\dialoopr$};

		\draw[prev] (z0) -- (z1);
		\draw[access] (z1) -- (x);
		\draw[access] (x) -- (z2);
		\draw[access] (z2) -- (z3);
		\draw[access] (z2) -- (z4);
		\draw[access] (z4) -- (z5);
		\draw[access] (z4) -- (y);
		
		\draw[prev] (z0') -- (z1');
		\draw[access] (z1') -- (x');
		\draw[access] (z2') -- (z3');
		\draw[access] (z4') -- (z5');

	\end{tikzpicture}
	
	\end{center}
\caption{An example of $\dialoopr$, where $\lb y$ is a $\DIA$-repetition of $\lb x$.}
\end{figure}

\begin{definition}[$\DIA$-repetition]
  Let $\sq G$ be a nice sequent, let $\lb x,\lb y\in\labelsof G$ such that $\lb y$ does not have a future and there is a labeled formula  $\labelsb y{\DIA A}\in \sq G$. 
  We say that $\lb y$ is a \defn{$\DIA$-repetition} of $\lb x$ if{f}
  $\lb x\neq \lb {y}$ and $\accsq xy$ and $\lb x\lbeq \lb y$
  and every $\lb z$ with $\accsq xz$ and $\accsq zy$ does not have a past in $\sq G$.\footnote{This includes $\lb x$ and $\lb y$.}
  This allows us to define the following rule:
  \begin{equation}
    \label{eq:dialoop}
    \vlinf{\dialoopr}{
      \proviso{$\lb y$ is a $\DIA$-repetition of $\lb x$ in $\sq G$}}{
      \sq G}{
      \sq G \labelsubst{y}{x}\cup\trans{\interval xy}
    } 
  \end{equation}
\end{definition}

\begin{lemma}\label{lem:dialoop-nice}
  If the conclusion of a $\dialoopr$-rule instance is nice, then so is its premise.
\end{lemma}

\begin{proof}
    All changes to the $\rel$-relational atoms  are within the layer $\layerof x$ and the changes to the $\le$-relational atoms  affect neither the source nor the destination layer, hence, the structure of layers in the sequent is not changed and (i) and (iii) hold trivially. Property (ii), i.e, that the premise is structurally happy,  also follows immediately because by the definition of $\trans{\interval xy}$ we ensure transitivity of both $\grel G$ and $\sle G$, as well as $( \fone )$ and $(\ftwo)$ where we also use the fact that $\layerof x$ is a topmost layer. Additionally,  $(\monl)$ is preserved because $\lb x \lbeq \lb y$. Finally, (iv) holds because $\lb x$ and $\lb y$ are on the same $\rel$-branch within  $\layerof x=\layerof y$, which means that this layer remains a tree of clusters. 
\end{proof}

\begin{remark}
  Observe that unlike the other macro-rules defined in the previous section, the rule $\dialoopr$ defined in~\eqref{eq:dialoop} is not derivable. It is not even sound in the general case. We will use it during proof search and we will show how it can be eliminated in Section~\ref{sec:soundness}.
\end{remark}

Given a sequent $\sq G$, we define an order $\lford$ on the set $\labelsof G$ as the transitive closure of the following: $\lord xy$ iff either $\layerof x<\layerof y$, or $\layerof x=\layerof y$ and there are $\lb {x'},\lb{y'}\in\labelsof G$, such that  and $\futsq{x'}x$ and $\futsq{y'}y$ and $\lb{x'}$ is the parent of $\lb{y'}$.

On the set $\cF$ on formulas, we can define a total order using some G\"odel numbering. We can then combine the two to define an ordering on the labelled formulas in a sequent $\sq G$ by $\labels xA\lford\labels yB$ iff $\lord xy$, or $\lb x =\lb y$ and $\fm A\lford \fm B$.

The purpose of this ordering is to make sure that in the algorithm below, (i) a $\labelsb y{\DIA A}$ formula is only being looked at when all labels $\lb x$ with $\accsq xy$ have been made almost happy, and (ii) at each label the $\lef\DIA$-formulas are treated in the same order.

\begin{algorithm}[Layer Saturation]\label{alg:sat}
  For a given sequent $\sq G$, in which all non-topmost layers are almost happy, we compute a set~$\ssatof G$ of sequents, called the \defn{saturation} of $\sq G$, by constructing the sequence $\sqset S_0,\sqset S_1,\sqset S_2,\ldots$ as follows:
  \begin{enumerate}[(1)]
  \item\label{l:s1} Initialise the set $\sqset S_0$ with the semi-saturation of $\sq G$.
  \item\label{l:s2} If all sequents in $\sqset S_i$ are almost happy, then let $\ssatof G=\sqset S_i$ and terminate.
  \item\label{l:s3} Otherwise, pick a sequent $\sq H\in\sqset S_i$ that is not almost happy, and pick an unhappy $\labelsb y {\DIA A}$ in~$\sq H$ which is minimal for $\lford$. 
    \begin{enumerate}[(a)]
    \item\label{l:dialoop}  If $\lb y$ is a $\DIA$-repetition of a label $\lb x$, then let $\sq H'= \sq H \labelsubst{y}{x}\cup\trans{\interval xy}$.
    \item\label{l:dianew} Otherwise, let $\lb z$ be a fresh label and let $\sq H'=\sq H \cup 
      \set{\accs yz,\accs z z, \futs zz, \labelsb zA} \cup \set{ \accs v z \mid \accs vy \in \sq H}$.
    \end{enumerate}
  \item Let $\sqset S_i'$ be the semi-saturation of $\sq H'$. 
  \item Let $\sqset S_{i+1} =  (\sqset S_i \setminus \{ \sq H\} ) \cup\sqset S'_i$, and go to \ref{l:s2}.
  \end{enumerate}
\end{algorithm}

\begin{lemma}\label{lem:sat-term}
	Algorithm~\ref{alg:sat} terminates.
\end{lemma}

\begin{proof}
  Any formula that can be added to a label during this process must be a subformula of~$\sq G$.  Let $n$ be the number of subformulas occurring in $\sq G$. Then there are at most $2^n$ different equivalence classes for the relation $\lbeq$ on labels. That means that in each branch of the layer (seen as a tree of clusters) we will eventually stop growing because of $\DIA$-repetitions.
\end{proof}

\begin{lemma}\label{lem:saturation}
  For a sequent $\sq G$, in which all non-topmost layers are almost happy, its saturation $\ssatof G$ is a finite set of almost happy sequents and comprises the leaves of a derivation tree constructed only from the rules $\satr$, $\bdiam$, and $\dialoopr$. Furthermore every $\sq H \in \ssatof G$ contains only subformulas of $\sq G$.
\end{lemma}

\begin{proof}
  By the previous lemma, Algorithm~\ref{alg:sat} terminates. Step~\ref{l:s2} is an application of $\satr$, Step~(3.a) is $\dialoopr$, and Step (3.b) is an instance of $\bdiam$.
\end{proof}

\begin{remark}\label{rem:sat}
	The saturation of a sequent is in general not uniquely defined since sequents can differ up to a renaming of labels. However, we can fix a strategy and naming scheme for fresh labels when applying the inference rules $\lef\AND,\rig\AND,\lef\OR,\rig\OR,\lef\IMP,\lef\BOX,\rig\DIA$, so that without loss of generality, we can for the context of this paper speak of \emph{the} saturation of a sequent $\sq G$, and denote it by~$\satof{\sq G}$.
\end{remark}

\begin{lemma}\label{lem:ssat-nice}
   Le $\sq G$ be a sequent in which all non-topmost layers are almost happy. If $\sq G$ is nice, then every sequent $\sq H\in \ssatof G$ is also nice.
\end{lemma}

\begin{proof}
  This follows immediately from Lemmas~\ref{lem:saturation}, \ref{lem:satr-nice}, \ref{lem:bdiam-nice}, and~\ref{lem:dialoop-nice}.
\end{proof}

  \subsection{Creating a new layer}
  \label{sec:lifting}
       
	The rules $\rig\IMP$ and $\rig\BOX$ create a new layer, and in order to preserve niceness, our algorithm will immediately complete the layer, making the sequent structurally happy, by applying the structural rules of $\labISfs$. 
	We introduce two macro-rules, taking care of $\rig\IMP$ and $\rig\BOX$ plus the structural happiness. 
	For this, it is convenient to introduce a notion of \emph{lifting}. 
	
        We first discuss the cases in which a $\rig{\Box}$- or a $\rig\IMP$-rule is applied to a label $\lb{x_0}$ within a layer~$L$, such that $\lb{x_0}$ is not part of a non-singleton cluster, generating a new layer $\hat L$. This is done by two distinct constructions (\emph{$\rig\Box$-lifting} and \emph{singleton $\rig\IMP$-lifting}, respectively). For the case in which rule $\rig\IMP$  is applied to a label occurring within a non-singleton cluster, we need a slightly more complicated layer construction (\emph{cluster $\rig\IMP$-lifting}).

The following two constructions are exemplified in Figure~\ref{fig:lifting}. 
	For $m ,h \in \mathbb{N}$, we often abbreviate $1 \leq m \leq h$ into $m = 1..h$.

\begin{construction}[$\rig\BOX$-lifting]
  \label{def:box-lifting}
  Let $\sq G$ be a nice and almost happy sequent, and let $L$ be a layer in $\sq G$ containing an unhappy formula $\whiteff{x_0}{\BOX A}$. We can assume $L =\set{\lb{x_0},\lb{x_1},\ldots,\lb{x_l}}$ for some  $l\geq 0$. 
  Now let $\hL$ be a set of fresh labels $\set{\lb{\hx_0},\lb{\hx_1},\ldots, \lb{\hx_l},\lb{\hy}}$. We define $\lliftfw G{x_0}{\BOX A}$, the \defn{$\rig\BOX$-lifting} of $\sq G$ at $\whiteff{x_0}{\BOX A}$, to be the sequent containing the following formulas: 
  \begin{enumerate}
  \item
    for every $i=0..l$ and $i'=0..l$ the relational atom
    $\accs{\hx_i}{\hx_{i'}}$ whenever $\accsq{x_i}{x_{i'}}$;
  \item
    for every $i=0..l$:
    \begin{enumerate}
    \item for each  $\lb w \in \labelsof{G}$ the
      relational atom $\futs {w}{\hx_i}$ whenever $\futsq {w}{x_i}$;
    \item for each formula $\fm C$,
      the labelled formula
      $\labelsb{\hx_i}{C}$ whenever $\blackf{x_i}{C} G$;
    \item the relational atom $\accs{\hx_i}{\hy}$ whenever $\accsq{x_i}{x_0}$;
    \end{enumerate}
 
  \item
    for each $\lb\hx\in\hL$ the relational atom $\futs {\hx}{\hx}$; and
  \item the labelled formula $\labelsw \hy {A}$.
  \item the relational atom $\accs{\hy}{\hy}$;
  \end{enumerate}
  We can now define the following macro rule:
  \begin{equation}
    \label{eq:lift-box}
    \vlinf{\liftrbox}{
    }{
      \sq G, \labelsw x{\BOX A}}{
      \sq G,  \labelsw x{\BOX A}\cup\lliftfw {G} x{\BOX  A}
    } 
  \end{equation}
  where $\hL=\set{\lb{\hx_0},\lb{\hx_1},\ldots, \lb{\hx_l},\lb{\hy}}$ is a new layer added to the sequent.
  We say that  $\lb{\hy}$ is the \defn{suricata label of the layer $\hL$}, and the label $\lb {\hx_0}$ is the \defn{parent} of $\lb\hy$. 
\end{construction}

\begin{construction}[singleton $\rig\IMP$-lifting]
  \label{def:imp-lifting-singleton}
  Let $\sq G$ be a nice and almost happy sequent, and let $L$~be a layer in $\sq G$ containing an unhappy formula $\labelsw{x_0}{ A \IMP B}$ such that $\lb{x_0}$ forms a singleton cluster. We have $L =\set{\lb{x_0},\lb{x_1},\ldots,\lb{x_l}}$ for some  $l\geq 0$. 
  Now let $\hL$ be a set of fresh labels $\set{\lb{\hx_0},\lb{\hx_1},\ldots, \lb{\hx_l}}$. We define $\lliftfw G{x_0}{ A \IMP B}$, the \defn{$\rig\IMP$-lifting} of $\sq G$ at $\whiteff{x_0}{ A \IMP B}$, to be the sequent containing the following formulas: 
  \begin{enumerate}
  \item
    for every $i=0..l$ and $i'=0..l$ the relational atom
    $\accs{\hx_i}{\hx_{i'}}$ whenever $\accsq{x_i}{x_{i'}}$;
  \item
    for every $i=0..l$:
    \begin{enumerate}
    \item for each  $\lb w \in \labelsof{G}$ the
      relational atom $\futs {w}{\hx_i}$ whenever $\futsq {w}{x_i}$;
    \item for each formula $\fm C$,
      the labelled formula
      $\labelsb{\hx_i}{C}$ whenever $\blackf{x_i}{C} G$;
    \end{enumerate}
  \item
    for each $\lb\hx\in\hL$ the relational atom $\futs {\hx}{\hx}$; and
  \item the labelled formulas $\labelsb {\hx_0} {A}$ and
    $\labelsw {\hx_0} {B}$. 
  \end{enumerate}
  We can define the following macro rule:
  \begin{equation}
    \label{eq:lift-imp}
    \vlinf{\liftrim}{
    }{
      \sq G, \labelsw x{A \IMP B}}{
      \sq G,  \labelsw x{A \IMP B}\cup\lliftfw {G} x{A\IMP B}
    } 
  \end{equation}
  where $\hL=\set{\lb{\hx_0},\lb{\hx_1},\ldots, \lb{\hx_l}}$ is a new layer added to the sequent, and we say that $\lb{\hx_0}$ is the \defn{suricata label of $\hL$}. 
\end{construction}


\begin{figure}[!t]
	\begin{minipage}{0.45\textwidth}
			\begin{tikzpicture}[thick, every node/.style={scale=1.2},font = {\large}]
		
		\tikzstyle{node}=[circle,fill=black,inner sep=1.5pt]
		\tikzstyle{nonode}=[inner sep=0pt]
		\tikzstyle{access}=[blue]
		\tikzstyle{future}=[dashed,->]
		
		\node[] (Lh) at (1,5) {$\hat L$};
		\node[] (L) at (1,2) {$L$};
		
		\node[label=below:{$x_0$}] (x0) at (6,2) [node] {};
		\node[label=below:{$x_1$}] (x1) at (2,2) [node] {};
		\node[label=below:{$x_2$}] (x2) at (3,2) [node] {};
		\node[label=right:{$x_3$}] (x3) at (4,2.7) [node] {};
		\node[label=below:{$x_4$}] (x4) at (8,2) [node] {};
		\node[label=below:{$x_5$}] (x5) at (9,2.7) [node] {};
		\node[label=below:{$x_6$}] (x6) at (9.5,1.3) [node] {};
		%
		
		\node[label=above:{$\hat x_0$}] (x0') at (6,5) [node] {};
		\node[label=above:{$\hat  y$}] (y') at (7,5.7) [node] {};
		\node[label=above:{$\hat x_1$}] (x1') at (2,5) [node] {};
		\node[label=above:{$\hat x_2$}] (x2') at (3,5) [node] {};
		\node[label=right:{$\hat  x_3$}] (x3') at (4,5.7) [node] {};
		\node[label=above:{$\hat  x_4$}] (x4') at (8,5) [node] {};
		\node[label=above:{$\hat x_5$}] (x5') at (9,5.7) [node] {};
		\node[label=above:{$\hat x_6$}] (x6') at (9.5,4.3) [node] {};
		%

		\draw[access] (x1) -- (x2);
		\draw[access] (x2) -- (x0);
		\draw[access] (x2) -- (x3);
		\draw[access] (x0) -- (x4);
		\draw[access] (x4) -- (x5);
		\draw[access] (x4) -- (x6);
		
		\draw[access] (x1') -- (x2');
		\draw[access] (x2') -- (x0');
		\draw[access] (x2') -- (x3');
		\draw[access] (x0') -- (x4');
		\draw[access] (x0') -- (y');
		\draw[access] (x4') -- (x5');
		\draw[access] (x4') -- (x6');
		
		\draw[future] (x1) -- (x1');
		\draw[future] (x2) -- (x2');
		\draw[future] (x3) -- (x3');
		\draw[future] (x0) -- (x0');
		\draw[future] (x4) -- (x4');
		\draw[future] (x5) -- (x5');
		\draw[future] (x6) -- (x6');

	\end{tikzpicture}
	
	\end{minipage}
	\begin{minipage}{0.5\textwidth}
		\begin{center}
				\begin{tikzpicture}[thick, every node/.style={scale=1.2},font = {\large}]
		
		\tikzstyle{node}=[circle,fill=black,inner sep=1.5pt]
		\tikzstyle{nonode}=[inner sep=0pt]
		\tikzstyle{access}=[blue]
		\tikzstyle{future}=[dashed,->]
		
		\node[] (Lh) at (1,5) {$\hat L$};
		\node[] (L) at (1,2) {$L$};
		
		\node[label=below:{$x_0$}] (x0) at (6,2) [node] {};
		\node[label=below:{$x_1$}] (x1) at (2,2) [node] {};
		\node[label=below:{$x_2$}] (x2) at (3,2) [node] {};
		\node[label=right:{$x_3$}] (x3) at (4,2.7) [node] {};
		\node[label=below:{$x_4$}] (x4) at (8,2) [node] {};
		\node[label=below:{$x_5$}] (x5) at (9,2.7) [node] {};
		\node[label=below:{$x_6$}] (x6) at (9.5,1.3) [node] {};
		%
		
		\node[label=above:{$\hat x_0$}] (x0') at (6,5) [node] {};
		\node[label=above:{$\hat x_1$}] (x1') at (2,5) [node] {};
		\node[label=above:{$\hat x_2$}] (x2') at (3,5) [node] {};
		\node[label=right:{$\hat  x_3$}] (x3') at (4,5.7) [node] {};
		\node[label=above:{$\hat  x_4$}] (x4') at (8,5) [node] {};
		\node[label=above:{$\hat x_5$}] (x5') at (9,5.7) [node] {};
		\node[label=above:{$\hat x_6$}] (x6') at (9.5,4.3) [node] {};
		%

		\draw[access] (x1) -- (x2);
		\draw[access] (x2) -- (x0);
		\draw[access] (x2) -- (x3);
		\draw[access] (x0) -- (x4);
		\draw[access] (x4) -- (x5);
		\draw[access] (x4) -- (x6);
		
		\draw[access] (x1') -- (x2');
		\draw[access] (x2') -- (x0');
		\draw[access] (x2') -- (x3');
		\draw[access] (x0') -- (x4');'
		\draw[access] (x4') -- (x5');
		\draw[access] (x4') -- (x6');
		
		\draw[future] (x1) -- (x1');
		\draw[future] (x2) -- (x2');
		\draw[future] (x3) -- (x3');
		\draw[future] (x0) -- (x0');
		\draw[future] (x4) -- (x4');
		\draw[future] (x5) -- (x5');
		\draw[future] (x6) -- (x6');

	\end{tikzpicture}
	
		\end{center}
	\end{minipage}
	\caption{\textbf{Left}: $\rig\Box$-lifting of the layer $L$. \textbf{Right}: Singleton $\rig\IMP$-lifting of the layer $L$. 
	}
	\label{fig:lifting}
\end{figure}


If the label $\lb{x_0}$ is in a non-singleton cluster, the lifting is a bit more complicated. We need to duplicate the cluster of $\lb{x_0}$, as indicated in Figure~\ref{fig:mm-lifting}. The reason is that the suricata label cannot be in a cluster. Otherwise the rule would be unsound. This is formally defined below.


\begin{figure}[!t]
  \begin{center}
    	\begin{tikzpicture}[thick, every node/.style={scale=1.2},font = {\large}]

		\tikzstyle{node}=[circle,fill=black,inner sep=1.5pt]
		\tikzstyle{nonode}=[inner sep=0pt]
		\tikzstyle{access}=[blue]
		\tikzstyle{future}=[dashed,->]
		
		\node[] (Lh) at (-1,6) {$\hat L$};
		\node[] (L) at (-1,2) {$L$};
		
		\node[label=below:{$x_0$}] (x2) at (5.8,2.95) [node] {};
		\node[label=below:{$x_1$}] (x1) at (5.5,1.15) [node] {};
		\node[label=below:{$x_2$}] (x3) at (6.8,1.4) [node] {};
		\node[label=below:{$y_1$}] (y1) at (2,2) [node] {};
		\node[label=below:{$y_2$}] (y2) at (3,2) [node] {};
		\node[label=left:{$y_3$}] (y3) at (4,2.7) [node] {};
		\node[label=below:{$y_4$}] (y4) at (8,2) [node] {};
		\node[label=below:{$y_5$}] (y5) at (9,2.7) [node] {};
		\node[label=below:{$y_6$}] (y6) at (9.5,1.3) [node] {};
		%

		\node[label=above:{$\hat x_1'$}] (x1') at (3.5,5.15) [node] {};
		\node[label=above:{$\hat x_1''$}] (x1'') at (7.5,5.15) [node] {};
		\node[label=above:{$\hat x_0'$}] (x2') at (3.8,6.95) [node] {};
		\node[label=above:{$\hat x_0$}] (xhat) at (5.8,6) [node] {};
		\node[label=above:{$\hat x_0''$}](x2'') at (7.8,6.95) [node] {};
		\node[label=above:{$\hat x_2'\:$}] (x3') at (4.8,5.4) [node] {};
		\node[label=above:{$\hat x_2''\:$}] (x3'') at (8.8,5.4) [node] {};
		%
		\node[label=above:{$\hat y_1$}] (y1h) at (0,6) [node] {};
		\node[label=above:{$\hat y_2$}] (y2h) at (1,6) [node] {};
		\node[label=above:{$\hat y_3$}] (y3h) at (2,6.7) [node] {};
		\node[label=above:{$\hat y_4$}] (y4h) at (10,6) [node] {};
		\node[label=above:{$\hat y_5$}] (y5h) at (11,6.7) [node] {};
		\node[label=above:{$\hat y_6$}] (y6h) at (11.5,5.3) [node] {};
		%
			
		\draw[access] (y1) -- (5,2);
		\draw[access] (y2) -- (y3);
		\draw[access] (6,2) circle [radius=1cm];
		\draw[access] (7,2) -- (y4);
		\draw[access] (y4) -- (y5);
		\draw[access] (y4) -- (y6);
		
		\draw[access] (y1h) -- (3,6);
		\draw[access] (y2h) -- (y3h);
		\draw[access] (4,6) circle [radius=1cm];
		\draw[access] (5,6) -- (7,6);
		\draw[access] (8,6) circle [radius=1cm];
		\draw[access] (9,6) -- (y4h);
		\draw[access] (y4h) -- (y5h);
		\draw[access] (y4h) -- (y6h);
		
		\draw[future] (x1) -- (x1');
		\draw[future] (x1) -- (x1'');
		\draw[future] (x2) -- (x2');
		\draw[future] (x2) -- (xhat);
		\draw[future] (x2) -- (x2'');
		\draw[future] (x3) -- (x3');
		\draw[future] (x3) -- (x3'');
		\draw[future] (y1) -- (y1h);
		\draw[future] (y2) -- (y2h);
		\draw[future] (y3) -- (y3h);
		\draw[future] (y4) -- (y4h);
		\draw[future] (y5) -- (y5h);
		\draw[future] (y6) -- (y6h);

	\end{tikzpicture}
  \end{center}
  \caption{$\rig\IMP $-cluster lifting of the layer $L$.
 }
  \label{fig:mm-lifting}
\end{figure}


\begin{construction}[cluster $\rig\IMP$-lifting]
  \label{def:imp-lifting}
  Let $\sq G$ be a nice and almost happy sequent, and let $L$ be a layer in $\sq G$ containing an unhappy formula $\labelsw{x_0}{ A \IMP B}$ such that $\lb{x_0}$ belongs to a cluster $\lbc{C_x}=\set{\lb{x_0},\ldots,\lb{x_h}}$ with $h\geq 1$. 
    	Let $\set{\lb{y_1},\ldots,\lb{y_l}}=L\setminus \lbc{C_x}$, for $l\geq 0$. 
    	Now let $\hL$ be a set of fresh variables  
    	$ \set{\lb{\hy_1},\ldots, \lb{\hy_l},\lb{\hx_0}, \lb{\hx_0'},\ldots, \lb{\hx_h'},\lb{\hx_0''},\ldots, \lb{\hx_h''}}$. We define $\lliftfw G{x_0}{ A \IMP B}$, the \defn{$\rig\IMP$-cluster lifting} of~$\sq G$ at $\whiteff{x_0}{ A \IMP B}$, to be the sequent containing the following formulas: 
    	\begin{enumerate}
    			\item
    		for every $i,i'=1..l$ and for every $j,j'=0..h$:
    		\begin{enumerate}
    			\item 
    			relational atom $\accs{\hy_i}{\hy_{i'}}$ whenever $\accsq{y_i}{y_{i'}}$;
    			\item 
    			relational atoms $\accs{\hx'_j}{\hx'_{j'}}$ and $\accs{\hx'_j}{\hx''_{j'}}$ and $\accs{\hx''_j}{\hx''_{j'}}$; 
    			\item 
    			relational atoms $\accs{\hy_i}{\hx_{j}'}$ and
    			$\accs{\hy_i}{\hx''_{j}}$ whenever $\accsq{y_i}{x_{j}}$;
   				\item 
   				relational atom $\accs{\hy_i}{\hx_0}$ whenever $\accsq{y_i}{x_0}$;
    			\item 
    			relational atoms $\accs{\hx'_j}{\hy_i}$ and
    			$\accs{\hx''_j}{\hy_i}$ whenever $\accsq{x_j}{y_i}$;
    			\item 
   			relational atom $\accs{\hx_0}{\hy_i}$ whenever $\accsq{x_0}{y_i}$; 
    			\item 
    			relational atoms  $\accs{\hx'_j}{\hx_0}$ and 
    			$\accs{\hx_0}{\hx''_{j}}$;
    		\end{enumerate}
    		\item for every $i=1..l$, for every $j=0..h$, and for every  $\lb w \in \labelsof{G}$:
    	\begin{enumerate}
    		\item  relational atom $\futs {w}{\hy_i}$ whenever $\futsq {w}{y_i}$;
    		\item
    		relational atom $\futs {w}{\hx_0}$ whenever $\futsq {w}{x_0}$;
    		\item	 relational atoms $\futs {w}{\hx_j'}$ and $\futs
    		{w}{\hx_j''}$ whenever $\futsq {w}{x_j}$;
    	\end{enumerate}
    		\item 
    	for every $i=1..l$, for every $j=0..h$,  and  for every formula~$\fm C$:
    	\begin{enumerate}
    		\item
    		labelled formula $\labelsb{\hy_i}{C}$ whenever $\blackf
    		{y_i}{C} G$;
    		\item	labelled formulas $\labelsb{\hx'_j}{C} $ and $\labelsb{\hx''_j}{C}$ whenever $\blackf {x_j}{C} G$;
    		\item labelled formulas $\labelsb{\hx_0}{C}$ whenever $\blackf {x_0}{C} G$;
    	\end{enumerate}
    		\item
    		for every $\lb\hv\in\hL$ the relational atom $\futs {\hv}{\hv}$; and
    		\item the labelled formulas $\labelsb {\hx_0} {A}$ and $\labelsw{ \hx_0} {B}$, and the relational atom $\accs{\hx_0}{\hx_0}$. 
    	\end{enumerate}
        The inference rule $\liftrim$ is as shown in~\eqref{eq:lift-imp} above, where $\hL= \set{\lb{\hy_1},\ldots, \lb{\hy_l},\lb{\hx_0}, \lb{\hx_0'},\ldots, \lb{\hx_h'},\lb{\hx_0''},\ldots, \lb{\hx_h''}}$ is the newly added layer, and we say that  $\lb{\hx_0}$  is the \defn{suricata label of $\hL$}.
    \end{construction}
        
\begin{definition}[Proper sequent]
	A sequent is~$\sq{G}$ is called \defn{proper} if{f} it is nice and all its clusters  are singletons. 
\end{definition}

\begin{lemma}\label{prop:liftr}
  An instance of the $\liftrim$-rule or the $\liftrbox$-rule is derivable in $\labISfs$ if its conclusion is proper. 
\end{lemma}
\begin{proof}
  Any application of the $\liftr$ rule to a formula $\whiteff{x}{F} \in \sq G$ can be simulated by the rules of $\labISfs$: it suffices to apply to $\sq G$ one instance of rule $\rig{\IMP}$ or $\rig\BOX$ (depending on the shape of $\fm{F}$), followed by possibly multiple applications of rules $\fone,\ftwo$,  $\Lref, \Ltr$, $\Rref$, $\Rtr$, and $\monl$.
\end{proof}

\begin{lemma}\label{lem:liftr-nice}
  If the conclusion of a $\liftrim$- or $\liftrbox$-rule instance is nice, then so is its premise.
\end{lemma}

\begin{proof}
  The new layer forms a new leaf in the tree of layers of the sequent.
  Structural happiness follows immediately by the construction. Similarly for the tree structure of the new layer. 
\end{proof}


\subsection{Lifting with Loops}
\label{sec:loops}

After layer saturation, the only unhappy formulas in sequent are the $\rig\IMP$- and the $\rig\BOX$-formulas. We use the lifting of layers, as implemented by the rules $\liftrbox$ and $\liftrim$, defined in~\eqref{eq:lift-box} and~\eqref{eq:lift-imp}, respectively.  We stop creating new layers, when we reach a layer that can be simulated by some past layer.

\begin{definition}[Simulation of layers]\label{def:simul-layers}
  Let $L'$ and $L$ be layers in a sequent~$\sq G$ (Def~\ref{def:layer}).
  A \defn{layer simulation} between~$L'$~and~$L$ is a non-empty binary relation $\simul \subseteq (L' \times L)\;\cap\sim$ such that for all $\lb{x'} \in L'$, and all $\lb x,\lb y \in L$,
  \begin{enumerate}
  \item [($\rn{S1}$)]   whenever $\lb {x'} \simul \lb x$ and $\accsq{x}{y}$, there exists $\lb{y'} \in L'$ such that 	$\accsq{x'}{y'}$ and $\lb {y'} \simul \lb y$, and 
  \item [($\rn{S2}$)]   whenever $\lb {x'} \simul \lb x$ and $\accsq{y}{x}$, there exists $\lb{y'} \in L'$ such that 	$\accsq{y'}{x'}$ and $\lb {y'} \simul \lb y$. 
	\end{enumerate}
  We say $L'$ \defn{simulates} $L$ if{f} there is a layer simulation $\simul$ between~$L'$~and~$L$.
  And we say that a layer $L$  \defn{is simulated} if{f} there is a layer~$L'$ in~$\sq G$, such that $L'< L$ and $L'$ simulates~$L$.
\end{definition}

\begin{proposition}
  \label{prop:sim_bij}
  Let $L'$ and $L$ be finite layers in a sequent\/~$\sq G$. If\/ $L'$ simulates $L$ then, for all $\lb x\in L $, there is a $\lb {x'}\in L'$ such that $\lb{x'}\simul\lb x$. %
\end{proposition}

\begin{proof}
	There is a layer simulation $\simul$ between~$L'$~and~$L$. 
	As $\simul \not= \emptyset$ there is $\lb x\in L$ and $\lb{x'}\in L'$ such that $\lb {x'} \simul \lb{x}$.
	Take $\lb y\in L$ such that $\lb x \grel G \lb y$, then by ($\rn{S1}$) we have that there is $\lb{y'}$ such that $\lb{y'}\simul\lb{y}$.
	Similarly, take $\lb y\in L$ such that $\lb x \grel{G}^{-1} \lb y$, then by ($\rn{S2}$) we have that there is $\lb{y'}$ such that $\lb{y'}\simul\lb{y}$.
	As $L$ is an equivalence class of $\lrel G$, i.e.~the transitive and reflexive closure of $\grel G \cup \grel{G}^{-1}$, we have by induction that for every $\lb z\in L$ there is $\lb{z'}\in L'$ such that $\lb{z'}\simul\lb z$.
\end{proof}

\begin{definition}[Lift of a sequent]\label{def:lift}
  Let $\sq G$ be a nice sequent that is almost happy and not axiomatic. Let $\unhappyof G$ be the set of unhappy labeled formulas $\labels xA$ in $\sq G$, such that $\layerof x$ (Def.~\ref{def:layer-of-label}) is not simulated (Def.~\ref{def:simul-layers}).
  We define the \defn{lift} of $\sq G$ to be the set
  \begin{equation}
    \label{eq:lift}
    \liftof G =\bigcup_{\labels xA\in\unhappyof G}\lliftf G{x}{A}
  \end{equation}
\end{definition}

\begin{remark}
  The lift $\liftof G$ is well-defined because in every sequent that is almost happy and not axiomatic, the only unhappy labeled formulas are of shape $\labelsw x{\BOX A}$ and $\labelsw x{A\IMP B}$, for which the lifting has been defined in Constructions~\ref{def:box-lifting}, \ref{def:imp-lifting-singleton}, and~\ref{def:imp-lifting}.
\end{remark}

\begin{figure}
  \begin{center}
  \includegraphics{figures/proof-search-gen.tex}
  \end{center}
  \caption{Proof search tree, alternating saturation and lifting steps}
  \label{fig:searchtree}
\end{figure}
The idea is to apply layer saturation again to the sequent $\sq G\cup\liftof G$ after lifting, and so on. The problem is, that the layers can grow and we might never come to an end. Therefore we need to create more loops, like we did for the $\DIA$ in~\eqref{eq:dialoop}.

For defining the correct loop conditions, we need to introduce a few more concepts. The basic idea of our proof search algorithm (to be spelled out in Section~\ref{sec:algorithm} below) is that we alternate between layer saturation (as described in Algorithm~\ref{alg:sat}) and lifting of layers. This lifting creates new layers which then will be saturated again, and so on. This creates a tree of growing sequents, as indicated in Figure~\ref{fig:searchtree}. In particular, whenever we create new layers via lifting (as in Definition~\ref{def:lift}), all layers below are saturated. However, for checking for a loop, we need to compare labels right after lifting, before saturation.

This motivates the following notation:

\begin{notation}
  Let $\sq G$ and $\sq H$ be sequents in a search tree as indicated in Figure~\ref{fig:searchtree}. 
  We write $\sq H\psbelow\sq G$ if $\sq H$ is below $\sq G$ on the path from the root to $\sq G$ in the search tree. In that case we have $\labelsof H\subseteq\labelsof G$. For disambiguation, we write $\lbG xH\lbeq\lbG yG$ if for all formulas $\fm A$, we have $\labelsb xA\in\sq H$ (resp.~$\labelsw xA\in\sq H$) iff $\labelsb yA\in\sq G$ (resp.~$\labelsw yA \in\sq G$). 
\end{notation}

 \begin{notation}[Equivalent clusters]\label{def:eqclusters}
  For two clusters $\lbc{C_1}$ and $\lbc{C_2}$, we write $\clsubG{C_1}H{C_2}G$ if   for every $\lb{z_1}\in\lbc{C_1}$ there is a $\lb{z_2}\in\lbc{C_2}$ with $\lbG{z_1}H\lbeq\lbG{z_2}G$. We write $\lbG{C_1}H\lbeq\lbG{C_2}G$ iff $\clsubG{C_1}H{C_2}G$ and $\clsubG{C_2}G{C_1}H$. Observe that $\lbeq$ is an equivalence relation on clusters, and that $\lbc{C_1}$ and $\lbc{C_2}$ do not need to be in the same sequent.
\end{notation}

Recall that in a nice sequent, every layer is a tree of clusters. This allows us to have the following definitions which give us the leaves of those trees, and embeddings for these layer trees.

\begin{definition}[Leaf]
  A label $\lb x$ in a sequent is called a \defn{leaf} if its cluster $\lbc{C_x}$ is a leaf of layer $\layerof x$ seen as a tree of clusters.
\end{definition}

\begin{definition}[Layer Embedding]\label{def:t-embedding}
  Let $\sq H$ and $\sq G$ be nice sequents, and let $L_1\subseteq \sq H$ and  $L_2\subseteq \sq G$ be layers in those sequents. An injective function $\emb\colon L_1\to L_2$ is a \defn{(layer) embedding} iff 
  \begin{enumerate}
  \item   for all $\lb x\in L_1$, we have $\lb x\lbeq\embof x$ and $\lbG{C_x}H\lbeq\lbG{C_{\embof{x}}}G$, and 
  \item for all $\lb x,\lb y\in L_1$, we have $\accsqq[H]{x}{y}$ iff $\accsq{{\embof{x}}}{{\embof{y}}}$. 
  \end{enumerate}
\end{definition}

The important property of this embedding is that it is also an embedding for the ``layers-as-trees-of-clusters'' structure.

\begin{lemma}[Layer-tree embedding]\label{lem:t-embedding}
  Let $\sq H$ and $\sq G$ be nice sequents, let $L_1\subseteq \sq H$ and  $L_2\subseteq \sq G$ be layers in those sequents, and let $\emb\colon L_1\to L_2$ be an embedding. Then for all $\lb x,\lb y\in L_1$, we have $\accsqq[H]{C_x}{C_y}$ iff $\accsq{C_{\embof{x}}}{C_{\embof{y}}}$. 
\end{lemma}

\begin{proof}
  This follows immediately from Definition~\ref{def:t-embedding}.
\end{proof}

\begin{remark}
  To be technically precise, such an embedding on the ``layers-as-trees-of-clusters'' structure is not a tree-embedding because of the transitivity of $\srel{G}$. It is a more general graph embedding: whenever there is an edge between $\lbc{C_u}$ and $\lbc{C_v}$, then there is a path between $\embof{C_u}$ and $\embof{C_v}$. But this is exactly the kind of embedding to which Kruskal's tree theorem can be applied, as we will do in Section~\ref{sec:termination}.
\end{remark}

\begin{definition}[Tree of a label and their embedding]\label{def:tree-of-label}
  Let $\sq G$ be a sequent and $\lb x\in\labelsof G$ be a singleton cluster. We can define the \defn{tree of $\lb x$}, denoted by $\treeof x$ to be the subtree of $\layerof x$ (Def.~\ref{def:layer-of-label}) that is rooted at $\lb x$. Given two such labels $\lb x\in\labelsof H$ and $\lb y\in\labelsof G$, an \defn{embedding} from  $\treeof x$ into $\treeof y$ is an injective function $\emb\colon\treeof x\to\treeof y$ obeying the conditions of Definition~\ref{def:t-embedding} above.
\end{definition}

\begin{remark}
  Whenever we use the embeddings defined in Definition~\ref{def:tree-of-label}, the $\lb x$ and $\lb y$ will be the suricata labels of freshly lifted layers, i.e., they will contain the only output formula (i.e., with $\rigsym$-polarity) of their respective layer. Hence, it follows that $\embof{x}=\lb y$. And for this reason we did not include that condition in the definition. 
\end{remark}

We can now generalise the notion of repetition:

\begin{definition}[Repetition]\label{def:repetition}
  Let $\sq G$ be a sequent that is the result of a lifting step. Let $\lb x,\lb y\in\labelsof G$ such that $\lb y$ does not have a future. We say that $\lb y$ is a \defn{repetition} of $\lb x$ if $\lb x\neq \lb {y}$ and  $\vissq xy$ and there is a sequent $\sq H\psbelow\sq G$ such that $\sq H$ is the result of a lifting step and such that $\lb x\in\labelsof H$ and $\lb x$ does not have a future in $\sq H$ and $\lbG xH\lbeq \lbG yG$. We call this $\sq H$ a \defn{seed} of the repetition.
  \begin{itemize} 
  \item A repetition is a \defn{leaf repetition} if additionally $\lb y$ is a leaf of $\sq G$, and $\lb x$ is a leaf of $\sq H$, and $\clsubG{C_x}H{C_y}G$ and $\nfutsq{C_x}{C_y}$.
  \item A repetition is a \defn{$\BOX$-suricata repetition} if additionally $\lb x$ is the suricata of $\layerof x$ in $\sq H$ and $\lb y$ is the suricata of $\layerof y$ in $\sq G$, both coming from a $\liftrbox$-rule application (Construction~\ref{def:box-lifting}). 
  \item A repetition is a \defn{$\IMP$-suricata repetition}  if additionally $\lb x$ is the suricata of $\layerof x$ in $\sq H$ and $\lb y$ is the suricata of $\layerof y$ in $\sq G$, both coming from a $\liftrim$-rule application (Constructions~\ref{def:imp-lifting-singleton}, and~\ref{def:imp-lifting}), and additionally there is an embedding of $\treeof{x}$ in $\sq H$ into $\treeof{y}$ in $\sq G$.
  \item A repetition is a \defn{suricata repetition} if it is a $\BOX$-suricata repetition or a $\IMP$-suricata repetition.
  \end{itemize}
\end{definition}

\begin{definition}[$\LEAF$-repetition pair and $\leafloopr$]
	\label{def:leafloop}
  Let $\sq G$ be a nice sequent and let $\lb{x'},\lb y\in\labelsof G$ with $\accsq{x'}y$ and $\lb{x'}\notin \lbc{C_y}$. We say that  $\tuple{\lb {x'},\lb y}$ is a \defn{leaf($\LEAF$)-repetition pair} if there is $\lb x\in \labelsof G$ with $\futsq x{x'}$ such that $\lb y$ is a leaf repetition of $\lb x$, and the suricata $\lb{s_y}$ of $\layerof y$ is a suricata repetition of the suricata $\lb{s_x}$ of $\layerof x$, such that both repetitions have the same seed $\sq H\psbelow\sq G$ (Def.~\ref{def:repetition}) and we do \emph{not} have that $\lb{x'}\srel G\lb {s_y}\srel G\lb y$. With this we define the following loop rule:
  \begin{equation}
    \label{eq:leafloop}
    \vlinf{\leafloopr}{
      \proviso{where $\tuple{\lb {x'},\lb y}$ is a leaf repetition pair in $\sq G$}}{
      \sq G}{
      \sq G \labelsubst{y}{x'}\cup\trans{\interval {x'}y}
    } 
  \end{equation}
\end{definition}

\begin{definition}[$\BOX$-repetition pair and $\boxloopr$]\label{def:boxloop}
  Let $\sq G$ be a nice sequent and $\lb{x'},\lb y\in\labelsof G$ with $\accsq{x'}y$ and $\lb{x'}\notin \lbc{C_y}$. We say that  $\tuple{\lb {x'},\lb y}$ is a \defn{$\BOX$-repetition pair} if there is $\lb x\in \labelsof G$ with $\futsq x{x'}$, such that $\lb y$ is
  a $\BOX$-suricata repetition of  $\lb{x}$. With this we define the following loop rule:
  \begin{equation}
    \label{eq:boxloop}
    \vlinf{\boxloopr}{
      \proviso{where $\tuple{\lb {x'},\lb y}$ is a $\BOX$-repetition pair in $\sq G$}}{
      \sq G}{
      \sq G 
      \cup \trans{\interval {x'}y}
    } 
  \end{equation}
\end{definition}
Observe that Definition~\ref{def:repetition} ensures that $\lb y$ is a repetition of $\lb x$ in Definition~\ref{def:boxloop}.

\begin{definition}[$\IMP$-repetition pair and $\imploopr$]\label{def:imploop}
  Let $\sq G$ be a nice sequent and $\lb{x'},\lb y\in\labelsof G$ with $\accsq{x'}y$ and $\lb{x'}\notin \lbc{C_y}$. We say that  $\tuple{\lb {x'},\lb y}$ is a \defn{$\IMP$-repetition pair} if there is $\lb x\in \labelsof G$ with $\futsq x{x'}$, such that $\lb y$ is a $\IMP$-suricata repetition of $\lb x$.
  With this we define the following loop rule:
  \begin{equation}
    \label{eq:imploop}
    \vlinf{\imploopr}{
      \proviso{where $\tuple{\lb {x'},\lb y}$ is a $\IMP$-repetition pair in $\sq G$}}{
      \sq G}{
      \sq G\cup\trans{\interval {x'}y}
    } 
  \end{equation}
\end{definition}
\begin{remark}
  The $\boxloopr$ and $\imploopr$-rules do not follow the same pattern as the other $\looprule$-rules because we cannot identify $\lb y$ and $\lb{x'}$, as $\lb {y}$ contains a $\rigmark$-formula and $\lb{x'}$ does not. This is similar to the cluster $\rig\IMP$-lifting in Construction~\ref{def:imp-lifting}.
\end{remark}

\begin{definition}[$\TRI$-repetition pair and $\triangleloopr$]\label{def:triangleloop}
  Let $\sq G$ be a nice sequent and let $\lb{x'},\lb y\in\labelsof G$ with $\accsq{x'}y$ and $\lb{x'}\notin \lbc{C_y}$. We say that  $\tuple{\lb {x'},\lb y}$ is a \defn{triangle($\TRI$)-repetition pair} if there is $\lb x\in \labelsof G$ with $\futsq x{x'}$ and $\futsq xy$  such that $\lb {x'}$ and $\lb y$ are both repetitions of $\lb x$ (Def.~\ref{def:repetition}), and the suricata $\lb{s_y}$ of $\layerof y$ is a suricata repetition of the suricata $\lb{s_x}$ of $\layerof x$, such that all three repetition have the same seed $\sq H\psbelow\sq G$ (Def.~\ref{def:repetition}), and such that we do \emph{not} have that $\lb{x'}\srel G\lb {s_y}\srel G\lb y$. 
  With this we define the following loop rule:
  \begin{equation}
    \label{eq:triangleloop}
    \vlinf{\triangleloopr}{
      \proviso{where $\tuple{\lb {x'},\lb y}$ is a triangle repetition pair in $\sq G$}}{
      \sq G}{
      \sq G \labelsubst{y}{x'}\cup\trans{\interval {x'}y}
    } 
  \end{equation}
\end{definition}

\begin{remark}\label{rem:suricata-in-loop}
  The condition that the suricata $\lb {s_y}$ must not occur between $\lb{x'}$ and $\lb y$, as stated in Definition~\ref{def:leafloop} for a leaf repetition pair and in Definition~\ref{def:triangleloop} for a triangle repetition pair, is also crucial for $\BOX$-repetition pairs and $\IMP$-repetition pairs. But we do not need to state it in Definition~\ref{def:boxloop} and~\ref{def:imploop} because it follows from the other conditions.
\end{remark}

\begin{definition}[Twins and $\shrinkr$]\label{def:shrink}
  Let $\sq G$ be a nice sequent, and let $\lb{x},\lb y\in\labelsof G$ with $\lb x\neq\lb y$.
  We say that $\lb y$ is an \defn{immediate child} of $\lb x$ if $\accsq xy$ and for all $\lb z$ with   $\accsq xz \srel G \lb y$, we have either $\accsq zx$ or $\accsq yz$. 
  We say that $\lb x$ and $\lb y$ are \defn{siblings} if either 
  \begin{itemize}
  \item
    they are in the same cluster, or
  \item there is a $\lb z$ such that both $\lb x$ and $\lb y$ are immediate children of $\lb z$.  
  \end{itemize}
  Finally, we say that  $\lb x$ and $\lb y$  are \defn{twins} if they are siblings and additionally $\lb x\lbeq \lb y$. We also define the following  macro rule:
  \begin{equation}
    \label{eq:shrink}
    \vlinf{\shrinkr}{
      \proviso{where $\lb {x}$ and $\lb y$ are twins in $\sq G$}}{
      \sq G}{
      \sq G \labelsubst{y}{x}
    } 
  \end{equation}
\end{definition}



\begin{algorithm}[Full Lifting]\label{alg:liftsat}
  For a non-axiomatic almost happy (Def.~\ref{def:almost-happy}) sequent $\sq G$ we construct its \defn{full lifting} $\liftssatof G$ by constructing a sequence $\sq H_0,\sq H_1,\sq H_2,\ldots$ as follows:
  \begin{enumerate}[(1)]
  \item\label{i:lift} Let $\sq H_0= \sq G \mathop{\cup} \mathord{\liftof G}$ (Def.~\ref{def:lift})
  \item\label{l:imploop} If there is a $\BOX$- or $\IMP$-repetition pair $\tuple{\lb{x'},\lb y}$ in $\sq H_i$, such that $\sq H_{i}\neq\sq H_i\cup\trans{\interval {x'}y}$\\ then let $\sq H_{i+1}=\sq H_i\cup\trans{\interval {x'}y}$ and go to \ref{l:imploop}
  \item\label{l:substloop} Otherwise, if there is a leaf- or triangle-repetition pair $\tuple{\lb{x'},\lb y}$ in $\sq H_i$,\\ then let $\sq H_{i+1}=\sq H_i \labelsubst{y}{x'}\cup\trans{\interval {x'}y}$ and go to \ref{l:imploop}
  \item\label{l:twin} Otherwise, if there is a pair of twins $\lb x$ and $\lb y$ in $\sq H_i$,\\ then let $\sq H_{i+1}=\sq H_i \labelsubst{y}{x}$ and go to \ref{l:imploop}
  \item Otherwise let $\liftssatof G=\sq H_i$ and terminate.
  \end{enumerate}
\end{algorithm}

\begin{observation}
  If $\liftof G=\emptyset$ then no lifting is possible (either because there are no unhappy $\labelsw xA$ formulas left, or because all layers with such unhappy formulas are simulated). In that case we have $\liftssatof G=\sq G$.
\end{observation}

\begin{remark}
  We need the check that $\sq H_{i}\neq\sq H_i\cup\trans{\interval {x'}y}$ only in \ref{l:imploop}, because in \ref{l:substloop} and \ref{l:twin} we have that $\lb{x'}\neq\lb y$ (resp.~$\lb x\neq\lb y$), and therefore we automatically have $\sq H_{i+1}\neq\sq H_i$. 
\end{remark}
\begin{lemma}\label{lem:lifting}
  Given a nice, non-axiomatic, almost happy sequent $\sq G$, its full lifting $\liftssatof G$ exists and there is a derivation with premise  $\liftssatof G$ and conclusion $\sq G$, using only the rules $\liftrim$, $\liftrbox$, $\leafloopr$, $\boxloopr$, $\imploopr$, $\triangleloopr$, $\shrinkr$. Furthermore, $\liftssatof G$ contains only subformulas of $\sq G$.
\end{lemma}

\begin{proof}
  The number of unhappy formulas in $\sq G$ is bound by the size of $\sq G$. Hence $\sq H_0$ is well-defined. Then, each of the Steps~\ref{l:imploop}--\ref{l:twin} decreases the size of the sequent, and hence this process terminates. Finally, Step~\ref{i:lift} is just an application of several instances of $\liftrim$ and $\liftrbox$ (one for each unhappy $\labelsw x{\BOX A}$ and  $\labelsw x{A\IMP B}$; Step~\ref{l:imploop} is a instance of $\imploopr$; Step~\ref{l:substloop} is an instance of $\leafloopr$, $\boxloopr$, or $\triangleloopr$; and Step~\ref{l:twin} is an instance of $\shrinkr$.
\end{proof}

\begin{lemma}\label{lem:lifting-nice}
  If $\sq G$ is nice, then so is its full lifting $\liftssatof G$.
\end{lemma}

\begin{proof}
  We already observed in Lemma~\ref{lem:liftr-nice} that  $\liftrim$ and $\liftrbox$ preserve niceness. And it is easy to see that also the rules $\leafloopr$, $\boxloopr$, $\imploopr$, $\triangleloopr$, and $\shrinkr$ do so.
\end{proof}

\begin{remark}
  Given a nice sequent $\sq G$, it is easy to see that $\liftssatof G$ contains only subformulas of $\sq G$. Furthermore, observe that each of the newly created layers has exactly one output formula $\labelsw xA$, and the label $\lb x$ with that output formula is the suricata label of that layer.
\end{remark}

\subsection{The Full Algorithm}
\label{sec:algorithm}

We are now going to combine the two processes of layer saturation and full lifting. We construct a sequence $\sqset S_0,\sqset S_1,\sqset S_2,\ldots$ of sets of nice and almost happy sequents (Definition~\ref{def:nice}), by alternating between layer saturation and layer lifting (see Figure~\ref{fig:searchtree}). One can think of such a set $\sqset S_i$ as the set of leaves of the current (incomplete) proof search tree after the saturation steps.
More precisely, for a given sequent $\sq G$, we compute its layer saturation $\ssatof G$, and then compute the full lifting $\liftssatof H$ of a sequent $\sq H\in\ssatof G$. For the sequent  $\liftssatof H$ we compute again the layer saturation, and so on, until we either have only axiomatic sequents, or we find a sequent from which we cannot make any progress anymore.


\begin{algorithm}[Proof search]\label{alg:proofsearch}
  For a given formula $\fm F$, we perform \defn{proof search} by constructing a sequence $\sqset S_0,\sqset S_1,\sqset S_2,\ldots$ of sets of nice sequents.
  \begin{enumerate}[(1)]
  \item For the formula $\fm F$, we define the sequent $\sq G(\fm F)=\futs rr, \accs rr,\labelsw rF$, and let  $\sqset S_0 \colonequals\satof{\sq G(\fm F)}$. 
  \item\label{l:ax} If all sequents in $\sqset S_i$ are axiomatic, then  terminate.\\
    The formula $\fm F$ is provable and we can give a proof of $\labelsw rF$ in $\labISfs$ (see Sect.~\ref{sec:soundness}).
  \item\label{la:lift} Otherwise, pick a non-axiomatic sequent $\sq G_i \in \sqset S_i$ and compute its full lifting $\liftssatof G_i$ (Algo.~\ref{alg:liftsat}). 
  \item \label{alg:block} If $\liftssatof G_i=\sq G_i$, then terminate.\\ 
  The formula $\fm F$ is not provable, and the sequent $\sq G_i$ defines a countermodel (see Sect.~\ref{sec:countermodel}).
  \item Otherwise, we replace $\sq G_i$ in the set by the saturation (Algo.~\ref{alg:sat}) of its full lifting (Algo.~\ref{alg:liftsat}), that is, set $\sqset S_{i+1}\colonequals\bigl(\sqset S_i\setminus\set{\sq G_i}\bigr)\cup\satof{\liftssatof G_i}$, and go to~\ref{l:ax}.
  \end{enumerate}
\end{algorithm}

\begin{observation}
  It is easy to see that indeed at each stage, all elements of a set $\sqset S_i$ are nice and almost happy. Niceness follows from Lemma~\ref{lem:lifting-nice} and Lemma~\ref{lem:ssat-nice}, and that the sequents are almost happy follows from the definition of $\satof{\cdot}$. It follows that Step~\ref{la:lift} is well-defined.
\end{observation}

\begin{figure}[t]
	\includegraphics[scale=0.8]{figures/agi-ex.tex}
	\caption{Proof search on \kurucz, where we use the following abbreviations:\\
		 $\lef{\fm \Delta} = \{ \lef{\fm{ \DIA b \AND \DIA c} },\, \lef{\fm {\DIA b}},\, \lef{\fm{\DIA c}}\}$;\\
		$\lef{\fm \Sigma} = \big\{ \lef{\fm{\BOX \big( a \IMP (\DIA b \AND \DIA c) \big)}},\, \lef{\fm{\BOX \big((d \IMP \bot) \IMP \bot \big)}},\, \lef{\fm{\BOX \big((e \IMP \bot) \IMP \bot \big)}}\big\}$; \\
		${\fm \Pi} = \{\lef{\fm{a \IMP (\DIA b \AND \DIA c)}},\, \lef{\fm{(d \IMP \bot) \IMP \bot}} ,\, \lef{\fm{(e \IMP \bot) \IMP \bot}},\, \rig{\fm{d \IMP \bot}},\, \rig{\fm{e \IMP \bot}}\}$.\\
		Moreover, we take the formulas in~$\fm \Pi $ to be in all the nodes (i.e.,~labelled with all the labels) greater than or equal to~$\lb 1$, as explicitly shown in node $\lb 1$ but omitted  in nodes $\lb 2$--$\lb{25}$. Boldface nodes are suricata nodes, and the nodes in the grey area are added in a layer saturation step. 
	}
	\label{fig:agi-ex}
\end{figure}

\section{Kurucz's counterexample}
\label{sec:agi-ex}

In this section we report a counterexample to our previous decision algorithm for $\ISfour$, published in~\cite{girlando2023intuitionistic}. The counterexample, found by Agi Kurucz, consists of a formula on which the algorithm from~\cite{girlando2023intuitionistic} would run forever. 
We call this formula \emph{\kurucz}. 
We discuss the counterexample here to acknowledge the fault in our decision algorithm from~\cite{girlando2023intuitionistic}, and also to showcase how the new version of the algorithm presented above (Algorithm~\ref{alg:proofsearch}) instead blocks proof search. 

\kurucz is displayed below. We take $\fm G$~to denote its antecedent:
\[
\fm{F} = \overbrace{\fm{ \Big( \BOX \big( a \IMP (\DIA b \AND \DIA c) \big) \AND \BOX \big((d \IMP \bot) \IMP \bot \big)  \AND \BOX \big((e \IMP \bot) \IMP \bot \big)\Big)}}^{\fm {G}} \fm{ \IMP \bot} 
\]

Figure~\ref{fig:agi-ex} depicts the evolution of a sequent, which grows as we move (bottom-up) along a branch of a proof search tree having the (nice and almost happy) sequent $\futs{0}{0},\accs 00,\labelsw{0}{F}$ at the root. 
Proof search follows Algorithm~\ref{alg:proofsearch}, which shares the same basic steps with our faulty algorithm from~\cite{girlando2023intuitionistic}. 
The first rule applied to the sequent is $\liftrim$, which generates a new label~$\lb{1}$, and the sequent $\futs{0}{0},\accs 00, \futs{0}{1}, \futs 11,\accs 11,\labelsw{0}{F},\labelsb{1}{G}, \labelsw{1}{\bot}$. Then, more formulas get labelled with $\lb 1$ within a $\satr$ step, and so on. 

We use colours to identify labels as ``the same'', meaning that they label the same formulas, i.e., $\lb x$ and $\lb y$ share the same colour if{f} $\eqlab{x}{y}$. 
At a first glance, our sequent contains only finitely many colours, which are nevertheless combined in a peculiar way, making it difficult to find exact repetitions between layers, which is what would make proof search terminate at step \ref{alg:block} of Algorithm~\ref{alg:proofsearch}. Indeed, we will need to use our loop rules to `collapse' layers and find exact repetitions. 

The graph represent only \emph{a part} of the sequent as it evolves along a branch --- namely a single $\leq$-branch. Observe that formulas~$\rig{\fm{d \IMP \bot}}$~and~$\rig{\fm{e \IMP \bot}}$ occur in all the nodes greater than or equal to~$\lb 1$. However, in the graph we do not show the result of applying~$\liftrim$ to each of these formulas, but only to some of them (namely, those occurring in nodes~$\lb 0$, $\lb 1$, $\lb 2$, $\lb 7$,~and~$\lb {10}$). 
The respective suricata labels are indicated in boldface in the figure, and they are $\lb 1$, $\lb 2$, $\lb {5}$, $\lb{ 13}$, and $\lb {17}$. 

The part of the sequent showcased in Figure~\ref{fig:agi-ex} is already enough to demonstrate non-termination of the algorithm from~\cite{girlando2023intuitionistic}. The non-termination is due to the interaction of rules $\liftrim$ and $\bdiam$, which can make a sequent ``grow'' (bottom-up) in two different ways. Every application of $\liftrim$ makes a sequent grow ``vertically'', and whenever this happens, we lose all the formulas~$\fmw a$, and in the newly introduced nodes we instead choose~$\lef{\fm \Delta}$, whence formulas~$\fmb{\DIA b}$~and~$\fmb{\DIA c}$ ``reappear'' in the sequent. These formulas make the sequent grow ``horizontally'', as rule~$\bdiam$ can further be applied to them. Our algorithm from~\cite{girlando2023intuitionistic} does not have a loop rule to block such non-terminating behaviour. 

Instead, the loop rules introduced in Section~\ref{sec:loops} restrict the growth of layers. 
In particular, one instance of $\leafloopr$ fires at the topmost layer displayed in Figure~\ref{fig:agi-ex}. 
First observe that to determine whether a $\leafloopr$ is applicable, we consider a sequent at specific stages of its construction --- namely, sequent $\sq H$ in Figure~\ref{fig:agi-ex-rep} (left) right after a lifting step $\liftrim$ has been applied and sequent $\sq G$ in Figure~\ref{fig:agi-ex-rep} (right) right after another lifting step $\liftrim$ has been applied, and in both cases before the layer saturation step begins. All the leaves of the layers, as displayed in Figure~\ref{fig:agi-ex}, are created by applications of $\bdiam$ in the layer saturation step. These are the nodes in the grey area in the figure, 
and, hence, they are absent in the topmost layers of sequents $\sq H$ and $\sq G$ used to verify the applicability of $\leafloopr$.  
With this in mind, note that $\tuple{\lb{18}, \lb{20}}$ is a leaf-repetition pair in $\sq G$ (so $\lb {x'} = \lb{18}$ and $\lb y = \lb{20}$, using the notation from Def.~\ref{def:leafloop}). If we take $ \lb x = \lb 6 $, we have that $\futsq{6}{18}$, and that $\lb{20}$ is a leaf repetition of $\lb{6}$, since $\lb{6}$ is a leaf of its layer in $\sq H$, and $\lb{20}$ is a leaf of its layer in $\sq G$,  and $\eqlab{6}{20}$. Moreover, the suricata $\lb{17}$ of $\layerof{20}$ is a suricata repetition of the suricata $\lb{5}$ of $\layerof{6}$: we define an embedding $\emb$ from $\treeof{5}$ in $\sq H$ to $\treeof{17}$ in $\sq G$ by setting $\embof{\lb 5} \colonequals \lb{17}$, $\embof{\lb 6}\colonequals \lb{20}$, and $\embof{\lb{7}} \colonequals \lb{23}$. Finally,   the suricata $\lb{17}$ of $\layerof{20}$ is in the correct position, as it does \emph{not} hold that $\lb{18} \srel{\sq G} \lb{17} \srel{\sq G} \lb{20}$. 

Thus, rule $\leafloopr$ can be applied to $\sq G$ in  Figure~\ref{fig:agi-ex-rep} (right). In the premise of the rule, label~$\lb{20}$ gets identified with label $\lb{18}$, thus creating a fresh cluster  $\lbc{ C} = \{ \lb{18}, \lb{19}\}$. Next, the algorithm proceeds with the application of the layer saturation step, which would create label $\lb{25}$ but would not create label $\lb{24}$ due to the newly created cluster, and proof search would continue. However, the layer now contains a cluster, which identifies a repeating behaviour, and will eventually prevent the layer from growing further.

\begin{figure}[t]
\begin{minipage}[b]{0.33\textwidth}
	\includegraphics[scale=0.7]{figures/agi-ex-rep-h.tex}
\end{minipage}
\begin{minipage}[b]{0.66\textwidth}
	\includegraphics[scale=0.7]{figures/agi-ex-rep-g.tex}
\end{minipage}
	\caption{Verification of the applicability of $\leafloopr$ to the topmost layer of Figure~\ref{fig:agi-ex}, where we use the following abbreviations:\\
		 $\lef{\fm \Delta} = \{ \lef{\fm{ \DIA b \AND \DIA c} },\, \lef{\fm {\DIA b}},\, \lef{\fm{\DIA c}}\}$;\\
		$\lef{\fm \Sigma} = \big\{ \lef{\fm{\BOX \big( a \IMP (\DIA b \AND \DIA c) \big)}},\, \lef{\fm{\BOX \big((d \IMP \bot) \IMP \bot \big)}},\, \lef{\fm{\BOX \big((e \IMP \bot) \IMP \bot \big)}}\big\}$; \\
		${\fm \Pi} = \{\lef{\fm{a \IMP (\DIA b \AND \DIA c)}},\, \lef{\fm{(d \IMP \bot) \IMP \bot}} ,\, \lef{\fm{(e \IMP \bot) \IMP \bot}},\, \rig{\fm{d \IMP \bot}},\, \rig{\fm{e \IMP \bot}}\}$;\\
		$\lef{\fm \Pi} = \{\lef{\fm{a \IMP (\DIA b \AND \DIA c)}},\, \lef{\fm{(d \IMP \bot) \IMP \bot}} ,\, \lef{\fm{(e \IMP \bot) \IMP \bot}}\}$.\\
		The left part of the figure shows the seed $\sq H$, using the notation from Def.~\ref{def:leafloop}, whereas the right part of the figure shows a sequent $\sq G$, preceding the one in Figure~\ref{fig:agi-ex} in the proof-search tree. Sequent $\sq G$ has  a leaf-repetition pair $\tuple{\lb{18}, \lb{20}}$. Boldface nodes are suricata nodes. The colouring represents the embedding $\emb$ of~$\treeof{5}$ in~$\sq H$ into $\treeof{17}$ in $\sq G$, as required by the definition of $\IMP$-suricata repetition in Def.~\ref{def:repetition}.
	}
	\label{fig:agi-ex-rep}
\end{figure}

\section{Termination}
\label{sec:termination}

This section is dedicated to showing that the algorithm presented in Section~\ref{sec:algorithm} does indeed terminate. The basic idea will be to look at branches of proof search trees as indicated in Figure~\ref{fig:searchtree}. A sequence of sequents that occur along such a branch will be called an \emph{evil sequence}. The point of the termination proof will be to show that such a sequence cannot be infinite.

\subsection{Evil Sequences}
\label{sec:evil}

\begin{notation}
	For two sequents $\sq G,\sq H$, we write $\sq G\sseqto\sq H$ if $\sq H\in\ssatof G$, we write $\sq G\lseqto\sq H$ if $\sq H=\liftssatof G$. 
\end{notation}

\begin{definition}\label{def:evil}
  An \defn{evil sequence} $\evils$ is a sequence of non-axiomatic nice sequents $\sq G_0'$, $\sq G_0$, $\sq G_1'$, $\sq G_1$, $\sq G_2'$, $\sq G_2$, \ldots, such that 
  \begin{equation}
    \label{eq:evil}
    \sq G_0'\sseqto \sq G_0 \lseqto \sq G_1'\sseqto \sq G_1 \lseqto \sq G_2'\sseqto \sq G_2 \lseqto \cdots \lseqto \sq G_i' \sseqto\sq G_i \lseqto \cdots
  \end{equation}
  and such that there is a sequence of layers
  $L_0',L_0,L_1',L_1,L_2',L_2,\ldots$, such that for every $i\ge 0$,
  we have that $L_0,L_1,\ldots,L_i$ are layers of $\sq G_i$ with
  $L_0<L_1<\cdots<L_i$, and that $L_i'\subseteq L_i$ 
  is a layer of $\sq G_i'$, and that $L_i$ is not simulated (Def.~\ref{def:simul-layers}). 
  For a given formula $\fm F$, we say that $\evils$ is \defn{$\fm F$-evil} if additionally $\sq G_0'=\futs rr, \accs rr,\labelsw rF$. We will often write $\sq G_i'(\evils)$, $\sq G_i(\evils)$, $L_i'(\evils)$, $L_i(\evils)$ for    $\sq G_i'$, $\sq G_i$, $L_i'$, $L_i$, respectively.
\end{definition}

The basic idea of the termination argument will be that the proof search Algorithm~\ref{alg:proofsearch} terminates if and only if every evil sequence is finite.
In the following, we are therefore going to show that there is no infinite evil sequence.

\begin{lemma}\label{lem:eqlabels}
	Le $\fm F$ be a formula and let $\evils$ be an $\fm F$-evil sequence. 
	The number of equivalence classes for the relation $\lbeq$ on labels (Def.~\ref{def:eq-labels}) that can occur in a sequent $\sq G$ in $\evils$ is bounded by $2^{2\subf{F}}$ where ${\subf{F}}$ is the number of subformulas of $\fm F$. 
\end{lemma}

\begin{proof}
	By Lemma~\ref{lem:lifting} the full lifting $\liftssatof G$ only contains subformulas of $\sq G$ and by Lemma~\ref{lem:saturation} any sequent $\sq H$ from the saturation $\ssatof G$ only contains subformulas of $\sq G$ (with input or output polarity). 
	Hence, any label in a sequent $\sq G$ in an $\fm F$-evil sequence can only contain subformulas of $\fm F$. 
\end{proof}

\begin{lemma}\label{lem:eqclusters}
	Le $\fm F$ be a formula and let $\evils$ be an $\fm F$-evil sequence. 
	The number of equivalence classes for the relation $\lbeq$ on clusters (Notation~\ref{def:eqclusters}) that can occur in a sequent $\sq G$ in $\evils$ is bounded by $2^{2^{2\subf{F}}}$ where ${\subf{F}}$ is the number of subformulas of $\fm F$. 
\end{lemma}

\begin{proof}
	An equivalence class for the relation $\lbeq$ on clusters in $\sq G$ is a set of sets of formulas occurring in $\sq G$. 
	If $\sq G$ occurs in an $\fm F$-evil sequence then it can only contain subformulas of $\fm F$.
\end{proof}

\begin{corollary} These are implied by the preceding lemmas:
	\begin{itemize}
		\item In the step from  $L'_{i}(\evils)$ to $L_{i}(\evils)$ (saturation step) a branch can only grow by at most $2^{\subf F+1}$.
		\item In the step from $L_{i}(\evils)$ to $L'_{i+1}(\evils)$ (lifting step) a branch can only grow if the suricata label is on that branch.
	\end{itemize}
\end{corollary}

\begin{lemma}\label{lem:finbranch}
  Le $\fm F$ be a formula, let $\evils$ be an $\fm F$-evil sequence, and $\sq G\in\evils$. Then every layer $L$ in $\sq G$ is finitely branching (when seen as a tree of clusters).
\end{lemma}

\begin{proof}
  Because of the $\shrinkr$-rule every cluster $\lbc C$ can have at most $2^{2\subf{F}}$ labels, and every label $\lb x\in\lbc C$ can have at most $2^{2\subf{F}}$ immediate children. Consequently, $\lbc C$ can have at most $16^{\subf{F}}$ children in the layer-tree.
\end{proof}

\subsection{Bounding the Length of a Branch}
\label{sec:bounding}

In this section we will make heavy use of Kruskal's tree theorem~\cite{Kruskal:60} which can be formulated in our setting as follows.

\begin{lemma}[Kruskal's Theorem]\label{lem:kruskal}
  Let $\fm F$ be a formula, let $\sq G_1,\sq G_2,\ldots$ be an infinite sequence of nice sequents containing only subformulas of $\fm F$, and let $\lb{x_k}\in\labelsofx{\sq G_k}$ be a singleton cluster for every $k\ge1$. Then there are $1\le i<j$ such that there is an embedding from   $\treeof{x_i}$ into $\treeof{x_j}$.
\end{lemma}

\begin{proof}
  The sequence $\treeof{x_1},\treeof{x_2},\ldots$ meets exactly the conditions of Kruskal's Theorem~\cite{Kruskal:60}.
\end{proof}

\begin{notation}
	Let $\sq G$ be a sequent. For a label $\lb x\in\labelsof G$ we denote by $\futuresof x=\set{\lb{x'}\in\labelsof G\mid \futsq x{x'}}$ the set of its \defn{futures}.
\end{notation}

\begin{definition}[Depth within a layer and relative depth]\label{def:depth}
  Let $\lb x\in\labelsof G$ be a label in a nice sequent (Def.~\ref{def:nice}). The \defn{depth} of $\lb x$, denoted by $\ldepthof x$ is the length of the path from the root of $\layerof x$ (when seen as a tree of clusters) to the cluster that contains $\lb x$. For labels $\lb x,\lb y\in\labelsof G$ we define their \defn{relative depth}, denoted by $\lreldepth xy$, as follows:
  \begin{equation}
    \label{eq:reldepth}
    \lreldepth xy =\left\{
    \begin{array}{ll}
      \ldepthof y-\max\set{\ldepthof{x'}\mid \lb{x'}\in\futuresof x\cap \layerof y}
      &
      \text{if $\vissq x y$ (Def.~\ref{not:fully-labelled}) } \strut
      \\[.5ex]
      -1
      &
      \text{otherwise}\strut
    \end{array}
    \right.
  \end{equation}
  For an evil sequence $\evils$, we write $\ldepthi{\evils}$ for the depth of the suricata label of $L_i'(\evils)$ 
  (which might not be the same as the depth of the suricata label of $L_i(\evils)$ given that the suricata in $L_i'(\evils)$ is never in a cluster while it could be in one $L_i(\evils)$). 
  And for $i <j$, we write $\lreldepthij ij\evils$ for $\lreldepth{s_i}{s_j}$ where $\lb{s_i}$ and $\lb{s_j}$ are the suricata labels of  $L_i'(\evils)$ and $L_j'(\evils)$, respectively. 
  Finally, we write $\lreldepthi{\evils}$ for $\max\set{\lreldepth{s_i}z\mid z\in L_i'(\evils)}=\max\set{\ldepthof z-\ldepthof{s_i}\mid \lb z\in L_i'(\evils)\mbox{~and~}\accsq {s_i}z}$, i.e., the maximal distance from $\lb{s_i}$ to a leaf of $\treeof{s_i}$.
\end{definition}

\begin{lemma}\label{lem:sur-depth}
  Let $\fm F$ be a formula and let $\evils$ be an $\fm F$-evil sequence. Then there is a number $\depthf{F}\in\Nat$, such that for all $i\ge0$, we have $\ldepthi{\evils}\le \depthf F$.
\end{lemma}

\begin{proof}
  By way of contradiction, assume that there is no such $\depthf{F}$, that is, for every $d\in\Nat$, there is $i\ge 0$ such that $\ldepthi{\evils} > d$.
  This means that $\evils$ is infinite and also has an infinite subsequence $\evils'$ in which the depth of the suricata increases at each step.  By Lemma~\ref{lem:finbranch} $L'_i$ is finitely branching. This means, that we can find a subsequence $\evils''$ of $\evils'$ such that additionally all suricata labels are on the same branch. 
  
  There are two cases to consider: 
  \begin{enumerate}
  \item There is an infinite sequence of natural numbers $0=k_0<k_1<k_2<k_3<\cdots$, such that for every $i>0$ we have that $k_{i+1}$ is the smallest number with
  \begin{equation}
    \label{eq:depinc0p}
    \lreldepthij{k_i}{k_{i+1}}{\evils}\ge1
    \quad.
  \end{equation}
  i.e., there is a subsequence of $\evils$ such that the relative depth between the suricatas of two consecutive steps is always positive. Next, we let $h_i=k_{2i}$ for all $i>0$. Then we have 
  \begin{equation}
    \label{eq:depinc0}
    \lreldepthij{h_i}{h_{i+1}}{\evils}\ge2
    \quad.
  \end{equation}
  Now  let $\lb{s_i}$ be the suricata 
  label of $L'_{h_i}(\evils)$, 
  for every $i>0$. 
  By the Pigeonhole Principle, among those suricatas must be either infinitely many $\BOX$-suricatas or infinitely many $\IMP$-suricatas.
   \begin{itemize}
  \item In the first case, we will find by Lemma~\ref{lem:eqlabels} 
  two numbers $i<j$ such that there is a $\BOX$-repetition pair (Def.~\ref{def:boxloop}) between $\lb{s_j}$ and a future of $\lb{s_i}$, which by  Algorithm~\ref{alg:liftsat}\ref{l:substloop} would give us $\lreldepthij ij\evils=\lreldepth{s_i}{s_j}=1$, contradicting~\eqref{eq:depinc0}. 
\item 
  In the second case, we have an infinite subsequence of trees 
  \begin{equation}
    \treeof{s_0}  
      \;,\;
    \treeof{s_1}\;,\;
    \treeof{s_2}\;,\;
    \ldots
  \end{equation}
  where each $\lb{s_i}$ is a $\IMP$-suricata and where $\treeof{s_i}$ is a subtree of $L'_{h_i}$ in $\sq G'_{h_i}$ for every $i>0$.
  We can now apply Lemma~\ref{lem:kruskal}, to obtain $i<j$ such that there is an embedding $\emb$ from   $\treeof{s_i}$ into $\treeof{s_j}$ (such that in particular $\embof{s_i} = \lb{s_j}$ because the $\rigsym$-formula is necessarily located at both $\lb{s_i}$ and~$\lb{s_j}$).
  This means that each element of the set $\futuresof{s_i}\cap L'_{h_{j}}(\evils)$ forms an $\IMP$-repetition pair (Def.~\ref{def:imploop}) with $\lb{s_j}$. 
  This means that by Algorithm~\ref{alg:liftsat}\ref{l:imploop} we have $\lreldepthij ij\evils=\lreldepth{s_i}{s_j}=1$, contradicting~\eqref{eq:depinc0}. 
   \end{itemize}
   This contradicts the existence of a subsequence with~\eqref{eq:depinc0p}.
 \item Consequently, there is some $N\in\Nat$ such that for all $N<i<j$ we have that
  \begin{equation}
    \label{eq:depinc0z}
    \lreldepthij ij{\evils}=\lreldepth{s_i}{s_j}\le0
    \quad,
  \end{equation}
  which remains true for all for all $N<i<j$ with $\sq G'_i,\sq G'_j\in \evils''$. Since $\ldepthi{\evils}=\ldepthof{s_i}$ is growing at each step and the sequents are all nice, we must have that the number of futures of certain labels are increasing. In other words, for every $m\in\Nat$, we have some label $\lb{x_i}\in\labelsofx{\sq G'_i(\evils)}$, such that $\accsqx[\sq G'_i]{x_i}{s_i}$ and $\sizeof{\futuresof {x_i}\cap L'_j(\evils)}>m$. Now observe that we cannot have an infinite subsequence $N<k_1<k_2<k_3<\cdots$, such that for all $j>0$:
  \begin{equation}
    \label{eq:xminc}
    \accsqx[\sq G'_{k_j}]{x_{k_j}}{s_{k_j}}
    \qquand
    \sizeof{\futuresof {x_{k_j}}\cap L'_{k_{j+1}}(\evils)}>1
    \qquand
    \forall i<j.\;\lb{x_{k_j}}\notin\futuresof{x_{k_i}}
  \end{equation}
  as this would contradict~\eqref{eq:depinc0z}.
  Hence, we can find some $\lb x \in L'_i$ for some $N<i$ whose number of futures is growing unboundedly. More precisely, we can find an infinite sequence $i<k_1<k_2<k_3<\cdots$ such that  
  \begin{equation}\label{eq:futinc1}
    \sizeof{\futuresof x\cap L'_{k_1}(\evils)}
    \;<\;
    \sizeof{\futuresof x\cap L'_{k_2}(\evils)}
    \;<\;
    \sizeof{\futuresof x\cap L'_{k_3}(\evils)}
    \;<\;
    \cdots
  \end{equation}
  Since by Lemma~\ref{lem:eqlabels}, the number of equivalence classes of $\lbeq$ for labels is bounded, we can also assume that there is an equivalence class $[\lb y]=\set{\lb{y'}\mid \eqlab{y'}{y}}$ such that we have a subsequence of natural numbers $0<h_1<h_2<h_3<\cdots$, such that
  \begin{equation}\label{eq:futinc1p}
    2\;\le\;\sizeof{\futuresof x\cap[\lb y]\cap L'_{h_1}(\evils)}
    \;<\;
    \sizeof{\futuresof x\cap[\lb y]\cap L'_{h_2}(\evils)}
    \;<\;
    \sizeof{\futuresof x\cap[\lb y]\cap L'_{h_3}(\evils)}
    \;<\;
    \cdots
  \end{equation}
  is strictly increasing at each step and more precisely that each $h_{j+1}$ is the smallest number such that
  \begin{equation}\label{eq:futinc2}
  \sizeof{\futuresof x\cap[\lb y]\cap L'_{h_j}(\evils)}<
  \sizeof{\futuresof x\cap[\lb y]\cap L'_{h_{j+1}}(\evils)}
  \end{equation}
  Now  let $\lb{s_{h_j}}$ be the suricata
  label of $L'_{h_j}(\evils)$ in $\sq G'_{h_j}(\evils)$, for every $j>0$.
  In the sequence $\lb{s_{h_1}},\lb{s_{h_2}},\ldots$ there are either infinitely many $\BOX$-suricatas or infinitely many $\IMP$-suricatas. We then therefore pick a subsequence consisting of only $\BOX$-suricatas or only $\IMP$-suricatas. Thus, we now assume without loss of generality, the sequence $\lb{s_{h_1}},\lb{s_{h_2}},\ldots$ either contains only $\BOX$-suricatas or only $\IMP$-suricatas.
  In both cases we have a sequence of trees (Def.~\ref{def:tree-of-label})
  \begin{equation}
    \treeof{s_{h_1}}\;,\;
    \treeof{s_{h_2}}\;,\;
    \treeof{s_{h_3}}\;,\;
    \ldots
  \end{equation}
  By Lemma~\ref{lem:kruskal}, there are $i<j$ such that there is an embedding from   $\treeof{s_{h_i}}$ into $\treeof{s_{h_j}}$. Since we have by~\eqref{eq:futinc1} that $\sizeof{\futuresof x\cap[\lb y]\cap L'_{h_i}(\evils)}<\sizeof{\futuresof x\cap[\lb y]\cap L'_{h_{j}}(\evils)}$ there must be some $\lb{x_i}\in \futuresof x\cap[\lb y]\cap L'_{h_i}(\evils)$ with $\sizeof{\futuresof {x_i}\cap[\lb y]\cap L'_{h_{j}}(\evils)}\ge 2$. Now, any two distinct elements of that set
  form a triangle repetition pair (Def.~\ref{def:triangleloop}). Observe that in the case of all $\lb{s_{h_i}}$ being $\BOX$-suricatas, all $\treeof{s_{h_i}}$ are singletons, and therefore the triangle repetition condition is fulfilled in both cases.
  This means that by Algorithm~\ref{alg:liftsat}\ref{l:substloop} we have 
  $\sizeof{\futuresof x\cap[\lb y]\cap L_{h_{j}}(\evils)} \ge \sizeof{\futuresof x\cap[\lb y]\cap L'_{h_{j+1}}(\evils)}$,
  contradicting \eqref{eq:futinc2}, 
  because $\sizeof{\futuresof x\cap[\lb y]\cap L'_{h_{j}}(\evils)} = \sizeof{\futuresof x\cap[\lb y]\cap L_{h_{j}}(\evils)}$ in any case.
  Therefore, we cannot have our infinite sequence with~\eqref{eq:futinc1}.
  \end{enumerate}
  Consequently, $\ldepthi{\evils}$ cannot grow unboundedly and $\depthf{F}$ must exist.
\end{proof}

\begin{lemma}\label{lem:sur-leaf}
  Let $\fm F$ be a formula and let $\evils$ be an $\fm F$-evil sequence. Then there is a number $\heightf F\in\Nat$, such that for all $i\ge0$, we have $\ldepthi{\evils}+\lreldepthi{\evils}\le \heightf F$.
\end{lemma}

\begin{proof}
  Observe that $\ldepthi{\evils}+\lreldepthi{\evils}$ is the length of the longest branch in $L'_i$ that contains the suricata $\lb{s_i}$, and that the lemma puts a bound on the branches of the layer trees that contain the suricata label.
  
  As in the previous proofs, we assume by contradiction, that no such $\heightf F$ exists, and therefore $\evils$ must be infinite, having a subsequence in which $\ldepthi{\evils}+\lreldepthi{\evils}$ increases at each step. 

  Let $0<k_1<k_2<k_3<\cdots$ be the indices of that subsequence. Let $\lb{s_i}$ the suricata of layer $L'_{k_i}(\evils)$ and let $\lb{z_i}\in L'_{k_i}(\evils)$ with $\lreldepthi{\evils}=\lreldepth{s_i}{z_i}$. In particular, each $\lb{z_i}$ is a leaf of $\treeof{s_i}$ and therefore also of $L'_{k_i}$.
  By Lemma~\ref{lem:finbranch}, we can select our subsequence such that $\vissq {z_i} {z_j}$ for all $i<j$. Note that then $\lreldepth{z_i}{z_{j}}<0$ is not possible.
  There are now two cases to consider:
  \begin{enumerate}
  \item We can select our subsequence such that
    \begin{equation}
      \label{eq:zizj}
      \lreldepth{z_i}{z_j}>0\qquad\mbox{for $i<j$}.
    \end{equation}
    Observe that by Lemma~\ref{lem:sur-depth}, we can also assume that our subsequence starts late enough such that we never have $\vissq {z_i} {s_j}$ for $i<j$ (otherwise we would have infinitely often \hbox{$\lreldepth{s_i}{s_j}>0$}, which has been ruled out in the first half of the proof of  Lemma~\ref{lem:sur-depth}). In particular, we can enforce that $\ldepthof{z_1}>\depthf{F}$.
    Furthermore, by Lemma~\ref{lem:eqclusters}, we can assume that the subsequence is chosen such that $\clusterof{z_i}\lbeq\clusterof{z_j}$ (which means in particular $\clusterof{z_i}\clsubsymb\clusterof{z_j}$) for all $i<j$. 
    As before, this gives us an infinite sequence of trees  $\treeof{s_1},\treeof{s_2},\treeof{s_3},\ldots$, such that we get by Lemma~\ref{lem:kruskal} an embedding from   $\treeof{s_i}$ into $\treeof{s_j}$ for some $i<j$. Therefore, all elements of $\futuresof{z_i}\cap L'_{k_{j}}(\evils)$ form a leaf-repetition pair with $\lb{z_j}$.
    (Observe that we have ensured that $\vissqx[\sq G'_j] {z_i} {z_j}$ and that for all $\lb{z'}\in \futuresof{z_i}$ we do \emph{not} have that $\lb{z'}\srelx{\sq G'_j}\lb {s_j}\srelx{\sq G'_j}\lb {z_j}$ because we do not have $\vissqx[\sq G'_j] {z_i} {s_j}$.) 
    But this means that because of  Algorithm~\ref{alg:liftsat}\ref{l:substloop} we should have $\lreldepth{z_i}{z_j}=0$, which is a contradiction.
  \item There is an $N\in\Nat$ such that
  \begin{equation}
    \label{eq:Nzi}
    \lreldepth{z_i}{z_{i+1}}=0 \qquad\mbox{for $N<i$}
  \end{equation}
  we proceed similarly to the case~2 in the proof of Lemma~\ref{lem:sur-depth}. Since $\ldepthi{\evils}+\lreldepthi{\evils}=\ldepthof{z_i}$ is growing at each step and the sequents are all nice, we must have that the number of futures of certain labels are increasing. In other words, for every $m\in\Nat$, we have $N<i<j$ and some label $\lb{x_i}\in\labelsofx{\sq G'_i(\evils)}$, such that $\accsqx[\sq G'_i]{s_i}{x_i}\srelx{\sq G'_i}\lb{z_i}$ and $\sizeof{\futuresof {x_i}\cap L'_j(\evils)}>m$. From here we proceed exactly as in case~2 in the proof of Lemma~\ref{lem:sur-depth}, the only difference being that $\accsqx[\sq G'_i]{s_i}{x_i}\srelx{\sq G'_i}\lb{z_i}$ instead of $\accsqx[\sq G'_i]{x_i}{s_i}$, equally making sure that the suricata does not interfere in the creation of the triangle repetition pair.
  \end{enumerate}
In both cases we get a contradiction to the existence of an infinite sequence with increasing $\ldepthof{z_i}$. Therefore the bound $\heightf F$ does exist.
\end{proof}

\begin{lemma}\label{lem:Lsize}
  Let $\fm F$ be a formula and let $\evils$ be an $\fm F$-evil sequence. Then there is a number $\maxlayerf F\in\Nat$, such that for all $i\ge0$, we have $\sizeof{L_{i}(\evils)}\le\maxlayerf F$.
\end{lemma}

\begin{proof}
  Recall that $L'_{i}(\evils)$ and $L_{i}(\evils)$ are trees of clusters, and by Lemmas~\ref{lem:eqlabels} and~\ref{lem:eqclusters} and the `shrink' of twins
  in Algorithm~\ref{alg:liftsat}\ref{l:twin}, the number of immediate children 
  that a node 
  can have in these trees is bounded by $2^{\subf F+1}\times 2^{\subf F+1}=2^{2\subf F+2}$ (see proof of Lemma~\ref{lem:finbranch}).
  Next, observe that in the step from  $L'_{i}(\evils)$ to $L_{i}(\evils)$ (saturation step) a branch can only grow by at most $2^{\subf F+1}$, again because of Lemmas~\ref{lem:eqlabels}. 
  And in the step from $L_{i}(\evils)$ to $L'_{i+1}(\evils)$ (lifting step) a branch can only grow if the suricata label is on that branch, as otherwise no new $\lef\DIA$-formulas can be added,
  and all the existing ones are by definition happy.
  But by Lemma~\ref{lem:sur-leaf} the length of such a branch is bounded by $\heightf F$. 
  This gives the bound $(2^{2\subf F+2})^{\heightf F+2^{\subf F+1}}$ for the size of the tree of $L_{i}(\evils)$, and therefore  $\maxlayerf F=2^{2^{\subf F+1}}\times2^{(2\subf F+2)(\heightf F+2^{\subf F+1})}$ is the desired bound for  $\sizeof{L_{i}(\evils)}$.
\end{proof}

\begin{theorem}\label{thm:termination}
  The proof search algorithm (Algorithm~\ref{alg:proofsearch}) terminates.
\end{theorem}

\begin{proof}
  First, observe that each step in Algorithm~\ref{alg:proofsearch} terminates, i.e., the computation of the saturation $\satof\cdot$ (Algorithm~\ref{alg:sat}) terminates (Lemma~\ref{lem:sat-term}), and the computation of the lifting $\liftsatof{(\cdot)}$ (Algorithm~\ref{alg:liftsat}) terminates (Lemma~\ref{lem:lifting}).

  Then, the only way for non-termination to occur would therefore be if there were an evil sequence that alternated between the steps of Algorithm~\ref{alg:proofsearch} infinitely many times.
  That would imply that the layers of the sequents within the evil sequence would be growing unboundedly.
  However, since by Lemma~\ref{lem:Lsize}, there is a bound for the size $\sizeof{L_{i}(\evils)}$ of a layer $L_{i}(\evils)$ of an evil sequence $\sigma$, there must be eventually a simulation (Def.~\ref{def:simul-layers}) which triggers termination at step~(4) of Algorithm~\ref{alg:proofsearch}. 
  Therefore, there is no infinite evil sequence $\evils$, and Algorithm~\ref{alg:proofsearch} must terminate eventually.
\end{proof}

\section{Completeness} 
\label{sec:countermodel}	

\begin{construction}\label{con:Gstar}
  Assume we initiate the algorithm with a formula~$\fm F$. If we terminate at Step~4, we have found a saturated sequent $\sq G_i$ such that $\whitef rF{G_\mathit{i}}$ and $\liftssatof G_i=\sq G_i$. This means that each topmost layer 
of $\sq G_i$  is either happy or  simulated. 
Thus,  for each  unhappy layer $L$ there is some happy inner layer $L'$ that  simulates $L$ via a simulation $\simul_L$. We now define $\sq G_i^\star$ to be obtained from $\sq G_i$ by adding a relational atom $\futs x{x'}$  whenever $\lb {x'}\simul_L \lb x$ for some unhappy layer $L$, and by closing the result under transitivity of $\le$. 
\end{construction}

\begin{lemma}\label{lem:counter-happy}
  The sequent  $\sq G_i^\star$ is happy.
\end{lemma}

\begin{proof}
	We need to show that  $\sq G_i^\star$ is structurally happy and all its formulas are happy. Since no new labels were added, $(\Lref)$- and $(\Rref)$-structural happiness is preserved. $(\Rtr)$-structural happiness is preserved because no new $R$-links were added. $(\Ltr)$-structural happiness is explicitly enforced. To show that $(\monl)$-, $(\fone)$-, and $(\ftwo)$-structural happiness are preserved, it is sufficient to demonstrate them for each of the  $\le$-links added before the transitive closure. For the \mbox{$(\fone)$-,} and $(\ftwo)$-structural happiness these are exactly the simulation conditions $\rn{S1}$ and $\rn{S2}$ from Definition~\ref{def:simul-layers}. For the $(\monl)$-simulation this follows from the fact that $\simul_L \subseteq \lbeq$ for all  $L$, which means that whenever $\blackf{x}AG$ and $\futsq x{x'}$ because $\lb{x'} \simul_L \lb x$ for some $L$, we have $\lb x \lbeq \lb{x'}$ and, hence, $\blackf{x'}AG_i^\star$. This completes the proof that $\sq G_i^\star$ is structurally happy. Since no formulas or links were removed, all formulas happy in $\sq G_i$ remain happy in $\sq G_i^\star$ (here we again rely on the local formulation of happiness for $\blackf x{A \IMP B} G_i^\star$ in Def.~\ref{def:happyformula}). All formulas from the inner layers of saturated sequent $\sq G_i$ were happy, as were all formulas apart from possibly some $\fmw{F}$ of one of the two shapes $\fmw{A \IMP B}$ or $\fmw{\BOX A}$ in topmost layers.   For any such  unhappy $\whitef x{F}G_i$ from some topmost layer~$L$ of $\sq G_i$, there is an inner layer $L'$ and, by Prop.~\ref{prop:sim_bij}, there is a label $\lb{x'} \in L'$ such that $\lb{x'}\simul_L\lb x$, hence, $\futsq x{x'}$.  At the same time $\lb x \lbeq \lb{x'}$ so $\whitef{x'}{F}G_i$  is happy in its inner layer. If $\fmw{F}=\fmw{A\IMP B}$, we must have $\blackf{z}{A}G_i$, $\whitef{z}{B}G_i$, and $\futsqp{x'}{z}G_i$. Then $\futsq x z$ by $(\Ltr)$-structural happiness of $\sq G_i^\star$, which makes $\whitef{x}{F}G_i^\star$ happy. Finally, if $\fmw{F}=\fmw{\Box A}$, we must have $\whitef{z}{A}G_i$, $\futsqp{x'}{u}G_i$, and $\accsqp{u}{z}G_i$. Then $\futsq x u$ by $(\Ltr)$-structural happiness of $\sq G_i^\star$ and $\accsq{u}{z}$ by construction, which makes $\whitef{x}{F}\sq G_i^\star$ happy. 
\end{proof}

\begin{theorem}[restate = countermodel, name = ] 
	\label{thm:step4}
	If Algorithm~\ref{alg:proofsearch} terminates in Step~4, then
	formula~$\fm F$ is not a theorem of\/~$\ISfour$.
\end{theorem}

\begin{proof}
  If the algorithm terminates in Step~4, then $\liftssatof G_i=\sq G_i$ for some non-axiomatic sequent $\sq G_i \in \sqset{S}_i$. By Lemma~\ref{lem:counter-happy}, the sequent $\sq G_i^\star$ constructed in Construction~\ref{con:Gstar} is happy. Furthermore, it is easy to see that the label $\lb r$, the only label in~$\sq G(\fm F)$ is still present in $\sq G_i$, and therefore also in $\sq G_i^\star$. Indeed, no step in the algorithm removes the root $\lb r$. Hence, we have $\whitef rFG_i^\star$. By Theorem~\ref{thm:completeness}, we have therefore that 
$\nforce{\mmodelof{\sq G_i^\star}}rF$. Therefore, by Theorem~\ref{thm:plotkin}, $\fm F$~is not a theorem in $\ISfour$.%
\end{proof}

\section{Soundness}
\label{sec:soundness}

\begin{definition}
  We define the following three proof systems: 
$$
\begin{array}{rcl}
  \labISloops&=&\set{
    \id,\lef\BOT,
    \satr, \bdiam,\liftrim, \liftrbox,
    \dialoopr,\leafloopr, \boxloopr, \imploopr,\triangleloopr,\shrinkr}\\
  \labISloop&=&\set{
    \id,\lef\BOT,
    \satr, \bdiam,\liftrim, \liftrbox,
    \dialoopr, \leafloopr, \boxloopr, \imploopr,\triangleloopr}\\
  \labISnoloop&=&\set{
    \id,\lef\BOT,
    \satr,\bdiam,\liftrim, \liftrbox}
\end{array}
$$
A \defn{nice} derivation in one of these systems is a derivation in which every occurring sequent is nice, and additionally the rules $\id,\lef\BOT, \satr, \bdiam$, and $\dialoopr$ (if present) are only applied in top-most layers. A \defn{proof} is a derivation tree in which every leaf is axiomatic. 
\end{definition}

\begin{lemma}
  Let $\fm F$ be a formula, and let $\deri$ be a derivation in $\labISloops$ or $\labISloop$ or $\labISnoloop$ with conclusion   $\futs rr, \accs rr,\labelsw rF$. Then $\deri$ is nice.
\end{lemma}

\begin{proof}
  This follows immediately from Lemma~\ref{lem:satr-nice}, Lemma~\ref{lem:dialoop-nice} Lemma~\ref{lem:liftr-nice}, and Lemma~\ref{lem:lifting-nice}. 
\end{proof}

Now assume Algorithm~\ref{alg:proofsearch} terminates with Step~\ref{l:ax}. By Lemmas~\ref{lem:saturation} and~\ref{lem:lifting}, it follows that we have a nice proof of our formula $\fm F$ in~$\labISloops$. The purpose of this section is to transform this nice  $\labISloops$-proof into a $\labISnoloop$-proof. By Lemmas~\ref{lem:satr}, \ref{lem:bdia} and \ref{prop:liftr}, it then follows that we also have a proof of $\fm F$ in~$\labISfs$.

The first step is to remove all instances of the $\shrinkr$-rule.

\begin{lemma}[Shrink elimination]\label{lem:shrinkelim}
  Let $\der$ be a nice proof of a sequent $\sq G$ in $\labISloops$. Then there is a nice proof of $\sq G$ in $\labISloop$.
\end{lemma}

\begin{proof}
  This is a straightforward induction on the height of the derivation, similar to the admissibility of weakening, observing that the reduction preserves ``niceness''.
\end{proof}

From now on, we write $\loopr$ for any instance of one of the rules
$\dialoopr$, $\leafloopr$, $\imploopr$, $\boxloopr$, $\triangleloopr$.
And we are going to show how to eliminate every instance of a $\loopr$-rule from a derivation~$\deric$ in $\labISloop$, which will result in a derivation in $\labISnoloop$. 

\subsection{Unfolding}
\label{sec:unfolding}

Observe that each $\loopr$-rule instance has as premise either  $\sq G \labelsubst{y}{x}\cup\trans{\interval {x}y}$ or $\sq G\cup\trans{\interval {x}y}$, when the conclusion is $\sq G$.
Recall that in this case we must have $\accsq xy$ and that we have defined $\interval xy=\set{\lb v\mid \accsq xv \mbox{ and } \accsq vy \mbox{ and } \lb v\neq\lb y}$ (see Notation~\ref{not:subst}). This creates a cluster, and in the following, we show how to unfold this cluster.

\begin{notation}
  Let $\sq G$ be a sequent with $\accsq xy$.
  For every $n\ge0$, we define the set $\interval xy^n=\set{\lb{v_i}\mid \lb v\in\interval xy\mbox{ and }1\le i\le n}$ (i.e., $\interval xy^0=\emptyset$). We now define $\sq G\interval xy^n$ to be the sequent obtained from~$\sq G$ by adding the following, where we always assume that  $\lb z\in\labelsof{G}\setminus\interval xy$ and $\lb u,\lb v\in\interval xy$:
  \begin{itemize}
  \item $\accs {u_i}{v_i}$ for every $i=1..n$ if $\accs uv\in\sq G$;
  \item $\accs {u}{v_i}$ for every $i=1..n$;
  \item $\accs {u_i}{v_j}$ for all $1\le i<j\le n$;
  \item $\accs z{v_i}$ for every  $i=1..n$ and $\lb v\in\interval xy$ if there is a $\lb u\in\interval xy$ with  $\accs zu\in\sq G$;
  \item $\accs {v_i}z$ for every  $i=1..n$ and $\lb v\in\interval xy$ if there is a $\lb u\in\interval xy$ with  $\accs uz\in\sq G$.
  \item $\labelsb {v_i}A$ for every $i=1..n$ if $\labelsb vA\in\sq G$;
  \item $\labelsw {v_i}A$ for every $i=1..n$ if $\labelsw vA\in\sq G$;
  \end{itemize}
  We will call the sequent $\sq G\interval xy^n$ the \defn{$n$-unfolding} of a sequent of shape $\sq G \labelsubst{y}{x}\cup\trans{\interval {x}y}$ or $\sq G\cup\trans{\interval {x}y}$.
  Observe that $\sq G\interval xy^0=\sq G$.
\end{notation}

\begin{notation}
  For a sequent $\sq G$, we define $\toplabelsof G$ to be the set of labels in a topmost layer, i.e., $\toplabelsof G=\set{x\in\labelsof G\mid\mbox{$x$ has no future in $\sq G$}}$. For two sequents $\sq G$ and $\sq H$, we write $\sq G\topeq\sq H$ iff there is a bijection $\emb\colon\toplabelsof G\to\toplabelsof H$, preserving the formulas and $\rel$, i.e., for all $x\in \toplabelsof G$ we have $\lb x\lbeq\lb{\embof x}$, and  for all $x,y\in \toplabelsof G$ we have $\accsq xy$ iff $\accsqq[H]{\embof x}{\embof y}$.
\end{notation}

\begin{lemma}
  The $\weakr$-rule is height-preserving admissible for $\labISnoloop$.
\end{lemma}

\begin{proof}
  By induction on the height of the proof.
\end{proof}

\begin{lemma}[Unfolding Lemma]
  \label{lem:unfold}
  Let $\sq G$ be a sequent with $\accsq xy$. If $\lb x\lbeq \lb y$ and $\sq G \labelsubst{y}{x}\cup\trans{\interval {x}y}$ has a proof $\derib$ in $\labISnoloop$, or $\sq G\cup\trans{\interval {x}y}$ has a proof $\derib$ in $\labISnoloop$, then there is an $n_0\ge0$ such that for every $n\ge n_0$, every sequent $\sq G'$ with $\sq G'\topeq\sq G\interval xy^n$ has a proof $\hderib$ in $\labISnoloop$, such that $\height\hderib\le\height\derib$. 
\end{lemma}

\begin{proof}
  We proceed by induction on $\height\derib$ and make a case analysis on the bottom-most rule instance $\rr$ in $\derib$.
  \begin{itemize}
  \item $\rr$ is $\idr$ or $\lef\BOT$. Then we can apply the corresponding axiom rule to $\sq G\interval xy^0$, and we are done by taking $n_0=0$.
  \item $\rr$ is $\satr$. Then the last rule in $\derib$ has the following shape, for $k \geq 0$:
  $$
  \vliiinf{\satr}{}{\sq G}{\sq H_1 }{\cdots}{\sq H_k }
  $$
  Since $\sq G$ contains $\trans{\interval{x}{y}}$, each premise $\sq H_i$ also contains $\trans{\interval{x}{y}}$, and we can apply the induction hypothesis. 
  So, there are $n_0^1, \dots, n_0^k$ such that, for each $i \leq k$, it holds that any sequent $\sq H'_i \topeq \sq H_i \interval{x}{y}^{m_i} $ is derivable, for $m_i \geq n_0^i$. We need to find an $n_0$ such that any sequent $\sq G' \topeq \sq G\interval{x}{y}^n$ is derivable, for any $n \geq n_0$. 
  
  We first decompose the instance of $\satr$ in a subderivation $\deric$ consisting only of rules in $\set{\lef\AND,\rig\AND,\lef\OR,\rig\OR,\lef\IMP,\lef\BOX,\rig\DIA}$. 
  We find our $n_0$ and construct a proof of an arbitrary $\sq G'$ simultaneously, reasoning by induction on $\height{\deric}$, and distinguishing cases depending on the last rule applied in $\deric$. We use the same induction hypothesis as in the main lemma, only `restricted' to $\deric$. 
  	\begin{itemize}
  	\item Each leaf of $\deric$ is one of the sequents $\sq H'_i$ detailed above, derivable by assumption.  
  	\item The last rule in $\deric$ is an instance of $\lef\AND$, applied to a formula $\labelsb{w}{A \AND B}$ in $\sq G$. 
  	Let us call $\sq H$ the premise of $\lef\AND$. 
  	By induction hypothesis, there is a $n_0' $ such that any sequent $\sq H' \topeq \sq H \interval{x}{y}^n$ is derivable, for any $n \geq n_0'$.  We let $n_0 = n_0' $, and consider an arbitrary $\sq G' \topeq \sq G\interval{x}{y}^n$, for $n \geq n_0$. In case $\lb w\notin \interval{x}{y}$, we apply rule $\lef\AND$ to formula $\labelsb w{A \AND B} $ in $\sq G'$. The resulting sequent is derivable by induction hypothesis. Otherwise, if $\lb w \in \interval{x}{y}$, then $\sq G'$ contains $n$ labels $\labelsb{w_1}{A \AND B}, \cdots, \labelsb{w_n}{A \AND B} $. We apply $n$ instances of rule $\lef\AND$ to these formulas, obtaining a sequent which is derivable by induction hypothesis. 
  	\item The last rule in $\deric$ is an instance of $\rig\OR$. This case is similar to the previous one.
  	\item	The last rule in $\deric$ is an instance of $\lef\OR$, applied to formula $\labelsb{w}{A \OR B}$ in $\sq G$. 
  	Let $\sq H_1$ and $\sq H_2$ be the left- and right-most premises of $\lef\OR$, containing formulas $\labelsb{w}{A}$ and $\labelsb{w}{B}$ respectively. The induction hypothesis gives us $n^1_0, n^2_0 \geq 0$. In case $\lb w \notin \interval{x}{y}$, we take $n_0 = \max(n_1, n_2)$. To obtain a proof of any sequent $\sq G' \topeq \sq G\interval{x}{y}^n$, for any $n \geq n_0$, we apply rule $\lef\OR$ to formula  $ \labelsb w{A \OR B}$ in $\sq G'$. Both premises are derivable by induction hypothesis. 
  	If instead $w \in \interval{x}{y}$, we take $n_0 = 2 \cdot \max(n^1_0, n^2_0) -1$. Consider an arbitrary $n\ge n_0$ and $\sq G'$ as defined as above. It contains $n$ `copies' of $\lb w$, that is, formulas $\labelsb{w_1}{A \OR B}, \dots, \labelsb{w_n}{A \OR B}$. We apply rule $\lef\OR$ to each of these formulas, obtaining $2^n$ premises, each still containing $n$ `copies' of $\lb w$, with formulas $\labelsb{w_1}{F_1}, \dots, \labelsb{w_n}{F_n}$, for $\fm{F_i} \in \{ \fm A, \fm B\}$. 
  	Moreover, because of the way $n$ was defined, every premise either contains at least $n^1_0$ copies of $\lb w$ labelling formula $\fm A$, or it contains at least $n^2_0$ copies of $\lb w$ labelling formula $\fm B$. 
  	In the former case, we apply $\w$ to delete every copy of $\interval{x}{y}$ in which (the copy of) $\lb w$ labels $B$. 
  	In the latter case, we proceed symmetrically. In either case, the resulting sequent is derivable by induction hypothesis. 
  	\item The last rule in $\deric$ is an instance of $\rig\AND$ or $\lef\IMP$. This is similar to the previous case.
  	\item The last rule in $\deric$ is an instance of $\vlinf{\lef\BOX}{}{\sq K, \accs w {k}, \labelsb w{\BOX A} }{\sq K, \accs wk,  \labelsb w{\BOX A}, \labelsb kA }$~.
  	Let $\sq G \labelsubst{y}{x}\cup\trans{\interval {x}y}$ or $\sq G\cup\trans{\interval {x}y}$ be the conclusion of $\lef\Box$. We consider the second case here, the first being similar. Then let  $\sq H\cup\trans{\interval {x}y}$ be the premise of $\lef\Box$.
  	We obtain $n_0'$ by inductive hypothesis. In case $\lb k \notin \interval{x}{y}$, we set $n_0 = n'_0$ and derive any sequent $\sq G' \topeq \sq G\interval{x}{y}^n$, for $n \geq n_0$ by applying $\lef\Box$ to some $\sq H' \topeq \sq H\interval{x}{y}^n$ that we have by induction hypothesis. Otherwise, if $\lb k \in \interval{x}{y}$, we distinguish the following cases: 
  	\begin{itemize}
  		\item If $\lb w \notin \interval{x}{y}$, or if $\lb w \in \interval{x}{y}$ and $\accsq{w}{k}$, then  we proceed as in the case of $\lef \AND$, by letting $n_0= n'_0$ and applying $n$ instances of rule $\lef\BOX$. 
  		\item If $\lb w \in \interval{x}{y}$ and it does \emph{not} hold that $\accsq{w}{k}$, we set $n_0 = n'_0 + 1$. To derive $\sq G'$, we apply rule $\lef\BOX$ to every pair of formulas $\accsqq[\hG]{w_i}{k_{i+1}}, \labels{w_i}{\BOX A}$, for $1 \leq i < n$. We thus apply $n-1 $ occurrences of rule $\lef\Box$. In the resulting sequent, every copy of $\lb{k}$, except the first one, labels formula $\fmb{A}$. We can then apply $\w$ (bottom-up) to 'delete' the first copy of $\interval{x}{y}$, giving us  $\sq H' \topeq \sq H\interval{x}{y}^{n-1}$, which is provable by induction hypothesis.
  	\end{itemize} 
  	\item The last rule in $\deric$ is an instance of $\rig\DIA$. This
  	is similar to the previous case.
  \end{itemize}
  This concludes our subproof. 
  We have constructed our $n_0$ and we have a proof $\hat \deric$ of a sequent $\sq G' \topeq \sq G \interval{x}{y}^n$, for any $n \geq n_0$. However, $\hat \deric$ is not our desired derivation: in particular, $\height{\hat \deric} \geq \height{\deric}$. The derivation $\hat \deric$ might contain instances of $\w$. We first permute all occurrences of $\w$ upwards in $\hat\deric$ as much as possible. We then obtain a subderivation $\hat \deric '$ which consists only of rules in $\set{\lef\AND,\rig\AND,\lef\OR,\rig\OR,\lef\IMP,\lef\BOX,\rig\DIA}$. We compress $\hat \deric'$ into a single instance of $\satr$ and, thanks to height-preserving admissibility of $\w$, we obtain a derivation of $\sq G' $ bounded by $\height{\derib}$.

  \item $\rr$ is $\bdiam$. Then last rule applied in $\derib$ has the following shape, where $\lb z$ is fresh:  
  $$
  \vlinf{\bdiam}{}{\sq K, \labelsb w{\DIA A} }{\sq K, \labelsb w{\DIA A} \cup \set{ \accs wz, \accs zz,\futs z,  \labelsb zA} \cup \set{ \accs v z \mid \accs v w \in \sq G} }
  $$
  Let us call $\sq H$ the premise of $\bdiam$, and $\sq G$ its conclusion. Since $\sq H$ also contains the interval $\interval xy$, we can apply the induction hypothesis to it, obtaining that there is an $n_0'$ such that, for every $n \geq n_0$, every sequent $\sq H'$ with $\sq H '\topeq \sq H$ is derivable. We distinguish two cases:
  \begin{itemize}
  \item If $\lb w\notin\interval xy$, we set $n_0 = n_0'$ and we consider an arbitrary sequent $\sq G'$ such that $\sq G' \topeq \sq G \interval{x}{y}^n$, for $n \geq n_0$. A proof of $\sq G'$ can be obtained by applying rule $\bdiam$ to formula $\labelsb{w}{\DIA A}$ in $\sq G'$. The premise of the rule is derivable by induction hypothesis. 
  \item 
    If $\lb w\in\interval xy$, we let $n_0=n_0'+1$, where $n_0'$ comes from the induction hypothesis. We have $n \geq n_0$ copies of $\lb w$ in $\sq G'\topeq\sq G\interval xy^n$, and we apply the rule $\bdiam$ to formula $\labelsb{w_{n}}{\DIA A}$, which is the `last' unfolding copy of $\lb w$ occurring in $\sq G\interval xy^n$. The resulting sequent $\sq H'$ contains a fresh label $\labelsb{z}{A}$, together with the relevant relational atoms. From $\sq H'$ we can with the  $\weakr$-rule obtain a sequent $\sq H''\topeq\sq H\interval xy^{n-1}$ (by deleting the $n$th copy of $\interval xy$ in $\sq H'$), which is provable in $\labISnoloop$ by induction hypothesis.
  \end{itemize}
  \item $\rr$ is  $\liftrim$. Let $\labelsw w{A\IMP B}$ be the formula to which the rule is applied.  There are two subcases:
    \begin{itemize}
    \item If $\lb w\notin\interval xy$, then we can apply the induction hypothesis to the premise of $\liftrim$, which is of shape $\sq H \labelsubst{y}{x}\cup\trans{\interval {x}y}$ or $\sq H\cup\trans{\interval {x}y}$ for some $\sq H$, as the cluster $\interval {x}y$ is not touched by the rule. This gives us $n_0$ and a proof for every $\sq H'\topeq \sq H\interval xy^{n}$ for $n\ge n_0$. A proof of $\sq G'\topeq\sq G\interval xy^n$ can be obtained by applying $\liftrim$, as its premise some $\sq H'\topeq \sq H\interval xy^{n}$.
   \item If $\lb w\in\interval xy$, then by Construction~\ref{def:imp-lifting}, the premise of  $\liftrim$ is of shape $\sq H\cup\trans{\interval{\hx'}{\hy'}}\cup\trans{\interval{\hx''}{\hy''}}$ where $\trans{\interval{\hx'}{\hy'}}$ and $\trans{\interval{\hx''}{\hy''}}$ are both copies of $\trans{\interval xy}$ (with the $\rigsym$-formulas removed) and we also have $\interval{\hx'}{\hy'}\srel{H}\lb\hw\srel{H}\interval{\hx''}{\hy''}$ where $\lb\hw$ is the suricata of the new layer with $\labelsb{\hw}A,\labelsw{\hw}B\in\sq H$. We can apply the induction hypothesis twice, giving us $n'_0,n''_0$ such that we have a proof of $\sq H'\topeq\sq H\interval{\hx'}{\hy'}^{n'}\interval{\hx''}{\hy''}^{n''}$ for every $n'\ge n'_0$ and $n''\ge n''_0$. We let $n_0=n'_0+n''_0+1$. Now assume we have $n\ge n_0$ and $\sq G'\topeq\sq G\interval xy^n$. We can apply $\liftrim$ to the $(n'_0+1)$th copy of $\lb w$ in $\interval xy^n$. Let the resulting premise be $\sq G''$, to which we apply the $\weakr$-rule to delete all other labels of the $(n'_0+1)$th copy of $\interval xy$, which results in some $\sq H'\topeq\sq H\interval{\hx'}{\hy'}^{n'_0}\interval{\hx''}{\hy''}^{n''}$ for some $n''\ge n''_0$.
    \end{itemize}
       \item  If $\rr$ is a $\liftrbox$. Let  $\labelsw w{\Box A}$ be the formula to which the rule is applied. Again, there are two subcases. 
       \begin{itemize}
       	\item  If $\lb w\notin\interval xy$, then we just apply the induction hypothesis, as in the case for $\liftrim$. 
       	\item  If $\lb w\in\interval xy$,  we let $n_0=n_0'+1$, where $n_0'$ comes from the induction hypothesis. We have $n \geq n_0$ copies of $\lb w$ in $\sq G'\topeq\sq G\interval xy^n$, and we apply the rule $\liftrbox$ to the formula $\labelsw{w_{n}}{\BOX A}$, which is the `last' unfolding copy of $\lb w$ occurring in $\sq G\interval xy^n$. The resulting sequent $\sq H'$ contains a fresh label $\labelsw{z}{A}$, together with the relevant relational atoms. From $\sq H'$ we can with the  $\weakr$-rule obtain a sequent $\sq H''\topeq\sq H\interval xy^{n-1}$ (by deleting the $n$th copy of $\interval xy$ in $\sq H'$), which is provable in $\labISnoloop$ by induction hypothesis.
          \qedhere
       \end{itemize}
  \end{itemize}
\end{proof}

In order to make use of this Unfolding Lemma, we have to produce a derivation which has a premise some $\sq G'\topeq\sq G\interval xy^n$ for any given $n\ge 0$. This is the purpose of the following section.


\begin{figure}[!t]
  $$
  \hskip-5em
  \scalebox{.85}{$
  \vlderivation{
    \vltrf{\mbox{$\deri$}}{\accs rr,\futs rr,\labelsw rF}{
      \vlhy{\sq H_1\;\ldots\;\sq H_{p-1}\hskip-.5em}}{
      \vlin{\dialoopr\!}{}{\sq G_1}{
        \vltr{\mbox{$\derib$}}{\sq G^\ast}{
          \vlhy{\sq H_p}}{
          \vlhy{\!\ldots\!}}{
          \vlhy{\sq H_q}}}}{
      \vlhy{\hskip-2em\sq H_{q+1}\;\ldots\;\sq H_k}}{1.3}
  }
  =\quad
  \vlderivation{
    \vltrf{\mbox{$\deri''$}}{\accs rr,\futs rr,\labelsw rF}{
      \vlhy{\sq H_1\;\ldots\;\sq H_{l-1}\hskip-5.5em}}{
      \vlin{\bdiam}{}{\sq G_0}{
        \vltrf{\mbox{$\deri'$}}{\sq G_0'}{
          \vlhy{\sq H_l\;\ldots\;\sq H_{p-1}\hskip-.5em}}{
          \vlin{\dialoopr}{}{\sq G_1}{
        \vltr{\mbox{$\derib$}}{\sq G^\ast}{
          \vlhy{\sq H_p}}{
          \vlhy{\!\ldots\!}}{
          \vlhy{\sq H_q}}}}{
          \vlhy{\hskip-2em\sq H_{q+1}\;\ldots\;\sq H_h}}{.5}}}{
      \vlhy{\hskip-6.5em\sq H_{h+1}\;\ldots\;\sq H_k}}{.5}
  }
  \qquad\to
  \vlstemheight=0pt
  \vlderivation{
    \vltrf{\mbox{$\deri''$}}{\accs rr,\futs rr,\labelsw rF}{
      \vlhy{\sq H_1\;\ldots\;\sq H_{l-1}\hskip-6em}}{
      \vlin{\bdiam}{}{\sq G_0}{
        \vltrf{\mbox{$\deri'$}}{\sq G_0'}{
          \vlhy{\sq H_l\;\ldots\;\sq H_{p-1}\hskip-6.5em}}{
          \vlin{\bdiam}{}{\sq G_1}{
            \vltrf{\mbox{$\deri_2'$}}{\sq G'_1}{
              \vlhy{\sq H_{l,2}\;\ldots\;\sq H_{p-1,2}\hskip-8em}}{
              \vlde{}{}{\sq G_2}{
                \vlde{}{}{\vdots}{
                  \vlin{\bdiam}{}{\sq G_{n-1}}{
                    \vltrf{\mbox{$\deri_n'$}}{\sq G'_{n-1}}{
                      \vlhy{\sq H_{l,n}\;\ldots\;\sq H_{p-1,n}\hskip-2em}}{
                      \vltr{\mbox{$\hat\derib$}}{\sq G_n}{
                        \vlhy{\sq K_1}}{
                        \vlhy{\!\ldots\!}}{
                        \vlhy{\sq K_m}}}{
                      \vlhy{\hskip-2em\sq H_{q+1,n}\;\ldots\;\sq H_{h,n}}}{.5}}}}}{
              \vlhy{\hskip-8em\sq H_{q+1,2}\;\ldots\;\sq H_{h,2}}}{.5}}}{
          \vlhy{\hskip-7.5em\sq H_{q+1}\;\ldots\;\sq H_h}}{.5}}}{
      \vlhy{\hskip-7em\sq H_{h+1}\;\ldots\;\sq H_k}}{.5}
  }
  $}
  $$
  \caption{Reduction of a $\dialoopr$-instance}
  \label{fig:dialoop-unfolding}
\end{figure}

\begin{figure}[!t]
  $$
  \hskip-5em
  \scalebox{.85}{$
  \vlderivation{
    \vltrf{\mbox{$\deri$}}{\accs rr,\futs rr,\labelsw rF}{
      \vlhy{\sq H_1\;\ldots\;\sq H_{p-1}\hskip-1.5em}}{
      \vlin{\liftr}{}{\sq G_1}{
        \vlin{\Xloopr}{}{\sq G'_1}{
        \vltr{\mbox{$\derib$}}{\sq G^\ast}{
          \vlhy{\sq H_p}}{
          \vlhy{\!\ldots\!}}{
          \vlhy{\sq H_q}}}}}{
      \vlhy{\hskip-2em\sq H_{q+1}\;\ldots\;\sq H_k}}{1.3}
  }
  \;=\; 
  \vlderivation{
    \vltrf{\mbox{$\deri''$}}{\accs rr,\futs rr,\labelsw rF}{
      \vlhy{\sq H_1\;\ldots\;\sq H_{l-1}\hskip-5.5em}}{
      \vlin{\liftr}{}{\sq G_0}{
        \vltrf{\mbox{$\deri'$}}{\sq G_0'}{
          \vlhy{\sq H_l\;\ldots\;\sq H_{p-1}\hskip-1.75em}}{
          \vlin{\liftr}{}{\sq G_1}{
            \vlin{\Xloopr}{}{\sq G'_1}{
              \vltr{\mbox{$\derib$}}{\sq G^\ast}{
                \vlhy{\sq H_p}}{
                \vlhy{\!\ldots\!}}{
                \vlhy{\sq H_q}}}}}{
          \vlhy{\hskip-2em\sq H_{q+1}\;\ldots\;\sq H_h}}{.5}}}{
      \vlhy{\hskip-6em\sq H_{h+1}\;\ldots\;\sq H_k}}{.5}
  }
  \quad\to\quad
  \vlstemheight=0pt
  \vlderivation{
    \vltrf{\mbox{$\deri''$}}{\accs rr,\futs rr,\labelsw rF}{
      \vlhy{\sq H_1\;\ldots\;\sq H_{l-1}\hskip-6em}}{
      \vlin{\liftr}{}{\sq G_0}{
        \vltrf{\mbox{$\deri'$}}{\sq G_0'}{
          \vlhy{\sq H_l\;\ldots\;\sq H_{q-1}\hskip-6.5em}}{
          \vlin{\liftr}{}{\sq G_1}{
            \vltrf{\mbox{$\deri_2'$}}{\sq G'_1}{
              \vlhy{\sq H_{l,2}\ldots\sq H_{p-1,2}\hskip-7.5em}}{
              \vlin{\liftr}{}{\sq G_2}{
                \vlde{}{}{\sq G_2'}{
                  \vlde{}{}{\vdots}{
                    \vltrf{\mbox{$\deri_n'$}}{\sq G'_{n-1}}{
                      \vlhy{\sq H_{l,n}\ldots\sq H_{p-1,n}\hskip-1.25em}}{
                      \vlin{\liftr}{}{\sq G_n}{
                        \vltr{\mbox{$\hat\derib$}}{\sq G'_n}{
                          \vlhy{\sq K_1}}{
                          \vlhy{\!\ldots\!}}{
                          \vlhy{\sq K_m}}}}{
                      \vlhy{\hskip-2em\sq H_{q+1,n}\ldots\sq H_{h,n}}}{.5}}}}}{
              \vlhy{\hskip-8.25em\sq H_{p+1,2}\ldots\sq H_{h,2}}}{.5}}}{
          \vlhy{\hskip-7.5em\sq H_{p+1}\;\ldots\;\sq H_h}}{.5}}}{
      \vlhy{\hskip-6.5em\sq H_{h+1}\;\ldots\;\sq H_k}}{.5}
  }
  $}
  $$
  \caption{Reduction of a $\Xloopr$-instance, where $\Xloopr$ is one of $\triangleloopr$, $\leafloopr$, $\boxloopr$, or $\imploopr$, and  where $\liftr$ is one of $\liftrim$ or $\liftrbox$.} 
  \label{fig:liftloop-unfolding}
\end{figure}

\subsection{Embeddings of Sequents and Derivations}
\label{sec:embedding}

\begin{notation}
  In $\rr_1$ and $\rr_2$ be two rule instances in a derivation $\deri$. We write $\rr_1\below\rr_2$ or $\rr_2\sabove\rr_1$ iff $\rr_1$ occurs below $\rr_2$ in the derivation tree of $\deri$.
\end{notation}

Assume that we have a $\loopr$-rule instance $\lr$ in a derivation $\deric$, whose conclusion is the nice sequent $\accs rr,\futs rr,\labelsw rF$,
as indicated on the left in Figures~\ref{fig:dialoop-unfolding} and~\ref{fig:liftloop-unfolding}, where $\sq H_1,\ldots,\sq H_k$ are the axiomatic premises of $\deric$.
  If $\lr$ is a $\dialoopr$-instance in a derivation $\deric$, then there is, by the $\DIA$-loop condition, a $\bdiam$-instance $\sr$ 
  with $\sr\below\lr$ such that $\sr$ creates the label $\lb x$ to which the $\dialoopr$-instance loops back. This is shown in the middle derivation in Figure~\ref{fig:dialoop-unfolding}.
We call the rule instance $\sr$ the \defn{seed-rule} of $\lr$.

For all other  $\loopr$-rule instances ($\triangleloopr$, $\leafloopr$, $\boxloopr$, or $\imploopr$), there must be an $\liftrim$- or $\liftrbox$-instance $\sr$ with $\sr\below\lr$, creating the layer that contains the suricata $\lb{s_x}$ that is mentioned in the corresponding loop condition. This is indicated in the middle derivation in Figure~\ref{fig:liftloop-unfolding}.
Again, we call this instance of  $\liftrim$ or $\liftrbox$ the \defn{seed-rule} of $\lr$.\footnote{In Definition~\ref{def:repetition} we used `seed' to denote a sequent, which is the premise of the seed-rule we defined here.}

In this section we show how to construct the derivation shown on the right of Figures~\ref{fig:dialoop-unfolding} and~\ref{fig:liftloop-unfolding}, where $\hderib$ and $n$ are given by Lemma~\ref{lem:unfold}.

\begin{notation}
  Let $\sq H\below\sq G$ and $\lb x,\lb y\in\labelsof G$ with $\layerof x\le\layerof y$ and also $\layerof x\subseteq \labelsof H$. We write $\emb\colon\layerxof xH\to\layerxof yG$ to express that we have a mapping from the layer-tree $\layerof x$ in $\sq H$ into the layer-tree $\layerof y$ in $\sq G$ (see Def.~\ref{def:t-embedding}). 
\end{notation}

\begin{definition}[$\emb$-repetition pair and $\Rloopr$ and $\Sloopr$]
  \label{def:eloop}
  Let $\sq G$ be a nice sequent and let $\lb{x'},\lb y\in\labelsof G$ with $\accsq{x'}y$ and $\lb{x'}\notin \lbc{C_y}$. We say that  $\tuple{\lb {x'},\lb y}$ is an \defn{$\emb$-repetition pair} if there is $\lb x\in \labelsof G$ with $\futsq x{x'}$ such that $\lb y$ is a repetition of $\lb x$ with some seed  $\sq H\psbelow\sq G$ (Def.~\ref{def:repetition}), such that there is an embedding $\emb\colon\layerxof xH\to\layerxof yG$ with $\embof{x}=\lb y$.
With this we define the following two $\loopr$-rules:
  \begin{equation}
    \label{eq:Rloop}
    \begin{array}{l}
    \vlinf{\Rloopr}{
      \proviso{where $\tuple{\lb {x'},\lb y}$ is a $\emb$-repetition pair in $\sq G$, and $\lb y\neq \lb{s_y}$}}{
      \sq G}{
      \sq G \labelsubst{y}{x'}\cup\trans{\interval {x'}y}
    }\\\\
    \vlinf{\Sloopr}{
      \proviso{where $\tuple{\lb {x'},\lb y}$ is a $\emb$-repetition pair in $\sq G$, and $\lb y= \lb{s_y}$}}{
      \sq G}{
      \sq G\cup\trans{\interval {x'}y}
    }
    \end{array}
  \end{equation}
\end{definition}

\begin{lemma}\label{lem:Rloop}
  The rules $\leafloopr$ and $\triangleloopr$ are instances of $\Rloopr$, and the rules $\imploopr$ and $\boxloopr$ are instances of $\Sloopr$.
\end{lemma}

\begin{proof}
  Since  $\layerof x$ in $\sq H$ and $\layerof y$ in $\sq G$ both come from a lifting step, $\lb{s_x}$ and $\lb{s_y}$ contain the only $\rigsym$-formulas in $\sq H$ and $\sq G$ respectively. In the cases of $\imploopr$ and $\boxloopr$ we immediately obtain the embedding $\emb$ from the monotonicity condition, as $\lefmark$-formulas are preserved in the premise.
  In the case of $\triangleloopr$ we additionally observe that $\futsq xy$, which defines $\embof{z}$ for all $\lb z$ with $\accsqq[H]xz$. For $\leafloopr$, this part is trivial as $\lb x$ and $\lb y$ are both leaves. Finally, we cannot have that $\lb{x'}\srel G\lb {s_y}\srel G\lb y$ because $\emb$ embeds $\treeof{x}$ into $\treeof{y}$ and  $\treeof{s_x}$ into $\treeof{s_y}$.
\end{proof}

Instead of $\labISloop$, we can now consider the system:
$$
\labISloopp\;=\;\set{
  \id,\lef\BOT,
  \satr, \bdiam,\liftrim, \liftrbox,
  \dialoopr, \Rloopr, \Sloopr}
$$
 Clearly, any derivation in $\labISloop$ is also a derivation in  $\labISloopp$. 
  
We now extend the notion of \emph{embedding} from layer-trees to sequents (see Definition~\ref{def:t-embedding}).

\begin{definition}[Sequent Embedding]
	\label{def:embedding}
	Let $\sq H$~and~$\sq G$ be nice sequents. A \defn{(sequent) embedding} $\emb\colon\sq H\to\sq G$ is an injective function $\emb\colon\labelsof H\to\labelsof G$ obeying the following conditions,  for all $\lb x,\lb y\in\labelsof H$:
	\begin{enumerate}[($\emb$1),leftmargin=2.2em]
		\item\label{emb:formula} 
		  $\lbG{x}H\lbeq\lbG{{\embof x}}G$;
		\item\label{emb:cluster} 
		  $\lbG{C_x}H\lbeq\lbG{C_{\embof x}}G$;
		\item\label{emb:R}  
		$\accsqq[H]{x}{y}$ if and only if $\accsqq[G]{{\embof x}}{{\embof y}}$; 
		\item\label{emb:layer}  
		$\futsqq[H]{x}{y}$ if and only if $\futsqq[G]{{\embof x}}{{\embof y}}$; 
		\item\label{emb:nopast}  
		If $\lb x$~does not have a past in~$\sq H$  then $\emb(\lb x)$~does not have a past in~$\sq G$;
		\item\label{emb:nofut}  
		If $\lb x$~does not have a future in~$\sq H$  then $\emb(\lb x)$~does not have a future in~$\sq G$;
	\end{enumerate}
\end{definition}

\begin{lemma}
	\label{lem:s-embedding}
	Let $\sq H$~and~$\sq G$ be nice sequents, and let $\emb\colon\labelsof H\to\labelsof G$ be a sequent embedding. Then we have  for all $\lb x,\lb y\in\labelsof H$:
	\begin{itemize}
		\item\label{embC:R}  
		$\accsqq[H]{C_x}{C_y}$ if and only if $\accsqq[G]{C_{\embof x}}{C_{\embof y}}$; 
		\item\label{embC:layer}  
		$\futsqq[H]{C_x}{C_y}$ if and only if $\futsqq[G]{C_{\embof x}}{C_{\embof y}}$.
        \end{itemize}
\end{lemma}

\begin{proof}
  This follows immediately from the definitions (see also Lemma~\ref{lem:t-embedding}).
\end{proof}

\begin{lemma}[First Embedding Lemma]
	\label{lem:embedding:ssat}
  Let $\sq H$ and $\sq G$ be nice sequents such that there is an embedding\/ $\emb\colon\sq H\to\sq G$. If there is a derivation~$\deri$ using only rules in $\set{\lef\AND,\rig\AND,\lef\OR,\rig\OR,\lef\IMP,\lef\BOX,\rig\DIA}$ with endsequent~$\sq H$ and premises\/ $\sq H_1,\ldots,\sq H_n$, then there is a derivation~$\deri'$ in $\set{\lef\AND,\rig\AND,\lef\OR,\rig\OR,\lef\IMP,\lef\BOX,\rig\DIA}$ with conclusion\/~$\sq G$ and premises\/ $\sq G_1,\ldots,\sq G_n$ such that for every $i=1..n$, we have that $\emb\colon\sq H_i\to \sq G_i$ is an embedding. 
\end{lemma}

\begin{proof}
  We proceed by induction on the height of $\deri$. If no rules are applied to $\sq H$, then $\deri$ consists of a single node, $\sq H$, whence  $\emb\colon\sq H\to\sq G$ is given by assumption. 
  For the inductive step, we distinguish cases according to the last rule in $\set{\lef\AND,\rig\AND,\lef\OR,\rig\OR,\lef\IMP,\lef\BOX,\rig\DIA}$  applied in $\deri$. Since none of these rules  introduces (bottom-up) fresh variables, all the cases are treated similarly. We only show the case in which the last rule applied to $\sq H$ is $\lef{\OR}$. The derivation $\deri$ is as follows: 
  $$
  \vlderivation{
  \vliin{\lef{\OR}}{}{\sq H, \labelsb{x}{A \OR B}}{
  \vltr{\deri_1}{\sq H, \labelsb{x}{A \OR B}, \labelsb x A}{\vlhy{\quad}}{\vlhy{}}{\vlhy{}}
	}{
\vltr{\deri_2}{\sq H, \labelsb{x}{A \OR B}, \labelsb x B}{\vlhy{}}{\vlhy{\quad }}{\vlhy{}}	
	}
	}
  $$
  By assumption, there is an embedding $\emb\colon\sq H\to\sq G$. By definition of embedding, there is a label $\embof{x} \in \labelsof{G}$ such that $\lb x\lbeq \embof{x}$, whence $\labels{\embof{x}}{A \OR B}$ occurs in $\sq G$. We construct the following derivation $\deri'$: 
    $$
  \vlderivation{
  	\vliin{\lef{\OR}}{}{\sq G, \labelsb{\embof{x}}{A \OR B} }{
  		\vltr{\deri'_1}{\sq G, \labelsb{\embof{x}}{A \OR B}, \labelsb{\embof x} A}{\vlhy{\quad}}{\vlhy{}}{\vlhy{}}
  	}{
  		\vltr{\deri'_2}{\sq G, \labelsb{\embof x}{A \OR B}, \labelsb{ \embof x} B}{\vlhy{}}{\vlhy{\quad }}{\vlhy{}}	
  	}
  }
  $$
  We now verify that $\emb$ is an embedding between the premises of $\lef\OR$ in $\deri $ and $\deri'$, i.e.
  $$\emb\colon\sq H, \labelsb{x}{A \OR B}, \labelsb x A \rightarrow  \sq G, \labelsb{\embof{x}}{A \OR B}, \labelsb{\embof x} A
  \quand
  \emb\colon\sq H, \labelsb{x}{A \OR B}, \labelsb x B \rightarrow  \sq G, \labelsb{\embof{x}}{A \OR B}, \labelsb{\embof x} B
  \;.
  $$
  Since no new label is introduced in the premise of $\lef\OR$, all the conditions immediately follow. We conclude by applying the inductive hypothesis to the premises of $\lef\OR$ in $\deri$. 
\end{proof}

\begin{construction}\label{con:lay-to-seq}
  Let $\deri$ be a nice derivation in $\labISloopp$, and $\sq G$ and $\sq H$ be two sequents occurring in $\deri$ with $\sq G\below\sq H$. Let $L_1$ be a topmost layer of $\sq G$ and $L_2$ be a topmost layer of $\sq H$ such that $L_1\le L_2$ in $\sq H$ and such that there is a layer embedding $\emb\colon \layerx 1G\to\layerx 2H$. Then this embedding can be extended to a sequent embedding $\embh\colon\sq G\to\sq H$ by letting $\embofh w=\embof w$ if $\lb w\in L_1$ and $\embofh w=\lb w$ otherwise. 
\end{construction}

Let us now come back to the derivation in the middle of Figure~\ref{fig:dialoop-unfolding}, where we singled out a (topmost) instance of $\dialoopr$ and its seed $\bdiam$:
\begin{equation}
  \label{eq:dia-loop+seed}
  \vlinf{\dialoopr}{}{
    \sq G}{
    \sq G \labelsubst{y}{x}\cup\trans{\interval xy}
  } 
  \qquad
  \vlinf{\bdiam}{\proviso{$\lb z$  fresh}}{\sq G, \labelsb x{\DIA A} }{\sq G, \labelsb x{\DIA A} \cup \{ \accs xz, \accs zz,\futs zz,  \labelsb zA\} \cup \{ \accs v z \mid \accs v x \in \sq G\} }
\end{equation}
Employing the terminology from Figure~\ref{fig:dialoop-unfolding}, let $ \sq G_0$ and $\sq G_0'$ be the conclusion and premise of $\bdiam$ respectively, and $\sq G_1$ and $\sq G_1'$ be the conclusion and premise of $\dialoopr$ respectively. 
Observe that during proof search using Algorithm~\ref{alg:proofsearch} both rule applications above must happen in the same saturation step. Let $L$ be the layer which contains the fresh label $\lb z$ created by $\bdiam$ in $\sq G_0'$. Then, the cluster created in the premise of rule $\dialoopr$ 
is in the same layer $L$. In fact, we have $\sq G_0\below\sq G_1$ and a (layer) embedding $\emb\colon\layerxxof{}{\sq G_0}\to\layerxxof{}{\sq G_1}$ with $\embof x=\lb y$, where $\lb x$ and $\lb y$ are as indicated in~\eqref{eq:dia-loop+seed} above. We extend this embedding to a sequent embedding $\embh\colon\sq G_0\to\sq G_1$, following Construction~\ref{con:lay-to-seq} above.

If we look at the derivation in the middle of Figure~\ref{fig:liftloop-unfolding}, we have a similar situation. By Lemma~\ref{lem:Rloop}, each $\Xloopr$ in that figure is and instance of $\Rloopr$ or $\Sloopr$. Let $\layerof x$ be the layer in $\sq G_0'$ created by the lower $\liftr$-instance and let $\layerof y$ be the layer in $\sq G_1'$ created by the upper $\liftr$-instance. Then we have $\layerof x\le\layerof y$ in $\sq G_1'$ and we have a (layer) embedding $\emb\colon\layerxxof x{\sq G_0'}\to\layerxxof y{\sq G_1'}$ with $\embof x=\lb y$, where $\lb x$ and $\lb y$ are as in Definition~\ref{def:eloop}. As before, this can be extended to  a sequent embedding $\embh\colon\sq G_0'\to\sq G_1'$ as indicated in Construction~\ref{con:lay-to-seq} above.

\begin{lemma}[Second Embedding Lemma]
  \label{lem:embedding-b}
  Let $\deri$ be a nice derivation in $\labISloopp$, let $\sq G$ be a sequent occurring in $\deri$, let $\deri'$ be the subsequent of $\deri$ rooted at $\sq G$, let $\sq G_1,\ldots,\sq G_n$ be the premises of $\deri'$, and let $\sq H=\sq G_i$ for some $i=1..n$. Furthermore, assume $L_1$ is a topmost layer in $\sq G$, and $L_2$ is a topmost layer in $\sq H$, such that $L_1\le L_2$ in $\sq H$, and such that we have a layer embedding $\emb\colon\layerx 1G\to\layerx 2H$. Then there is a derivation $\deri_2'$ in $\labISloopp$ with conclusion $\sq H$ and premises  $\sq H_1,\ldots,\sq H_n$, such that there are sequent embeddings $\embh_i\colon \sq G_i\to\sq H_i$, which all extend the embedding $\embh\colon \sq G\to\sq H$ determined by $\emb\colon\layerx 1G\to\layerx 2H$ and Construction~\ref{con:lay-to-seq}.
\end{lemma}

\begin{proof}
  The basic idea is to apply in $\deri'_2$ every rule applied  in $\deri'$, thus `copying' $\deri'$ into $\deri'_2$. Moreover, $\deri'_2$ carries around all the extra labels that are in 
  	 $\sq H$ but not in $\sq G$, and the embedding tells us to which of the `new' labelled formula a rule needs to be applied.
  For this, we proceed in two steps. First, we construct $\deri'_2$ by induction on $\deri'$, and by distinguishing cases according to the last rule applied in $\deri'$. In this first step, we do not verify the loop conditions for $\Rloopr$, and $\Sloopr$ (we only apply the rules). Then, after having built $\deri'_2$, we confirm that the side conditions are satisfied.
  \begin{itemize}
  \item Case $\satr$:
    $$
    \vliiinf{\satr}{}{\sq G}{\sq G_1}{\ldots}{\sq G_n}
    $$
    Suppose that there is an embedding $\embh\colon\sq G\to\sq H$. The premises $\sq G_1, \dots, \sq G_n$ of $\sq G$ are obtained by applying rules in    $\set{\lef\AND,\rig\AND,\lef\OR,\rig\OR,\lef\IMP,\lef\BOX,\rig\DIA}$ to $\sq G$. We can therefore apply Lemma~\ref{lem:embedding:ssat}, and obtain a derivation of $\sq H$ whose premises $\sq H_1,\ldots,\sq H_n$ are such that for every $i=1..n$, we have that $\embh\colon\sq G_i\to \sq H_i$ is an embedding. We set $\embh_i = \embh$, for $i=1..n$. We can then conclude by inductive hypothesis. 
		\item Case $\liftrbox$:
		$$
		\vlinf{\liftrbox}{
		}{
			\sq G, \labelsw{x_0}{\BOX A}}{
			\sq G, \labelsw{x_0}{\BOX A}\cup\lliftfw {G}{x_0}{\BOX  A}
		} 
		$$
		Suppose that there is an embedding $\embh\colon\sq G, \labelsw{x_0}{\BOX A}\to\sq H, \labelsw{\embof{x_0}}{\BOX A}$, with $\embofh{x_0} = \lb{w_0}$. We can therefore apply rule $\liftrbox$ to sequent $\sq H, \labelsw{w_0}{\BOX A}$, obtaining the following sequent, which we shall prove to be derivable by induction hypothesis: 
		$$
		\vlinf{\liftrbox}{
		}{
			\sq H, \labelsw {w_0} {\BOX A}}{
			\sq H, \labelsw {w_0}{\BOX A}\cup\lliftfw {H}{w_0}{\BOX  A}
		} 
		$$
		Recall that, if $\layerof{x_0} = \{  \lb{x_0}, \lb{x_1}, \dots, \lb{x_l}\}$, for some $l \geq 0$, the sequent $\lliftfw {G}{x_0}{\BOX  A}$ contains $l+1$ fresh labels, namely $\layerof{\hx_0}  = \{ \lb{\hx_0}, \lb{\hx_1}, \dots, \lb{\hx_l},\lb\hy\}$. 
		Similarly, in $\sq H$,  if $\layerof{{w_0}} = \{ \lb{w_0}, \lb{w_1}, \dots, \lb{w_m}\}$, for some $m \geq l$, then the sequent $\lliftfw {H}{w_0}{\BOX  A}$ contains $m + 1$ fresh labels,  namely $\layerof{\hw_0}  = \{ \lb{\hw_0}, \lb{\hw_1}, \dots, \lb{\hw_m},\lb{\hat v}\}$. 
		
		We define the embedding $\embh_1\colon \sq G, \labelsw{x_0}{\BOX A}\cup\lliftfw {G}{x_0}{\BOX  A} \rightarrow 	\sq H, \labelsw {w_0}{\BOX A}\cup\lliftfw {H}{w_0}{\BOX  A}$ as follows. For any label $\lb z $ such that $\lb z\in \labelsof{G}$, we set $\embofhn{1}{z} = \embofh{z}$. For the fresh labels, we set $\embofhn{1}{\hy} = \lb{\hat v}$ and, for any label $\lb{ \hx_i} \in \layerof{\hx_0}$ (for $i \leq l$), we let $\embofhn{1}{\hx_i} = \lb{\hw_j}$, where $\lb{\hw_j}$ is the fresh label in $\layerof{\hw_0}$ (for $j \leq m$) such that $\embofh{x_i}= \lb{w_j}$. 

		To conclude this case, we need to verify that $\emb_1$ is indeed an embedding. All the conditions immediately follow from the definition of $\rig\BOX$-lifting (Construction~\ref{def:box-lifting}), and the fact that $\embh_1$ is an extension of $\embh$. 

		\item Case $\liftrim$:
		$$
		\vlinf{\liftrim}{
		}{
			\sq G, \labelsw{x_0}{A \IMP B}}{
			\sq G,  \labelsw{x_0}{A \IMP B}\cup\lliftfw {G}{x_0}{A\IMP B}
		} 
		$$
		We distinguish two cases depending on whether $\lb x_0$ belongs to a singleton cluster or to a non-singleton cluster. If  $\lb x_0$ belongs to a singleton cluster, we proceed similarly to the case for~$\liftrbox$. 
		Now assume $\lb{x_0}$ belongs to a cluster $\clusterof{x_0} = \{ \lb{x_0}, \lb{x_1}, \dots, \lb{x_h}\}$, for $h \geq 1$.  
		Suppose that there is an embedding $\embh\colon\sq G, \labelsw{x_0}{A\IMP B}\to\sq H, \labelsw{\embofh{x_0}}{A \IMP B}$, and that $\embofh{x_0} = \lb{w_0}$. First observe that if $\lb{x_0}$ belongs to a non-singleton cluster $\clusterof{x}$ in $\sq G$, then $\lb{w_0}$ will also belong to a non-singleton cluster $\clusterof{w_0} = \{ \lb{w_0}, \lb{w_1}, \dots, \lb{w_k}\}$, for $k \geq h$, in $\sq H$. We can therefore apply the cluster $\rig\IMP$-lifting (Construction~\ref{def:imp-lifting}) to $\sq H$, obtaining the following derivation:  
			$$
		\vlinf{\liftrim}{
		}{
			\sq H, \labelsw{w_0}{A \IMP B}}{
			\sq H,  \labelsw{w_0}{A \IMP B}\cup\lliftfw {H}{w_0}{A\IMP B}
		} 
		$$
		Recall that, in Construction~\ref{def:imp-lifting}, we ``duplicate'' the cluster $\clusterof{x_0} $. So the fresh labels in $\lliftfw {G}{x_0}{A\IMP B}$ are $ \set{\lb{\hy_1},\ldots, \lb{\hy_l},\lb{\hx_0}, \lb{\hx_0'},\ldots, \lb{\hx_h'},\lb{\hx_0''},\ldots, \lb{\hx_h''}}$, where $\set{\lb{y_1},\ldots,\lb{y_l}}=\layerof{x_0}\setminus \clusterof{x_0}$, for $l\geq 0$. Below is a picture, for $h = 2$: 
		 
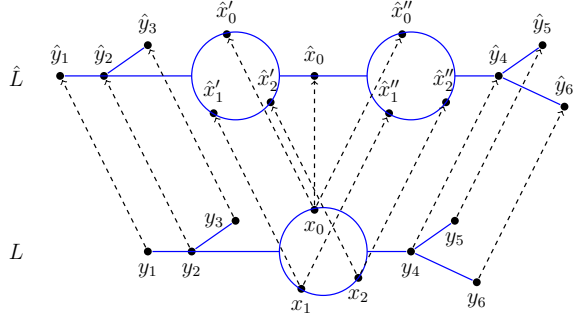
\begin{center}
				\begin{tikzpicture}[thick, every node/.style={scale=1.2},font = {\large}]

		\tikzstyle{node}=[circle,fill=black,inner sep=1.5pt]
		\tikzstyle{nonode}=[inner sep=0pt]
		\tikzstyle{access}=[blue]
		\tikzstyle{future}=[dashed,->]
		
		\node[] (Lh) at (-1,6) {$\hat L$};
		\node[] (L) at (-1,2) {$L$};
		
		\node[label=below:{$x_0$}] (x2) at (5.8,2.95) [node] {};
		\node[label=below:{$x_1$}] (x1) at (5.5,1.15) [node] {};
		\node[label=below:{$x_2$}] (x3) at (6.8,1.4) [node] {};
		\node[label=below:{$y_1$}] (y1) at (2,2) [node] {};
		\node[label=below:{$y_2$}] (y2) at (3,2) [node] {};
		\node[label=left:{$y_3$}] (y3) at (4,2.7) [node] {};
		\node[label=below:{$y_4$}] (y4) at (8,2) [node] {};
		\node[label=below:{$y_5$}] (y5) at (9,2.7) [node] {};
		\node[label=below:{$y_6$}] (y6) at (9.5,1.3) [node] {};
		%

		\node[label=above:{$\hat x_1'$}] (x1') at (3.5,5.15) [node] {};
		\node[label=above:{$\hat x_1''$}] (x1'') at (7.5,5.15) [node] {};
		\node[label=above:{$\hat x_0'$}] (x2') at (3.8,6.95) [node] {};
		\node[label=above:{$\hat x_0$}] (xhat) at (5.8,6) [node] {};
		\node[label=above:{$\hat x_0''$}](x2'') at (7.8,6.95) [node] {};
		\node[label=above:{$\hat x_2'\:$}] (x3') at (4.8,5.4) [node] {};
		\node[label=above:{$\hat x_2''\:$}] (x3'') at (8.8,5.4) [node] {};
		%
		\node[label=above:{$\hat y_1$}] (y1h) at (0,6) [node] {};
		\node[label=above:{$\hat y_2$}] (y2h) at (1,6) [node] {};
		\node[label=above:{$\hat y_3$}] (y3h) at (2,6.7) [node] {};
		\node[label=above:{$\hat y_4$}] (y4h) at (10,6) [node] {};
		\node[label=above:{$\hat y_5$}] (y5h) at (11,6.7) [node] {};
		\node[label=above:{$\hat y_6$}] (y6h) at (11.5,5.3) [node] {};
		%
			
		\draw[access] (y1) -- (5,2);
		\draw[access] (y2) -- (y3);
		\draw[access] (6,2) circle [radius=1cm];
		\draw[access] (7,2) -- (y4);
		\draw[access] (y4) -- (y5);
		\draw[access] (y4) -- (y6);
		
		\draw[access] (y1h) -- (3,6);
		\draw[access] (y2h) -- (y3h);
		\draw[access] (4,6) circle [radius=1cm];
		\draw[access] (5,6) -- (7,6);
		\draw[access] (8,6) circle [radius=1cm];
		\draw[access] (9,6) -- (y4h);
		\draw[access] (y4h) -- (y5h);
		\draw[access] (y4h) -- (y6h);
		
		\draw[future] (x1) -- (x1');
		\draw[future] (x1) -- (x1'');
		\draw[future] (x2) -- (x2');
		\draw[future] (x2) -- (xhat);
		\draw[future] (x2) -- (x2'');
		\draw[future] (x3) -- (x3');
		\draw[future] (x3) -- (x3'');
		\draw[future] (y1) -- (y1h);
		\draw[future] (y2) -- (y2h);
		\draw[future] (y3) -- (y3h);
		\draw[future] (y4) -- (y4h);
		\draw[future] (y5) -- (y5h);
		\draw[future] (y6) -- (y6h);

	\end{tikzpicture}
		\end{center}
		Similarly, the sequent $\lliftfw {H}{w_0}{A\IMP B}$ contains a duplicate of the cluster  $\clusterof{w_0}$, where the fresh labels are $ \set{\lb{\hz_1},\ldots, \lb{\hz_m},\lb{\hw_0}, \lb{\hw_0'},\ldots, \lb{\hw_k'},\lb{\hw_0''},\ldots, \lb{\hw_k''}}$, where $\set{\lb{z_1},\ldots,\lb{z_m}}=\layerof{w_0}\setminus \clusterof{w_0}$, for some $m\geq l$. 
		
		We define $\embh_1: 	\sq G,  \labelsw{x_0}{A \IMP B}\cup\lliftfw {G}{x_0}{A\IMP B} \rightarrow 			\sq H,  \labelsw{w_0}{A \IMP B}\cup\lliftfw {G}{w_0}{A\IMP B}$ as follows: 
		\begin{itemize}
			\item For any  label $\lb v $ such that $\lb v\in \labelsof{G}$, we set $\embofhn{1}{v} = \embofh{v}$; 
			\item For any label $\lb {\hy_i} \in \lliftfw {G}{x_0}{A\IMP B}$ for $i = 1..l$, we set $\embofhn{1}{\hy_i} = \lb {\hz_j}$, where $\lb {\hz_j}$ is such that $\embofh{y_i}=\lb{z_j}$, for some $j=1..m$; 
			\item For any label $\lb {\hx'_i} \in \lliftfw {G}{x_0}{A\IMP B}$ for $i = 1..h$, we set $\embofhn{1}{\hx'_i} =\lb{ \hw'_j}$, where $\lb{ \hw'_j}$ is such that  $\embofh{x_i}=\lb{w_j}$, for some $j=1..k$; 
			\item For any label $\lb {\hx''_i} \in \lliftfw {G}{x_0}{A\IMP B}$ for $i = 1..h$, we set $\embofhn{1}{\hx''_i} =\lb{ \hw''_j}$, where $\lb{ \hw''_j}$ is such that  $\embofh{x_i}=\lb{w_j}$, for some $j=1..k$; 
			\item $\embofhn{1}{\hx_0} = \lb{\hw_0}$.
		\end{itemize}
		To conclude the proof for this case, we need to verify that $\embh_1$ is indeed an embedding. But it is easy to see that Conditions~\ref{emb:formula}--\ref{emb:nofut} follow immediately.
		\item Case 	$\bdiam$: 
		$$
		\vlinf{\bdiam}{\proviso{where $\lb y$  fresh}}{\sq G, \labelsb x{\DIA A} }{\sq G, \labelsb x{\DIA A} \cup \{ \accs xy, \accs y y, \futs y y, \labelsb yA\} \cup \{ \accs v y \mid \accs v x \in \sq G\} }
		$$
		Suppose that there is an embedding $\embh\colon\sq G, \labelsb x{\DIA A}\to\sq H, \labelsb{\embofh{x}}{\DIA A}$. We construct the following derivation, where $\lb {y'}$ is fresh, by applying rule $\bdiam$ to formula $\labelsb{\embofh{x}}{\DIA A}$:
		$$
		\vlderivation{
			\vlin{\bdiam}{}{\sq H, \labelsb{\embofh{x}}{\DIA A}}{
					\vlhy{\sq H', \labelsb{\embofh{x}}{\DIA A}, \accs{\embofh{x}}{y'}, \accs{y'}{y'}, \futs{ y'}{y'}, \labelsb{y}{A}  \cup \{ \accs v {y'} \mid \accs{v}{\embofhn{}{x}}  \in \sq H\}  }
				}
		}
		$$
        Let $\sq G_1$ be the premise of sequent $\sq G, \labelsb x{\DIA A}$, and let $\sq H_1$ be the premise of sequent $\sq H, \labelsb{\embofh{x}}{\DIA A}$. 
		We define the embedding $\embh_1  : \sq G _1 \rightarrow \sq H_1$ as follows. Let $\embofhn{1}{x} = \embofh{x}$ and, for any other label $\lb w$ occurring in $\sq G$, let $\embofhn{1}{w} = \embofh{w}$. Finally, take $\embofhn{1}{y} =\lb{ y'}$. 
		It is easy to check that $\embh_1$ satisfies all properties of embedding from Definition~\ref{def:embedding}. 
		
		\item Case $\dialoopr$: 
		$$
		 \vlinf{\dialoopr}{
			\proviso{where $\lb y$ is a $\DIA$-repetition of $\lb x$ in $\sq G$}}{
			\sq G}{
			\sq G \labelsubst{y}{x}\cup\trans{\interval xy}
		} 
		$$	
		Suppose that there is an embedding $\embh\colon\sq G\to\sq H$. Since rule $\dialoopr$ has been applied to $\sq G$, it holds that $\lb y$ is a $\DIA$-repetition of $\lb x$ in $\sq G$. By definition, this means that $\lb y$ is a repetition of $\lb x$, that $\lb x \in \layerof{{y}}$, and that $\lb{x}$ does not have a past in $\sq G$. In order to apply rule $\dialoopr$ to the labels $\embofh{x}$ and $\embofh{y} \in \sq H$, we need to check that $\embofh{y}$ is a $\DIA$-repetition of $\embofh{x}$ in $\sq H$: we have  $\embofh{y}\lbeq\lb y\lbeq \lb x\lbeq\embofh x$, and $\embofh{{x}} \in \layerof{\embofh{y}}$ because $\layerof{x} \sleq G \layerof{y}$, whence $\layerof{\embofh x} \sleq H \layerof{\embofh y}$, and we have $\embofh{x}$ does not have a past in $\sq H$ by \ref{emb:nopast}.

		We can therefore apply the following rule: 
		$$
		\vlinf{\dialoopr}{
		}{
			\sq H}{
			\sq H \labelsubst{\embofh{y}}{\embofh{x}}\cup\trans{\interval{\embofh{x}}{\embofh{y}}}
		} 
		$$	
		Next, we define an embedding $\embh_1: \sq G \labelsubst{y}{x}\cup\trans{\interval xy} \rightarrow 			\sq H \labelsubst{\embofh{y}}{\embofh{x}}\cup\trans{\interval{\embofh{x}}{\embofh{y}}}$. Since no new labels are introduced when applying rule  $\dialoopr$ to $\sq H$, we set $\embh_1 = \embh \setminus \set{ \tuple{\lb y, \embofh y}}$. 
		It is immediate to check that $\embh_1$ is indeed the desired embedding. 
		
		\item Cases $\Rloopr$ and $\Sloopr$: 
		$$
		  \vlinf{\Rloopr}{}{
		    \sq G}{
		    \sq G \labelsubst{y}{x'}\cup\trans{\interval {x'}y}
		  }
                  \qquand
		\vlinf{\Sloopr}{}{
		  \sq G}{
		  \sq G\cup\trans{\interval {x'}y}
		} 
		$$
                We have an embedding $\embh\colon\sq G\to\sq H$, and apply the rules as follows:
		$$
		\vlinf{\Rloopr}{
		}{
		  \sq H}{
		  \sq H \labelsubst{\embofh{y}}{\embofh{x'}}\cup\trans{\interval{\embofh{x'}}{\embofh{y}}}
		}
                \qquand
		\vlinf{\Sloopr}{
		}{
		  \sq H}{
		  \sq H\cup\trans{\interval{\embofh{x'}}{\embofh{y}}}
		}
		$$	
                In the first case, we have $\embh_1 = \embh \setminus \set{ \tuple{\lb y, \embofh y}}$ and in the second case we have $\embh_1 = \embh$. 
                We can therefore conclude by induction hypothesis. 

  \end{itemize}
  This finishes the construction of $\deri'_2$. Let us now verify the correct applications of $\Rloopr$ and $\Sloopr$ in $\deri'_2$. 
  We only verify $\Rloopr$, the case for $\Sloopr$ being similar. 
  Let 
  $$
  \vlinf{\Rloopr}{}{
  \sq K}{
    \sq K \labelsubst{y}{x'}\cup\trans{\interval {x'}y}}
  $$
  be an instance of $\Rloopr$ in $\deri'$, and let
  $$
  \vlinf{\Rloopr}{
  }{
    \sq L}{
    \sq L \labelsubst{\embofhn1{y}}{\embofhn1{x'}}\cup
    \trans{\interval{\embofhn1{x'}}{\embofhn1{y}}}
  }
  $$
  be the corresponding instance of $\Rloopr$ in $\deri'_2$, where $\embh_1\colon\sq K\to\sq L$ is the embedding determined by the construction above. 
  We have that $\sq G= \sq K$ or $\sq G\below\sq K$, and that $\sq H = \sq L$ or  $\sq H\below \sq L$. 
   and we have that $\tuple{\lb {x'},\lb y}$ is an $\emb$-repetition pair in $\sq K$. By Definition~\ref{def:eloop}, this means that there is a seed $\sq K'\below \sq K$ in $\deri$, and we have a layer $\layerof x$ in $\sq K'$ with a layer embedding $\emb\colon\layerxxof x{\sq K'}\to\layerxof yK$ with $\embof x=\lb y$. 
   We have two possible cases:
  \begin{itemize}
  \item $\sq K'\below\sq G$.  Then we also have $\sq K'\below \sq L$ because $\sq G\below\sq H\below\sq L$. 
  We have a layer embedding $\emb'\colon \layerxxof x{\sq K'}\to\layerxxof{\embofhn1y}{\sq L}$.  
  This follows by composing two layer embeddings, namely $\emb\colon\layerxxof x{\sq K'}\to\layerxof yK$, which we have by assumption, and $\embh_1 \colon \layerxof yK =  \layerxxof {x'}{\sq K}\to\layerxof {\embofhn1{y}}L$. This latter layer embedding is the restriction of the sequent embedding $\embh \colon \sq K \to \sq L$ which we constructed in our inductive proof above. It is easy to see that the composition of two layer embeddings is also a layer embedding, and so we get $\emb'\colon \layerxxof x{\sq K'}\to\layerxxof{\embofhn 1 y}{\sq L}$.  
 This layer embedding can be extended to a sequent embedding, and we therefore obtain that  $\embofhn1{x'}$ is a future of $\lb x$, as sequent embeddings preserve futures, as guaranteed by \ref{emb:layer} of Def.~\ref{def:embedding}. Hence, $\tuple{\embofhn1{x'},\embofhn1{y}}$ is an $\emb$-repetition pair in $\sq L$. 
  \item $\sq K'=\sq G$ or $\sq G\below\sq K'$. Then there must be an $\sq L'$ with $\sq H\below\sq L'\below\sq L$. 
  Furthermore, there is an embedding $\embh_2\colon\sq K'\to\sq L'$, which has been constructed in the first half of this proof (and which on labels is a restriction of the sequent embedding $\embh_1 \colon \sq K \to \sq L$ which we have by construction. Therefore,  the embedding $\embh_2\colon\sq K'\to\sq L'$ is also an extension of the embedding  $\embh\colon\sq G\to\sq H$. 
  We now need to construct a layer embedding $\emb'\colon\layerxxof{\embofhn2x}{\sq L'}\to\layerxxof{\embofhn1y}{\sq L}$. For this, let $\lb z\in \layerxxof{\embofhn2x}{\sq L'}$. If there is a $\lb w\in\layerxof{x}{\sq K'}$ with $\lb z= \embofhn2w$ then we let $\embofp z=\embofhn1{\embof w}$, i.e., we apply the composition of $\emb\colon\layerxxof x{\sq K'}\to\layerxof yK$ and $\embh_1 \colon \sq K \to \sq L$ to $\lb w$. If there is no such $\lb w$, then we can define $\embofp z$ to be the unique future of $\lb z$ in $\layerof{\embofhn1y}$. Indeed, in this case, $\lb z$ is a future of a label in $\sq H$ that is not in the image of the embedding $\embh\colon \sq G\to\sq H$, and therefore no rule in $\deri'$ changes the number of futures (which can only happen in a  $\liftrim$-rule).
  \qedhere
  \end{itemize}
\end{proof}
The important observation here is that if the seed $\sq K'$ is below $\sq G$ in the derivation, then both instances of $\Rloopr$ (or $\Sloopr$), the one in $\deri$ and the one in $\deri$, have the same seed. However, if the seed $\sq K'$ of the $\Rloopr$ (or $\Sloopr$) in $\deri$ also occurs inside $\deri$, then the seed $\sq L'$ of the  $\Rloopr$ (or $\Sloopr$) in $\deri'$ also occurs inside $\deri'$. This will play a crucial role in the termination proof in the next section.

\subsection{Loop Reduction}
\label{sec:loop reduction}

\begin{definition}
   A $\loopr$-instance $\lr$  in a derivation $\deri$ is called \defn{upmost} iff there is no $\lr'$ with $\lr\below\lr'$. 
\end{definition}

\begin{construction}[Single Loop Reduction]\label{con:singleloop}
  Let us now come back to Figures~\ref{fig:dialoop-unfolding} and~\ref{fig:liftloop-unfolding}.  In both cases we can by Lemma~\ref{lem:embedding-b} construct the subderivation $\deri'_2$ on the right of each figure from the subderivation $\deri'$ in the middle. We can repeat the same process to get $\deri'_3$ from $\deri'_2$, and so on until we construct $\deri'_n$, where $n\ge0$ is given to us by Lemma~\ref{lem:unfold}. That lemma also allows us to replace $\derib$ by $\hderib$ because in Figure~\ref{fig:dialoop-unfolding} we have that $\sq G_n\topeq\sq G_1\interval xy^n$, where $\lb x$ and $\lb y$ are as in the $\dialoopr$-instance (see~\eqref{eq:dialoop}) to be eliminated, and in Figure~\ref{fig:liftloop-unfolding} we have that $\sq G'_n\topeq\sq G'_1\interval{x'}y^n$, where $\lb x'$ and $\lb y$ are as in the $\Rloopr$- and $\Sloopr$-instance (see~\eqref{eq:Rloop}) to be eliminated. Thus, we have shown how to eliminate an upmost $\loopr$-instance from the derivation.
\end{construction}

However, observe that in both cases, all $\loopr$-instance that occur in $\deri'$ are now $n$-times duplicated. Nonetheless, by using a \emph{nested multiset ordering}~\cite{dershowitz1979proving}, we can show that all $\loopr$-instances in the derivation can be eliminated in the way described in~Construction~\ref{con:singleloop}. 

\begin{definition}
  We call an inference rule instance $\sr$ in a derivation $\deri$ a \defn{seed-rule} if it is the seed-rule for a $\loopr$-rule instance. For a seed-rule $\sr$, we define $\numberof\sr$ to be the number of $\loopr$-rule instances for which $\sr$ is the seed-rule.
  Let $\lr$ be a $\loopr$-rule instance and $\sr$ be its seed, and let $\sr'$ be another seed rule (for a different $\loopr$-rule instance $\lr'$).
   We write $\sr\sdominv\sr'$ iff  $\sr\below\sr'\below\lr$ in the derivation tree~$\delta$.
  For two seed-rules $\sr_1$ and $\sr_2$, we write $\sr_1\lll\sr_2$ iff $\sr_1\sdominv\sr_2$ and there is no seed-rule $\sr_3$ with   $\sr_1\below\sr_3$ and $\sr_3\sdominv\sr_2$. We can now define the \defn{rank} of a seed-rule $\sr$, denoted by $\rank\sr$, to be the multiset containing the value $\numberof\sr$ and the ranks of all seed-rules $\sr'$ with $\sr\lll\sr'$, i.e.,
  \begin{equation}
    \label{eq:ranks}
    \rank\sr\;=\;\mset{\numberof\sr}~\uplus~\mset{~\rank{\sr'}\mid \sr\lll\sr'~}
  \end{equation}
  A seed-rule $\sr$ is \defn{lowermost} iff there is no seed-rule $\sr'$ with $\sr'\sdominv \sr$. 
  With this, we can define the rank of the derivation~$\deri$:
   \begin{equation}
    \label{eq:rankd}
    \rank\deri\;=\;\mset{~\rank{\sr}\mid \mbox{$\sr$ is a lowermost seed-rule in $\deri$}~}
  \end{equation}
The ranks are nested multisets and subject to the nested multiset ordering~\cite{dershowitz1979proving}.
\end{definition}

The relation $\sr_1\sdominv\sr_2$ is visualised in Figure~\ref{fig:lasagna}, where $\sr_1$ is seed-rule for $\lr_1$ and $\sr_2$ is seed-rule for $\lr_2$.

\begin{figure}
	\begin{center}
		\includegraphics[scale =0.65]{figures/lasagna-bis.tex}
	\end{center}
	\caption{Three possibilities how $\sr_1\sdominv\sr_2$ can happen, where $\lr_1$ and $\lr_2$ are $\loopr$-rule instances for which $\sr_1$ and $\sr_2$ are seed-rules (respectively).
	}
	\label{fig:lasagna}
\end{figure}

\begin{theorem}[Loop Elimination]
  \label{thm:unfold}
  Let $\sq G$ be the sequent $\futs rr, \accs rr,\labelsw rF$ for some formula $\fm F$. If there is a $\labISloop$-proof~$\der$ with endsequent $\sq G$, then there is also a $\labISnoloop$-proof $\hder$ of~$\sq G$.\looseness=-1
\end{theorem}

\begin{proof}
  We show that if $\hderi$ is obtained from $\deri$ by applying Construction~\ref{con:singleloop} to an upmost $\loopr$-rule instance $\lr_0$, then $\rank\hderi<\rank\deri$. Let $\sr_0$ be the seed-rule for $\lr_0$. There are two cases:
  \begin{itemize}
  \item First, assume $\sr_0$ is lowermost. If $\numberof{\sr_0}=1$ then $\lr_0$ is the only $\loopr$-rule with $\sr_0$ as seed-rule, and $\rank{\deri} = \mset{ \mset{ 1, {X_1},\ldots {X_h}} } \uplus {Z} $, where $ \numberof{\sr_0}=1$ and $X_1, \ldots, X_h, Z$ are nested multisets. 
  Then,  in $\rank\hderi$, the element $\rank{\sr_0}$ is removed and replaced by $n$ copies of the elements of $\rank{\sr_0}$. Namely: 
  \[
  \rank\hderi = \mset{ \overbrace{{X_1},\ldots, {X_h}, \ldots, {X_1},\ldots, {X_h}}^{n \text{ times}} }  \uplus {Z}
  \]
	Every element $X_i$ for $i=1..h$ of  $  \rank\hderi $ is smaller than the element $ \mset{ 1, {X_1},\ldots, {X_h}} $ of $  \rank\deri $, whence $    \rank\deri >  \rank\hderi $.  
  If $\numberof{\sr_0}\ge2$ then in $\rank\hderi$, the element $\rank{\sr_0}$ is removed and replaced by $n$ copies of $\rank{\sr_0}$, where in each copy the value $\numberof{\sr_0}$ has decreased by 1. 
    In both cases we have $\rank\hderi<\rank\deri$.
  \item Let us now assume there is a $\sr_1$ in $\deri$ with $\sr_1\lll\sr_0$. Then $\rank{\sr_1}$ reduces for the same reasons. Note that it can happen that $\numberof{\sr_1}$ increases by the reduction, but due to the nested multiset ordering  $\rank{\sr_1}$ still decreases. Let us now assume we have $\sr_k\lll\sr_{k-1}\lll\cdots\lll\sr_1\lll\sr_0$ for some $k\ge0$. It follows by induction on $k$ that for each $i$, the rank $\rank{\sr_i}$ decreases (using the same argument as for $\rank{\sr_1}$).
    Then it follows that $\rank\hderi<\rank\deri$ (as the ranks of the other seed-rules are not affected).
  \end{itemize}
  Since the order $<$ on nested multisets is well-founded, it follows that the process of eliminating upmost  $\loopr$-rule instances via Construction~\ref{con:singleloop} does terminate.
  Consequently, we eventually reach a loop-free derivation in $\labISnoloop$.
\end{proof}

\begin{theorem}
  \label{thm:step2}
  If Algorithm~\ref{alg:proofsearch} terminates in Step~\ref{l:ax}, then
  formula~$\fm F$ is a theorem of\/~$\ISfour$.
\end{theorem}

\begin{proof}
  By Lemmas~\ref{lem:saturation} and~\ref{lem:lifting}, there is a derivation in $\labISloops$ of the sequent $\futs rr, \accs rr,\labelsw rF$. By Lemma~\ref{lem:shrinkelim} and Theorem~\ref{thm:unfold}, there is also a proof $\deri$ in $\labISnoloop$ of that sequent. Since no rule in $\labISnoloop$ creates a cluster, all sequents occurring in $\deri$ are proper. Hence, $\deri$ can be translated into a derivation $\deri'$ in $\labISfs$, from which we can (by adding one instance of $\Rref$ and one instance of $\Lref$) obtain a proof in $\labISfs$ of $\labelsw rF$. Hence, by Theorem~\ref{thm:labIKs}, $\fm F$ is a theorem of $\ISfour$.
\end{proof}

\begin{theorem}
  $\ISfour$ is decidable and has the finite model property.
\end{theorem}

\begin{proof}
  By Theorem~\ref{thm:termination}, our proof search algorithm terminates for every formula $\fm F$, and we either obtain a proof of $\fm F$ (by Theorem~\ref{thm:step2}) or a countermodel for $\fm F$ (by Theorem~\ref{thm:step4}). Furthermore, the model $\mmodelof{\sq G_i^\star}$ constructed in the proof of Theorem~\ref{thm:step4} is finite.
\end{proof}

\section{The case of IK4}
\label{sec:IK4}

The logic $\IKfour$ is obtained from $\IK$ by adding the $\vax$-axiom, as shown in~\eqref{eq:vax}, i.e., it is $\ISfour$ without the $\tax$-axiom. Then, Theorem~\ref{thm:plotkin} has to be adjusted by removing the reflexivity of $\rrel$:

\begin{theorem}[Completeness~\cite{fischer-servi:84,plotkin:stirling:86}]\label{thm:plotkin-IK4}
  A formula~$\fm A$ is a theorem of\/~$\IKfour$ if and only if $\fm A$~is valid in every birelational frame~$\langle W, \rel, \le \rangle$ where $\rel$~is transitive.
\end{theorem}

Following the work of \cite{mar:mor:str:2021} and the discussion in Section~\ref{sec:proof:system}, it is easy to see that the proof system in Figure~\ref{fig:labIKp} for $\ISfour$ becomes a proof system for $\IKfour$ if we remove the rule $\Rref$. Let us call the resulting proof system~$\labIKfs$.

It is now rather straightforward to adapt our decision procedure for $\ISfour$, presented in Section~\ref{sec:macro-rules}, to the case of $\IKfour$. Most importantly, the macro rules $\bdiam$ and $\liftrbox$ and $\liftrim$ have to change, as the reflexive $\rrel$-atom has to be dropped for every newly created label. In the case of $\bdiam$ in~\eqref{eq:bdiam}, this would be the relation atom $\accs y y$ in the premise. For the same reason, we have to start the proof search with the sequent $\futs rr,\labelsw rF$ instead of $\futs rr, \accs rr,\labelsw rF$.

Another important point is that the suricata label of a layer must never be in a loop. For that reason, in a large part of this paper we make a distinction between singleton clusters and non-singleton clusters. In the case of $\IKfour$, this distinction is more subtle, as a singleton cluster $\set{\lb z}$, can be with or without reflexive loop $\accs z z$. Only if this reflexive loop is absent, we treat the label $\lb z$ as singleton cluster. If $\accs z z$ is in the sequent, then $\set{\lb z}$ is treated as non-singleton cluster. This concerns in particular the macro rule $\liftrim$ in Section~\ref{sec:lifting}. More precisely, if we have $\labelsw{x_0}{ A \IMP B}$ and $\accs{x_0} {x_0}$ is \emph{not} in the sequent, then we use Construction~\ref{def:imp-lifting-singleton}. If $\accs{x_0} {x_0}$ is present in the sequent, then we use Construction~\ref{def:imp-lifting}. And whenever we use Construction~\ref{def:imp-lifting}, we have to drop the relational atom $\accs{\hx_0}{\hx_0}$ in the last clause.

The same applies to the loop conditions for the $\leafloopr$-, $\boxloopr$-, $\imploopr$-, and $\triangleloopr$-rules. We already stated (see Remark~\ref{rem:suricata-in-loop}) that we must not have $\accsq {x'}{s_y}\srel G\lb y$, where $\lb{s_y}$ is the suricata of the topmost layer containing $\lb{x'}$ and $\lb y$. We also have to ensure that $\lb x\in\interval xy$ and $\lb y\notin\interval xy$. For $\IKfour$ we must additionally demand that $\lb{x'}\neq\lb{s_y}$ and $\lb{y}\neq\lb{s_y}$, which automatically follows for $\ISfour$ because of reflexivity. 

Then, termination (Section~\ref{sec:termination}) and completeness (Section~\ref{sec:countermodel}) and soundness (Section~\ref{sec:soundness}) of the modified algorithm follows in the same way as for $\ISfour$. We therefore have:

\begin{theorem}
  $\IKfour$ is decidable and has the finite model property.
\end{theorem}

\section{Conclusions}
\label{sec:conclusions}

We introduced a decision algorithm for logics $\ISfour$ and $\IKfour$, based on a labelled calculus for the logics, from~\cite{mar:mor:str:2021}. As usual, invertibility of all the rules guarantee that a countermodel can be `read off' one sequent on which proof search fails. 
Usual proof search algorithms, however, might include complex procedures to extract a countermodel from such a failed sequent. 
The main novelty of our approach is the introduction of loop rules, which introduce clusters into our sequents. These rules, which are not, in general, sound, allow to `shrink' a layer, by encoding some information that is safe to repeat within a cluster, while maintaining  the desirable structural properties of the sequent, such as structural saturation. 
This ensures that a countermodel can be immediately extracted from a failed sequent. 
However, some post-possessing needs to be done to transform a proof with loop rules (and clusters) into a proof without loop rules (and clusters). The process of loop reduction does exactly this --- even though the size of the resulting proof can explode, we can safely eliminate the unsound rules from the proof. 

In future work, we wish to further refine and streamline the presentation of this draft, and submit our contribution to a journal. We also plan to formalise our proof in a proof assistant, such as Rocq or Lean. 



	\bibliographystyle{plain}
	\bibliography{references}

\end{document}